\documentclass[12pt, oneside]{amsart}
\pdfoutput=1
\usepackage{geometry}
\usepackage{footmisc}
\usepackage{lmodern}
\allowdisplaybreaks
\usepackage{amsmath}
\usepackage{amssymb}
\usepackage{amsthm}
\usepackage{color}
\usepackage[usenames,dvipsnames]{xcolor}
\usepackage{graphicx}
\usepackage{threeparttable}
\usepackage{subfig}
\usepackage{soul}
\usepackage[countmax]{subfloat}
\usepackage{float}
\usepackage{rotfloat}
\usepackage{comment}
\usepackage{setspace}
\usepackage{esint}
\usepackage{booktabs}
\usepackage{tikz}
\usepackage{multicol}
\usepackage{lscape}
\usepackage[authoryear]{natbib}
\definecolor{MyDarkBlue}{rgb}{0,0.08,0.45}
\usepackage[unicode=true,pdfusetitle, bookmarks=true,bookmarksnumbered=false,bookmarksopen=false, breaklinks=false,pdfborder={0 0 1},backref=section,colorlinks=true, linkcolor = MyDarkBlue, citecolor = MyDarkBlue]{hyperref}
\usepackage{breakurl}

\usepackage{makecell}

\PassOptionsToPackage{normalem}{ulem}
\usepackage{ulem}

\makeatletter

\numberwithin{equation}{section}

\newtheorem{theorem}{{\bf\sc Theorem}}
\newtheorem{lemma}{{\bf\sc Lemma}}

\newtheorem{corollary}{{\bf\sc Corollary}}

\newtheorem{proposition}{{\bf\sc Proposition}}

\newtheorem{lem}{{\bf \sc Lemma}}[section]

\newtheorem{prop}{{\bf\sc Proposition}}[section]

\DeclareMathOperator*{\argmin}{argmin}
\DeclareMathOperator*{\argmax}{argmax}
\DeclareMathOperator*{\cvgto}{\stackrel{\mathrm{P}}{\longrightarrow}}
\DeclareMathOperator*{\wkto}{\rightsquigarrow}
\providecommand{\E}{\mathrm{E}}
\providecommand{\var}{\mathrm{var}}

\providecommand{\cov}{\mathrm{cov}}

\providecommand{\Prob}{\mathrm{P}}

\providecommand{\sign}{{\rm sign}}
\providecommand{\Q}{{\overline{q}}}
\renewcommand{\Pr}{\Prob}

\renewcommand{\bar}{\overline}

\renewenvironment{proof}[1][Proof]{\noindent\text{#1.} }{\ \rule{0.5em}{0.5em}}

\providecommand{\cvgto}{\longrightarrow^{\rm P}}
\providecommand{\wkto}{\rightsquigarrow}

\providecommand{\betaRML}{\hat \beta_R^{\mathrm{ML}}}
\providecommand{\betaRMD}{\hat \beta_R^{\mathrm{MD}}}

\providecommand{\betaRMDL}{\hat \beta_{R,L}^{\mathrm{MD}}}

\providecommand{\betaRMDT}{\hat \beta_{R,T}^{\mathrm{MD}}}

\providecommand{\betanDC}{\breve{\beta}_n^{\mathrm{DC}}}
\providecommand{\betanMD}{\breve{\beta}_n^{\mathrm{MD}}}
\providecommand{\bnDC}{\breve{b}_n^{\mathrm{DC}}}
\providecommand{\bnMD}{\breve{b}_n^{\mathrm{MD}}}
\providecommand{\betanDCell}{\breve{\beta}_{n,\ell}^{\mathrm{DC}}}
\providecommand{\betanMDell}{\breve{\beta}_{n,\ell}^{\mathrm{MD}}}

\providecommand{\bnMDell}{\breve{b}_{n,\ell}^{\mathrm{MD}}}

\makeatother

\begin{document}

\title{A Network Formation Model Based on Subgraphs}

\author{Arun G. Chandrasekhar$^{\ddagger}$}
\author{Matthew O. Jackson$^{\star}$ }
\date{Revision: \today}

\thanks{This grew out of a paper: ``Tractable and Consistent Random Graph Models,'' (\url{http://arxiv.org/abs/1210.7375}), which we have now split
into two pieces.  This part contains the material on subgraph generation models and includes new results on identification, asymptotic normality, and estimation via minimum distance
that were not part of the original paper.
We thank Alberto Abadie, Isaiah Andrews, Emily Breza, Aureo de Paula,
Paul Goldsmith-Pinkham,
Bryan Graham, Han Hong, Guido Imbens, Michael Leung, Shane Lubold, Elena Manresa, Tyler McCormick,
Angelo Mele, Joe Romano, Elie Tamer,  and the referees, as well as seminar participants, for helpful comments and suggestions.  We thank Shreya Chaturvedi,  Vasu Chaudhary, Shobitha Cherian,  Andres Drenik, Anoop Singh Rawat, and Meghna Yadav for valuable research assistance.
Chandrasekhar is grateful for support from the NSF Graduate Research Fellowship Program,
NSF grant  SES-1156182, and the Alfred P. Sloan Foundation.
Jackson gratefully acknowledges financial support from the NSF under grants SES-1629446  and  SES-2018554 and from grant FA9550-12-1-0411
from the AFOSR and DARPA, and ARO MURI award No. W911NF-12-1-0509. }
\thanks{$^{\ddagger}$Department of Economics, Stanford University; member of the NBER; member of J-PAL}
\thanks{$^{\star}$Department of Economics, Stanford University; external faculty member of the Santa Fe Institute}

\begin{abstract}
We develop a new class of random graph models for the statistical estimation of network formation---subgraph generated models (SUGMs).
Various subgraphs---e.g., links, triangles, cliques, stars---are generated and their union results in a network. We show that SUGMs are identified and establish the consistency and asymptotic distribution of parameter estimators in empirically relevant cases. We show that a simple four-parameter SUGM matches basic patterns in empirical networks more closely than four standard models (with many more dimensions): (i) stochastic block models; (ii) models with node-level unobserved heterogeneity; (iii) latent space models; (iv) exponential random graphs.
We illustrate the framework's value via several applications using networks from rural India. We study whether network structure helps enforce risk-sharing and  whether cross-caste interactions are more likely to be private. %, and how the introduction of microcredit changes network formation incentives.
We also develop a new central limit theorem for correlated random variables,
which is required to prove our results and is of independent interest.


\end{abstract}

\thispagestyle{empty}
\setcounter{page}{0}

\maketitle

\newpage

\section{Introduction}

Networks of interactions impact many economic behaviors including  insuring one's self (e.g., \citet*{caijs2015}), participating in microfinance (e.g., \citet{banerjee2013diffusion}),
educating one's self (e.g., \citet*{calvoarmengolpz2009,carrellsw2013}),
and engaging in criminal behavior (e.g., \citet*{glaeserss1996,patacchiniz2008}).
Networks of interactions are also essential to understanding financial contagions
(e.g., \citet*{gaik2010,elliottgj2014,acemogluot2015}), as well as world trade
(e.g., \citet{chaney2016}), inter-state war (e.g., \citet*{jacksonn2015,koenigrtz2015}), and a host of other economic phenomena.
As such, the structure that a network takes has profound consequences---changing the possibility of contagions, the decisions that people make, and the beliefs that people hold---and so it essential to understand and estimate network formation.
Moreover, networks are of interest precisely because there are externalities---one agent's behavior impacts the welfare and behaviors of others.\footnote{For detailed discussions see \citet*{jacksonrz2016} and \citet*{jackson2017}.}
This feature means that connections between agents are not independent,  and so appropriate models of network formation must admit correlations in connections.

Despite the importance of network formation, general, flexible, and tractable econometric models for the estimation of network formation are lacking.   This stems from two challenges:  the aforementioned dependence in connections and the fact that many studies involve one (large) network.  Thus, one is often confronted with estimating a model of formation by taking advantage of the large number of connections, but having them all be dependent observations.
Despite the dependence, it is possible that the many relationships in a network still provide rich enough information to  consistently estimate the parameters of a network model and test hypotheses from a single observed network, at least hypothetically.
Here we develop a class of models that admit
correlations in links and also provide practical techniques of estimating the models, showing that they are estimable, even with just a single network.%even if a researcher only has one network, as well as in cases with many networks.

Before describing our model, it is useful to discuss some of the alternative approaches.

\subsection{Alternative Models of Network Formation}

The most basic models are what are known as `stochastic block models', in which links may depend on node characteristics but are (conditionally) independent of each other.
That approach
requires correlation between links
to be well-approximated by observable characteristics, and may not be sufficient for most
applications.\footnote{A variation on this is community detection where nodes are estimated to
belong to certain groups, though this calculation is NP-hard.
See \cite{bickel2011} for a ``non-parametric view'' of network formation, and \cite{jacksons2017} for an approach to estimating the blocks even if they are latent.}
In particular, stochastic block models are not a good option for estimation in
applications in which there is substantial clustering (triangles) or other cliques in the network.
In fact, in Section \ref{sec:Applications} we show that our model (even with only four parameters) models the graph structure of real-world data better than a stochastic block model even when the block model admits a rich set of covariates and unobserved node level heterogeneity (fixed effects) \citep{chatterjee2010random,graham2014empirical}.\footnote{In fact, correlations can be viewed as  driven by
unobserved heterogeneity \citep*{chatterjee2010random}, which  has links be uncorrelated
conditional on all (observed and unobserved) characteristics
(as extended by \cite{graham2014empirical}). See
	also \cite{charbonneau2017multiple} for related work that is a directed networks version of \cite{graham2014empirical}.
	Such models have been  studied in the mathematics and
	statistics literatures
	(e.g., \cite{holland1981exponential,park2004statistical,blitzstein2011sequential}).}
Although there are challenges in taking such models to data,
they are useful if link correlation is not a concern.

Given the importance of clustering and other local network architectures in many applications, a literature spanning several disciplines (sociology, statistics, economics, and computer science) turned to using exponential random graph models---henceforth ``ERGMs''.
ERGMs admit link interdependencies and have become the workhorse models for estimating network formation.\footnote{These grew from work on what were known as
Markov models (e.g., \citet{frank1986markov}) or $p*$ models (e.g., \citet{wasserman1996logit}).  An alternative  is to work with regression models at the link  level, but to allow for dependent error terms, as in the ``MRQAP'' approach (e.g.,  \citet{krackhardt1988mrqap}).
}
However, from the onset of the use of ERGMs, researchers realized that the parameter estimators could be unstable on all except excessively small networks.  It has been shown that
maximum likelihood and Bayesian estimators are not computationally tractable (the required Gibbs sampler will take exponential time to mix) nor consistent for important classes of such models---and in particular for the ERGMs that
include many link dependencies of interest (and neither parameter estimators nor standard errors can be trusted).
For details see \citet*{bhamidi2008mixing,shalizi2012,chandrasekhar-jackson}.\footnote{Recent work has made progress on both
the speed of convergence of estimation algorithms as well as characterizing the
asymptotic distribution of sufficient statistics in some classes of ERGMs that avoid extensive
link dependencies
(see e.g., \cite*{mele2017structuraldense,mele2017structural,mele2017approximate}).}

A set of models that allow for link dependencies and are estimable is the class of models
 based on explicit link formation algorithms (e.g., \citet*{barabasi1999emergence,jacksonwatts2001exist,jackson2007meeting,currarinijp2009,currarinijp2010,christakis2010empirical,bramoulle-etal}).
These models  can be estimated since the algorithms are particular enough so that one can directly derive how parameters in the model translate into aggregate network statistics, such as the degree distribution or homophily levels.   The advantage of such models is that a specific algorithm allows for estimation.
The disadvantage is that the specificity of the algorithms also necessarily results in highly-structured models.
Thus, these approaches are useful in some contexts, but they are not designed, nor intended, for
general statistical testing of a wide variety of network formation models and hypotheses.
For instance, such models cannot generate considerable   triadic closure (where  links correlated across
triples of nodes---so  if two people have a friend in common, are they more likely to be friends with
each other than if link formation were independent).\footnote{The \citet{jackson2007meeting} model does
have a parameter that affects triadic closure, but in that model closure cannot be separated from the shape
of the degree distribution. So, it is best suited for growing random networks where new nodes are born over time.}

Another approach has roots in the spatial statistics  literature.
Such models organize nodes such that pairs can be evaluated in terms of distance, with  linking probabilities decaying in distance. The distance may be latent (unobserved) or in observed characteristic space (such as geography or demographics). Such models have foundations in the mathematics literature on random geometric graphs \citep*{penrose2003rgg}---where nodes are distributed in a latent space according to some Poisson point process and linking is much more likely among proximate nodes---and have been analyzed in the statistics literature in work on latent space models such as in \cite{hoff2002latent}.
Links between distant-enough pairs of nodes are asymptotically independent and such models have been developed in more detail in the econometrics literature (e.g., \cite{boucher2012my,leung2014}).  This approach holds promise for some enormous networks with appropriate spatial structures---in which the graph can almost be decomposed into independent pieces.\footnote{\citet*{mccormick2015latent} merge the insights from the unobserved heterogeneity  and the latent space distance models, and \citet*{breza2017using} evaluate its empirical performance.}
However, there are many applications for which these latent space (and generally spatial) models---particularly the geometry
of the space the nodes---overly dictates and limits the structure of
link correlation.
Using these models may in fact require estimating an
unobserved manifold, which presents its own challenges.\footnote{ Estimating the latent geometry by using the network data to identify the underlying metric signature is the subject of one of the authors' related work in  \citep{lubold2020identifying}.}
Our model dispenses with these problems in a straightforward way, allowing correlations across nodes
but not forcing correlations generated through distances in unobserved or characteristic space.

Finally, there is a large literature on the theory of network formation from a strategic
perspective (for references, see \cite{jackson2005,jackson2008social}).  Since the first writing of
this paper, researchers have started to derive versions of such models that can be taken to data.
One approach builds
 upon the relationship between certain classes of strategic network formation models and potential games; some of which leverage subgraphs, but in a rather different way from us (\cite{butts2009using,mele2017structuraldense,badev2013}).
 Another derives restrictions on parameters of an observed network under the presumption that it is in equilibrium (pairwise stable) (\citet*{de2014identi,sheng2013}). The latter makes the observation that by using pairwise stability restrictions of \cite{jackson1996strategic} on subnetworks, one can partially identify preference parameters in the model, whereas doing so on the full graph can be computationally infeasible.\footnote{For a  recent overview of the recent literature, see \citet{depaula2015}.}${}^{,}$\footnote{In both the potential games and the partial identification literatures, subgraphs play very different roles from their role here.}
Although the progress to date requires strong restrictions on
 how links can enter agent's payoffs, they provide important first steps in deriving implications of the arsenal of strategic network formation models.
 Below, we also provide ways to incorporate strategic formation in SUGMs, thus in part bridging our approach here and the strategic formation approach.

\subsection{Our Subgraph Model Approach}

Our approach is  distinct from all of the above, both in terms of the approach (working with subgraphs as the basic
building blocks) and the technicalities of allowing nontrivial conditional correlations. Our contribution is to develop models of network formation that admit considerable interdependency without spatial restrictions, and still prove consistency and asymptotic normality of parameter estimators.  As part of this, we develop a new central limit theorem for non-trivially correlated random variables that moves away from relying on spatial-style mixing arguments that force decaying dependence in distance.

The paucity of flexible models that are computable and can be used across many applications for hypothesis testing and inference is what motivates our work here. Although our models are simple conceptually, we provide  different applications that illustrate how
such models admit strategic network formation, general covariates, and generate rich network features.

In Section \ref{sec:model} we introduce \emph{subgraph generated models}---henceforth SUGMs.
In these models, various subgraphs (e.g., links, triangles, cliques, and stars) are generated directly.
For instance, students may form friendships with their roommate(s), members of a study group, teammates,
band members, etc.; researchers may form collaborations on writing papers in pairs, or triples, or
quadruples, etc; villagers may form specific bilateral or multilateral agreements independently,
each to sustain some collection of favors between those individuals involved in the agreement.
This results in links and those links are then naturally correlated since they are formed in combinations.
The union of all these subgraphs results in a network. In this section, we also introduce three %four
motivating applications to demonstrate how this model could be used: (i) descriptively modeling network structure, (ii) motives for risk-sharing, and (iii) incentives to link across social boundaries.%, and (iv) changes in the incentive to link due to the introduction of microfinance.

The statistical challenge is that often only the final network is observed:  a survey may ask people to list their friends or acquaintances, or links may be observed on a social platform, or emails or phone calls are observed, etc., but the original formation process is often not observed. Subgraphs may overlap and may incidentally generate new subgraphs: e.g.,  three links may form and result in a triangle.
Thus, the true rate of formation of the subgraphs cannot generally be inferred just by counting their presence in the resulting network.\footnote{The closest work to ours is a \cite{bollobas2011sparse} piece on random graph theory, which looks at percolation processes, giant component structure, and degree distributions in a model where the observed graph is generated by a set of atoms (subgraphs in our language). That paper focuses on a specific rate of arrival of subgraphs (to maintain a sparsity where a core problem we study is ruled out) and is not interested in statistical estimation.}

Despite this, in Section \ref{sec:Identification} we prove that every subgraph generated model is identified.  That is, if we consider a SUGM---a collection of subgraphs that can potentially form together with a set of parameters governing the probabilities of each subgraph forming---then any two distinct set of parameters necessarily has two distinct set of distributions over the set of possible networks. Furthermore, we explore specific cases that are of empirical relevance---for instance, links and triangles models---and demonstrate that not only are the distributions generally distinct, but that  simple statistics (such as the share of links or triangles that form) allow us to identify the parameters of interest.

Next we turn to estimation of the underlying parameters describing subgraph formation rates in Section \ref{sec:Asymptotics}. We show that we can consistently estimate the parameters and
we derive the asymptotic distribution of the estimators so we
can conduct inference.  There are two situations that a researcher may face.

In the first case, the
researcher has access to ``many networks''. This could be because they have collected network data
from numerous schools, many villages, or so on. Here we demonstrate (using standard
results) that the parameters governing the SUGM can be estimated consistently  with maximum likelihood  estimators that are
asymptotically normally distributed. For some
empirically relevant classes of models, we demonstrate that
there are computationally simple, minimum distance estimators which satisfy consistency and asymptotic normality.

The second case is where the researcher has one (or just a few) ``large network''. This could be because they have collected very rich network data with resource constraints in just a few communities, or because they are looking at a single market, or because they are looking at one social media platform, etc. In this case, the  asymptotics are more technically challenging  for two reasons. First, the network cannot be too sparse, as enough subgraphs must form to make estimation possible, nor too dense  because it becomes impossible to distinguish which subgraph likely generated a candidate link. %In random graph theory, sparsity is modeled by considering a triangular array of graphs and letting the probability of a link change with the network size, thereby controlling the expected degree of a node as a function of the population size.
So formally, we have ``rate requirements'' on the parameters governing the probabilities of subgraphs forming, although these turn out to be quite accomodating.  Second,
 existing central limit theorems from the spatial and time-series econometrics literatures do not
 apply to our setting, as we need to allow subgraphs to form on arbitrary groups of nodes, which
 then results in correlation patterns across all links in the network.  We overcome this problem by
 developing a new central limit theorem and use it to characterize when certain classes of SUGMs have
 estimators that are consistent and asymptotically normally
 distributed.\footnote{An interesting consideration for future work is to employ the
 techniques in \cite{bhattacharyya2015subsampling},
 who develop a bootstrapping method to estimate the distribution of empirical counts of different subgraphs in
 enormous networks.}

With the statistical properties established, we turn to our empirical applications in
Section \ref{sec:Applications}. In each application we use the detailed network data we collected in
75 villages in Karnataka, India \citep{banerjee2014gossip}. We begin by comparing SUGMs to
four archetypical models from the literature in terms of how well they model real-world data.
Specifically, we fit each model to the data and then draw from the distribution at the estimated
parameters for each model. We are interested in a variety of economically relevant network features
(none of which are directly used to estimate any of the models). We find that across the board a four
parameter SUGM outperforms a stochastic block model with flexible covariates; a model of unobserved
heterogeneity at the node level as well as rich covariates; a latent space model with unobserved
locations and heterogeneity as well as observed covariates; and an exponential random graph model with rich covariates.
Only the SUGM comes close to capturing the average path length, homophily, maximal eigenvalue,
size of the giant component, isolates, and clustering. Having established this, the second example
turns to whether the structure of the networks is consistent with the idea that there are stronger
incentives to have supported relationships for risk sharing links rather than informational
links \citep{jackson2012quilts} and we find evidence consistent with this. The third example
explores whether linking across social boundaries---here links between upper caste and lower
caste (Dalit communities)---is more
likely to form  in private (bilateral) rather than group (triadic) settings and we find exactly this. Together, these examples demonstrate the utility of our general framework.

In Section \ref{sclt} we return to state our Central Limit Theorem, which is  of independent   interest.  We provide covariance conditions that are high-level but also straightforward to interpret, check, and micro-found. We use  a powerful lemma from \cite{stein1986approximate} in our proof. Many CLTs build upon Stein's method,\footnote{For instance, \cite{bolthausen1982} uses a pre-cursor lemma from \cite{stein1972}
to derive CLTs from some mixing conditions. In time-series and spatial econometrics,
a non-exhaustive but illustrative list of papers using \cite{bolthausen1982}
include \citet*{conley1999gmm}, \citet*{jenishp2009}, \citet*{bester2011inference},
among others.} but we allow for much richer dependence---all random variables can have non-zero correlation---which admits the correlations in subgraph counts that can arise due to the incidental generation discussed above; and we also allow for triangular arrays.
We discuss the relationship of our Central Limit Theorem and its proof to precursors in Section \ref{sclt}.

%%%%%%%%%%%%%%%%%%% MODEL %%%%%%%%%%%%%%%%%%%%%%%%%%%%%%%%
\section{A Model of Network Formation via Subgraphs}\label{sec:model}

\subsection{Networks}

$n\geq 3$ is the number of nodes on which a network is formed.  Nodes may have characteristics, such as age, profession, gender, race, caste, etc.,
that we denote by the vector $X_i$ for a generic $i\in \{1,\ldots, n\}$. The $X_i$ have finite support.\footnote{This is a limitation since there are network models that do not require discrete covariates. While continuous variables can be discretized, this is a trade-off.}
As such nodes can be classified by a finite set of observable
types.\footnote{We conjecture that our results extend to allow for continuous
covariates as well, though that requires specifying parametric functions for the probability of subgraphs as a function of covariates and so remains beyond the scope
of this paper.  If one expands the set of covariates, the number of parameters to fit increases.  In the limit, if one allows continuous covariates, one then has to fit functions for every type of subgraph (e.g., probability of a triangle as a function of the covariates of the three nodes).  That is only simplified if one imposes restrictions on those functions (e.g., some form of linearity).}

We denote a network by $g$, the collection
of subsets of $ \{1,\ldots, n\}$ of size 2
that lists the edges or links that are present in its graph.  So, $g=\{\{1,3\}, \{2,5\}\}$ indicates
the network that has links between nodes 1 and 3 and between nodes 2 and 5.  For notational ease, we simply write $g=\{ 13, 25\}$, and  write $ij\in g$ to denote that link $ij$ is present in network $g$.
Our model easily accommodates directed graphs, and all of the definitions below extend directly, in which case instead of pairs of nodes, these would be ordered pairs so that $ij$ and $ji$ would differ.  However, for ease of exposition, most of the examples and discussion refer to the undirected case. $\mathcal{G}^n$ denotes the set of all networks on $n$ nodes (which given our definition of a network above, is the set of all labeled (undirected) graphs on the set $ \{1,\ldots, n\}$).

\subsection{Subgraphs and SUGMs}

In a subgraph generation model, subgraphs are each directly generated, and then the resulting network is the union of all of the links in all of the subgraphs.  Degenerate examples of this are Erdos-Renyi random networks, and the generalization of that model,  stochastic-block models, in which links are formed with probabilities based on nodes' attributes.
The more interesting classes of SUGMs include richer subgraphs, and hence involve dependencies in link formation.
It might be that people of the same caste meet more frequently or are more likely to form a relationship when they do meet, as in a stochastic block model, but it could also be
that  groups of three (or more)  meet and can decide whether to form a triangle, with the meeting probability and decision potentially driven by their castes and/or other characteristics.   The model can then be described by a list of probabilities, one for each type of subgraph, where subgraphs can be based on the subgraph shape as well as the nodes' characteristics.

A SUGM is formally defined as follows.

 \begin{itemize}
\item There are finitely many types of nonempty  subgraphs, indexed by $\ell\in \{1,\ldots, k\}$.
\item Each of the $k$ subgraph types corresponds to a set $(G_\ell)_{\ell\in \{1,\ldots,k\}}$,
where each $G_\ell \subset \mathcal{G}^n$ is a set of possible subgraphs on $m_\ell\leq n$ nodes.
\item For each $\ell$ and pair of subgraphs  $g'\in G_\ell$ and $g''\in G_\ell$ there exists a bijection $\pi$ on $\{1,\ldots,n\}$ such that
$ij\in g'$ if and only if $\pi(i)\pi(j)\in g''$.
\item No subgraph is contained in two different sets:
if $\ell\neq \ell'$  then $g \in G_\ell$ implies that $g \notin G_{\ell'}$.
\end{itemize}

This definition does not admit isolates since we define subgraphs to be nonempty and connected, but isolates are easily admitted with notational complications, and are illustrated in some of our examples below as well as the supplementary appendix.
As an example, in the links and triangles case $\ell\in \{L,T\}$.  Then, for $n=4$ and $\ell=T$, $G_T$ is the set of triangles,  where $m_\ell=3$ and
$G_T=\{  \{12,23,31\}, \{12,24,41\}, \{13,34,41\}, \{23,34,42\}\}$.

Note, however, that the definition does not require that $G_\ell$ contain all triangles.  In examples in which node characteristics matter, different triangles could be categorized into different $G_\ell$s.
In particular, definitions of the subgraph types can have restrictions based on node characteristics,
for instance, requiring that the characteristics $X_i$ and $X_{\pi(i)}$ be the same---e.g.,  $G_\ell$ for some $\ell$ could be the set of
 ``triangles that involve one child and two adult nodes''.
As another example, the set $G_\ell$ for some $\ell$ could be all stars with one central node and four other nodes, and another $G_{\ell'}$ could be all the links that involve people of different castes, and so forth.

A few examples are pictured in Figure \ref{fig:ells}.

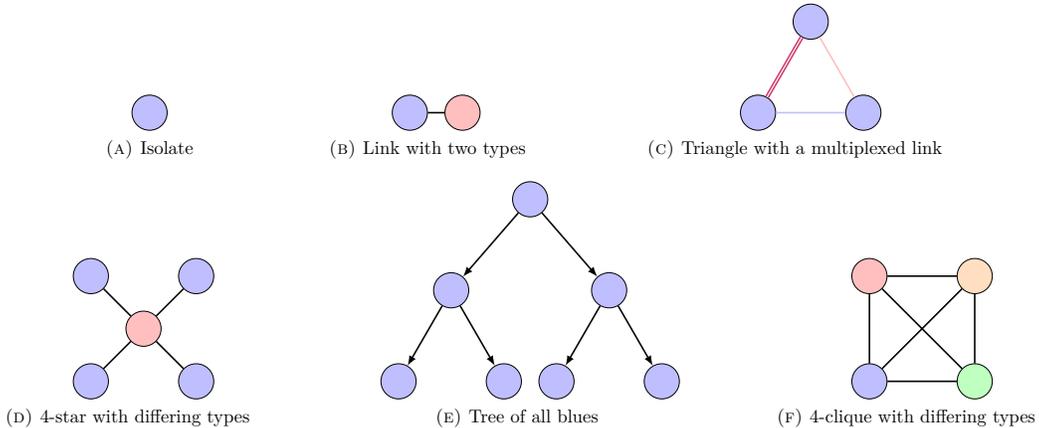
\begin{figure}[h]
\centering
\scalebox{0.7}{\subfloat[Isolate]{
\begin{tikzpicture}
\def \n {5}
\def \radius {2cm}
\def \margin {8}
\node[draw, circle,fill = blue!25,   minimum size=19pt] at (0,0) (v1){$$};
\node[  minimum size=19pt] at (-2,0) (v10){$$};
\node[  minimum size=19pt] at (2,0) (v11){$$};
\end{tikzpicture}
}
\subfloat[Link with two types]{
\begin{tikzpicture}
\def \n {5}
\def \radius {2cm}
\def \margin {8}
\node[draw, circle, fill = blue!25, minimum size=19pt] at (0,0) (v1){$b$};
\node[draw, circle,fill = red!25,  minimum size=19pt] at (1,0) (v2){$r$};
\node[  minimum size=19pt] at (-2,0) (v10){$$};
\node[  minimum size=19pt] at (2.67,0) (v11){$$};
\draw[line width = 0.3mm, >=latex] (v1) to (v2);
\end{tikzpicture}
}
\subfloat[Triangle with a multiplexed link]{
\begin{tikzpicture}
\def \n {5}
\def \radius {2cm}
\def \margin {8}
\node[draw, circle,fill = blue!25,   minimum size=19pt] at (0,0) (v1){$$};
\node[draw, circle, fill = blue!25, minimum size=19pt] at (2,0) (v2){$$};
\node[draw, circle,fill = blue!25,  minimum size=19pt] at (1,1.73) (v3){$$};
\node[  minimum size=19pt] at (-3,0) (v10){$$};
\node[  minimum size=19pt] at (4.4,0) (v11){$$};
\draw[line width = 0.3mm, color = blue!25 >=latex] (v1) to (v2);
\draw[line width = 0.3mm, color = red!25 >=latex] (v2) to (v3);
\draw[line width = 0.3mm, double, color = purple!75 >=latex] (v1) to (v3);
\end{tikzpicture}
}}
\scalebox{0.7}
{\subfloat[4-star with differing types]{
\begin{tikzpicture}
\def \n {5}
\def \radius {2cm}
\def \margin {8}
\node[draw, circle,fill = blue!25,   minimum size=19pt] at (0,0) (v1){$b$};
\node[draw, circle, fill = blue!25, minimum size=19pt] at (2,0) (v2){$b$};
\node[draw, circle,fill = red!25,  minimum size=19pt] at (1,1) (v3){$r$};
\node[draw, circle,fill = blue!25,  minimum size=19pt] at (0,2) (v4){$b$};
\node[draw, circle,fill = blue!25,  minimum size=19pt] at (2,2) (v5){$b$};
\node[  minimum size=19pt] at (-2,0) (v10){$$};
\node[  minimum size=19pt] at (3.4,0) (v11){$$};
\draw[line width = 0.3mm, >=latex] (v1) to (v3);
\draw[line width = 0.3mm, >=latex] (v2) to (v3);
\draw[line width = 0.3mm, >=latex] (v4) to (v3);
\draw[line width = 0.3mm, >=latex] (v5) to (v3);
\end{tikzpicture}
}
\subfloat[Tree of all blues]{
\begin{tikzpicture}
\def \n {5}
\def \radius {2cm}
\def \margin {8}
\node[draw, circle,fill = blue!25,   minimum size=19pt] at (0,0) (v1){$$};
\node[draw, circle, fill = blue!25, minimum size=19pt] at (2,0) (v2){$$};
\node[draw, circle,fill = blue!25,  minimum size=19pt] at (1,1.73) (v3){$$};
\node[draw, circle,fill = blue!25,   minimum size=19pt] at (3,0) (v4){$$};
\node[draw, circle, fill = blue!25, minimum size=19pt] at (5,0) (v5){$$};
\node[draw, circle,fill = blue!25,  minimum size=19pt] at (4,1.73) (v6){$$};
\node[draw, circle,fill = blue!25,  minimum size=19pt] at (2.5,3.46) (v7){$$};
\node[  minimum size=19pt] at (-1.5,0) (v10){$$};
\node[  minimum size=19pt] at (6,0) (v11){$$};
\draw[line width = 0.3mm,<-, >=latex] (v1) to (v3);
\draw[line width = 0.3mm, <-, >=latex] (v2) to (v3);
\draw[line width = 0.3mm,<-,  >=latex] (v4) to (v6);
\draw[line width = 0.3mm, <-, >=latex] (v5) to (v6);
\draw[line width = 0.3mm, <-, >=latex] (v3) to (v7);
\draw[line width = 0.3mm, <-, >=latex] (v6) to (v7);
\end{tikzpicture}
}
\subfloat[4-clique with differing types]{
\begin{tikzpicture}
\def \n {5}
\def \radius {2cm}
\def \margin {8}
\node[draw, circle,fill = blue!25,   minimum size=19pt] at (0,0) (v1){$b$};
\node[draw, circle, fill = green!25, minimum size=19pt] at (2,0) (v2){$g$};
\node[draw, circle,fill = red!25,  minimum size=19pt] at (0,2) (v3){$r$};
\node[draw, circle,fill = orange!25,  minimum size=19pt] at (2,2) (v4){$y$};
\node[  minimum size=19pt] at (-2,0) (v10){$$};
\node[  minimum size=19pt] at (3.4,0) (v11){$$};
\draw[line width = 0.3mm, >=latex] (v1) to (v2);
\draw[line width = 0.3mm, >=latex] (v1) to (v3);
\draw[line width = 0.3mm, >=latex] (v1) to (v4);
\draw[line width = 0.3mm, >=latex] (v2) to (v3);
\draw[line width = 0.3mm, >=latex] (v2) to (v4);
\draw[line width = 0.3mm, >=latex] (v3) to (v4);
\end{tikzpicture}
}}
\caption{\label{fig:ells}Examples of subgraphs. Links could be directed or
undirected or even multiplexed (take on multiple edge types) and nodes can have
different characteristic combinations
(denoted by node colors
and labels).}
\end{figure}

The probability that various subgraphs form is described by a vector of parameters, denoted $\beta =(\beta_1, \ldots\beta_\ell, \ldots, \beta_k)\in \mathcal{B}^k$,
where $\mathcal{B}$ is (unless otherwise noted)
a compact subset of $[0,1)^k$.\footnote{We treat vectors as row or columns as is convenient in what follows.}
For instance, $\beta=(\beta_{L},\beta_{T})\in \mathcal{B}\subset [0,1)^2$ in a links and triangles example.\footnote{In some examples below, we expand this demonstrating how $\beta$ can have entries that are monotone functions of  preference parameters (or equilibrium behavior), which allows us to study certain economic questions. Estimating $\beta$ allows us to either recover the parameters or behavior of interest in some cases or conduct loose hypothesis testing using our estimates of $\beta$.}

A network $g$ on $n$ nodes is randomly formed as follows:
\begin{enumerate}
\item Each of the possible subgraphs $g_\ell\in G_\ell$ forms with probability $\beta_\ell$ independently  of all other subgraphs (including others in $G_\ell$).
\item The resulting network, $g$, is the union of all
the links that appear in any of the generated subgraphs.
\end{enumerate}

\subsection{An Example with Node Characteristics}
Suppose that nodes come in two colors: blue and red (for instance different genders, age groups,
religions, etc., and clearly this extends directly to more than two colors).
In our example of links and triangles, there are now three types of links:  (blue, blue),   (blue, red), (red, red);   and four types of triangles (blue,blue,blue),  (blue,blue,red), (blue,red,red), (red,red,red) which comprise the set of subgraphs indexed by $\ell$.

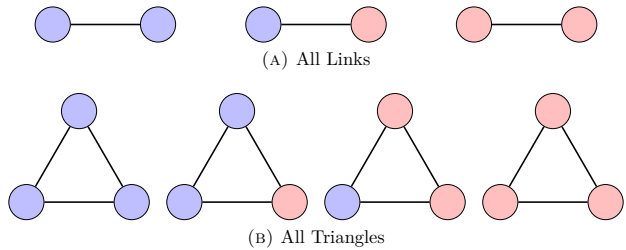
\begin{figure}[h]
\centering
\scalebox{0.6}{\subfloat[All Links]{
\begin{tikzpicture}
\def \n {5}
\def \radius {2cm}
\def \margin {8}
\node[draw, circle,fill = blue!25,   minimum size=19pt] at (0,0) (v1){$b$};
\node[draw, circle,fill = blue!25,    minimum size=19pt] at (2,0) (v2){$b$};
\node[draw, circle,fill = blue!25,    minimum size=19pt] at (4,0) (v3){$b$};
\node[draw, circle,fill = red!25,    minimum size=19pt] at (6,0) (v4){$r$};
\node[draw, circle,fill = red!25,    minimum size=19pt] at (8,0) (v5){$r$};
\node[draw, circle,fill = red!25,    minimum size=19pt] at (10,0) (v6){$r$};
\draw[line width = 0.3mm, >=latex] (v1) to (v2);
\draw[line width = 0.3mm, >=latex] (v3) to (v4);
\draw[line width = 0.3mm, >=latex] (v5) to (v6);
\end{tikzpicture}
}}
\hspace{1cm}
\scalebox{0.6}{\subfloat[All Triangles]{
\begin{tikzpicture}
\def \n {5}
\def \radius {2cm}
\def \margin {8}
\node[draw, circle,fill = blue!25,   minimum size=19pt] at (0,0) (v1){$b$};
\node[draw, circle, fill = blue!25, minimum size=19pt] at (2,0) (v2){$b$};
\node[draw, circle,fill = blue!25,  minimum size=19pt] at (1,1.73) (v3){$b$};
\node[draw, circle,fill = blue!25,   minimum size=19pt] at (3,0) (v4){$b$};
\node[draw, circle, fill = red!25, minimum size=19pt] at (5,0) (v5){$r$};
\node[draw, circle,fill = blue!25,  minimum size=19pt] at (4,1.73) (v6){$b$};
\node[draw, circle,fill = blue!25,   minimum size=19pt] at (6,0) (v7){$b$};
\node[draw, circle, fill = red!25, minimum size=19pt] at (8,0) (v8){$r$};
\node[draw, circle,fill = red!25,  minimum size=19pt] at (7,1.73) (v9){$r$};
\node[draw, circle,fill = red!25,   minimum size=19pt] at (9,0) (v10){$r$};
\node[draw, circle, fill = red!25, minimum size=19pt] at (11,0) (v11){$r$};
\node[draw, circle,fill = red!25,  minimum size=19pt] at (10,1.73) (v12){$r$};
\draw[line width = 0.3mm, >=latex] (v1) to (v2);
\draw[line width = 0.3mm, >=latex] (v2) to (v3);
\draw[line width = 0.3mm, >=latex] (v1) to (v3);
\draw[line width = 0.3mm, >=latex] (v4) to (v5);
\draw[line width = 0.3mm, >=latex] (v4) to (v6);
\draw[line width = 0.3mm, >=latex] (v5) to (v6);
\draw[line width = 0.3mm, >=latex] (v7) to (v8);
\draw[line width = 0.3mm, >=latex] (v7) to (v9);
\draw[line width = 0.3mm, >=latex] (v8) to (v9);
\draw[line width = 0.3mm, >=latex] (v10) to (v11);
\draw[line width = 0.3mm, >=latex] (v10) to (v12);
\draw[line width = 0.3mm, >=latex] (v11) to (v12);
\end{tikzpicture}
}}
\caption{\label{fig:Links-Triangles-Covars}Panel (A) shows all possible links and Panel (B) shows all possible triangles when a node has characteristic $X_i\in \{red,\ blue\}$.}
\end{figure}

Thus, in this example the sets of subgraphs are
$$G_{(blue,blue)}= \{ ij : X_i=blue, X_j=blue \}$$
and
$$G_{(blue,blue,red)}= \{ ijk : X_i=blue, X_j=blue, X_k=red \},$$
and so forth, as depicted in Figure \ref{fig:Links-Triangles-Covars}.
The
parameters
$$ \{ \beta_{(blue, blue)},   \beta_{(blue, red)},   \beta_{(red, red)},   \beta_{(blue,blue,blue)},   \beta_{(blue,blue,red)},   \beta_{(blue,red,red)},   \beta_{(red,red,red)} \},$$
are the probabilities that the corresponding subgraphs form.

One could restrict or enrich the model by having simpler or more complex sets of parameters -- for instance requiring that $\beta_{(blue, blue)}=   \beta_{(red, red)}$, or by having preference parameters
that govern the probabilities of various subgraphs forming, as we discuss below.

\subsection{Links and Triangles as Our Leading Example}

The bulk of our illustrations and applications are based on link and triangle SUGMs, though other subgraphs can be included
and are covered by our general results (e.g., Theorems \ref{newid}, \ref{thm:many_networks}, and \ref{thm:Largenetwork}).
Our illustrations focus on links and triangles for two reasons:
first, this case is simple to understand and illustrates the main points since it exhibits correlated links and incidental generation;
second, the link and triangle model already matches the moments that are of interest in many research
projects (larger cliques are rare). In fact, as we show below, simply looking at a links and
triangle SUGM tagged with whether the nodes involved are homogenous or heterogeneous in
demographics (e.g., just a 4 parameter model), replicates real-world network features better than
far-richer models. Still, we leave further specification to the researcher as it will depend on their
context and the phenomenon being modeled. If there are other the types of subgraphs that are
hypothesized to arise in some particular context, then that model can be constructed and estimated
in the ways we outline and are covered by our general results.\footnote{One could also have a list of subgraphs as a possible basis for the
SUGM with only a subset of them actually forming the true SUGM;
allowing the data to tell the researcher which to include.  Some of that can be done here, including the various subgraphs that might be involved
and then seeing which have nontrivial parameter estimates.   This marries SUGMs with model selection, a topic which could be explored further in future research.}

%%%%%%%%%%%%%%%% IDENTIFICATION %%%%%%%%%%%%%%%%%%%%%%%%%%
\section{Identification\label{sec:Identification}}

\subsection{The Challenge of Identification}

The researcher's goal is to use the observed data---from one or more networks---to recover the parameters of interest,
for example, the $(\beta_{L}, \beta_{T})$ in a SUGM of links and triangles.
If the researcher observed the links  and
triangles that were formed directly,
then estimation would be straightforward.
Indeed, in some instances a researcher has direct information on all the various groups a given individual is involved in: for instance in the case of a co-authorship network,
the researcher may observe all the papers a researcher has written and thus observes papers with two authors, three authors, and so forth.
Instead, for instance, it may be that there are groups of three people who commonly share favors and risks together---who really form a triangle, but the researcher
only has information from a survey asking with which alters a given person interacts
(as in networks derived from the Add Health data set as in \citet{currarinijp2009}),  or who borrows from whom and who lends kerosene and rice to whom and other bilateral nominations (as in our Indian village data \citet{banerjee2013diffusion}), or from observing that who are friends on a social
platform
(as in Facebook network data as in \citet{baileyetal2015,chettyetal2022I}), or from observing that two people phone each other or remit payments to each other (as in \citet{blumenstockef2011}).

Thus, the general problem is that
the formation of the subgraphs is not directly observed, and so must be inferred in order to estimate the parameters of interest.
For example, if three links are observed between $i$, $j$, and $k$,
is it the case that $ijk$ formed as a triangle, or that $ij$, $jk$ and $ik$ formed as links, or that $ij$ and $jk$ formed as links and $ik$ formed as part of a different triangle $ikm$, or some combination of these or other combinations? Figure \ref{fig:ident} provides an illustration.

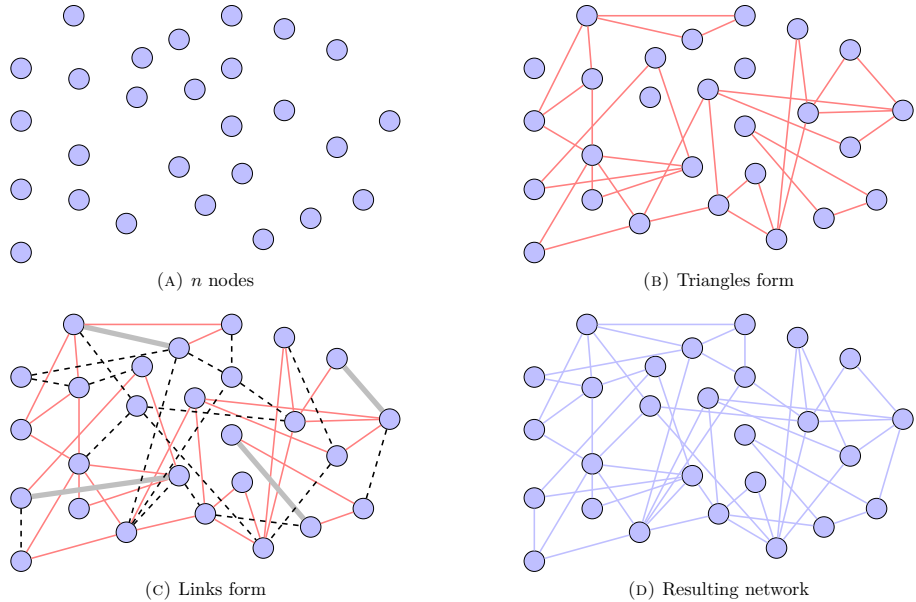
\begin{figure}[!h]
\centering
\scalebox{0.7}{\subfloat[$n$ nodes]{
\begin{tikzpicture}
\def \n {5}
\def \radius {2cm}
\def \margin {8}

\node[draw, circle,fill = blue!25,   minimum size=2pt] at (1,2) (v1){$$};
\node[draw, circle, fill = blue!25, minimum size=2pt] at (3,1.55) (v2){$$};
\node[draw, circle,fill = blue!25,  minimum size=2pt] at (4,2) (v3){$$};
\node[draw, circle,fill = blue!25,   minimum size=2pt] at (5,1.75) (v4){$$};
\node[draw, circle, fill = blue!25, minimum size=2pt] at (6,1.35) (v5){$$};

\node[draw, circle,fill = blue!25,   minimum size=2pt] at (0,1) (v6){$$};
\node[draw, circle, fill = blue!25, minimum size=2pt] at (1.1,0.8) (v7){$$};
\node[draw, circle,fill = blue!25,  minimum size=2pt] at (2.2,0.45) (v8){$$};
\node[draw, circle,fill = blue!25,   minimum size=2pt] at (2.3,1.2) (v9){$$};
\node[draw, circle, fill = blue!25, minimum size=2pt] at (3.3,0.6) (v10){$$};
\node[draw, circle, fill = blue!25, minimum size=2pt] at (4,1) (v11){$$};

\node[draw, circle,fill = blue!25,   minimum size=2pt] at (0,0) (v12){$$};
\node[draw, circle, fill = blue!25, minimum size=2pt] at (1.1,-0.65) (v13){$$};
\node[draw, circle,fill = blue!25,  minimum size=2pt] at (3,-0.875) (v14){$$};
\node[draw, circle,fill = blue!25,   minimum size=2pt] at (4,-0.1) (v15){$$};
\node[draw, circle, fill = blue!25, minimum size=2pt] at (5,0.2) (v16){$$};
\node[draw, circle, fill = blue!25, minimum size=2pt] at (6,-0.5) (v17){$$};
\node[draw, circle, fill = blue!25, minimum size=2pt] at (7,0) (v18){$$};

\node[draw, circle,fill = blue!25,   minimum size=2pt] at (0,-1.3) (v19){$$};
\node[draw, circle, fill = blue!25, minimum size=2pt] at (1.1,-1.5) (v20){$$};
\node[draw, circle,fill = blue!25,  minimum size=2pt] at (2,-1.95) (v21){$$};
\node[draw, circle,fill = blue!25,   minimum size=2pt] at (3.5,-1.6) (v22){$$};
\node[draw, circle, fill = blue!25, minimum size=2pt] at (4.2,-1) (v23){$$};
\node[draw, circle, fill = blue!25, minimum size=2pt] at (5.5,-1.85) (v24){$$};
\node[draw, circle, fill = blue!25, minimum size=2pt] at (6.5,-1.5) (v25){$$};

\node[draw, circle, fill = blue!25, minimum size=2pt] at (0,-2.5) (v26){$$};
\node[draw, circle, fill = blue!25, minimum size=2pt] at (4.6,-2.25) (v27){$$};

\node[ minimum size=2pt] at (-1,0) (v28){$$};
\node[ minimum size=2pt] at (8,0) (v29){$$};

\end{tikzpicture}
}}
\scalebox{0.7}{\subfloat[Triangles form]{
\begin{tikzpicture}
\def \n {5}
\def \radius {2cm}
\def \margin {8}

\node[draw, circle,fill = blue!25,   minimum size=2pt] at (1,2) (v1){$$};
\node[draw, circle, fill = blue!25, minimum size=2pt] at (3,1.55) (v2){$$};
\node[draw, circle,fill = blue!25,  minimum size=2pt] at (4,2) (v3){$$};
\node[draw, circle,fill = blue!25,   minimum size=2pt] at (5,1.75) (v4){$$};
\node[draw, circle, fill = blue!25, minimum size=2pt] at (6,1.35) (v5){$$};

\node[draw, circle,fill = blue!25,   minimum size=2pt] at (0,1) (v6){$$};
\node[draw, circle, fill = blue!25, minimum size=2pt] at (1.1,0.8) (v7){$$};
\node[draw, circle,fill = blue!25,  minimum size=2pt] at (2.2,0.45) (v8){$$};
\node[draw, circle,fill = blue!25,   minimum size=2pt] at (2.3,1.2) (v9){$$};
\node[draw, circle, fill = blue!25, minimum size=2pt] at (3.3,0.6) (v10){$$};
\node[draw, circle, fill = blue!25, minimum size=2pt] at (4,1) (v11){$$};

\node[draw, circle,fill = blue!25,   minimum size=2pt] at (0,0) (v12){$$};
\node[draw, circle, fill = blue!25, minimum size=2pt] at (1.1,-0.65) (v13){$$};
\node[draw, circle,fill = blue!25,  minimum size=2pt] at (3,-0.875) (v14){$$};
\node[draw, circle,fill = blue!25,   minimum size=2pt] at (4,-0.1) (v15){$$};
\node[draw, circle, fill = blue!25, minimum size=2pt] at (5.2,0.15) (v16){$$};
\node[draw, circle, fill = blue!25, minimum size=2pt] at (6,-0.5) (v17){$$};
\node[draw, circle, fill = blue!25, minimum size=2pt] at (7,0.2) (v18){$$};

\node[draw, circle,fill = blue!25,   minimum size=2pt] at (0,-1.3) (v19){$$};
\node[draw, circle, fill = blue!25, minimum size=2pt] at (1.1,-1.5) (v20){$$};
\node[draw, circle,fill = blue!25,  minimum size=2pt] at (2,-1.95) (v21){$$};
\node[draw, circle,fill = blue!25,   minimum size=2pt] at (3.5,-1.6) (v22){$$};
\node[draw, circle, fill = blue!25, minimum size=2pt] at (4.2,-1) (v23){$$};
\node[draw, circle, fill = blue!25, minimum size=2pt] at (5.5,-1.85) (v24){$$};
\node[draw, circle, fill = blue!25, minimum size=2pt] at (6.5,-1.5) (v25){$$};

\node[draw, circle, fill = blue!25, minimum size=2pt] at (0,-2.5) (v26){$$};
\node[draw, circle, fill = blue!25, minimum size=2pt] at (4.6,-2.25) (v27){$$};

\node[ minimum size=2pt] at (-1,0) (v28){$$};
\node[ minimum size=2pt] at (8,0) (v29){$$};

\draw[line width = 0.3mm, color=red!50, >=latex] (v1) to (v2);
\draw[line width = 0.3mm, color=red!50, >=latex] (v1) to (v3);
\draw[line width = 0.3mm, color=red!50, >=latex] (v2) to (v3);

\draw[line width = 0.3mm, color=red!50, >=latex] (v1) to (v12);
\draw[line width = 0.3mm, color=red!50, >=latex] (v1) to (v7);
\draw[line width = 0.3mm, color=red!50, >=latex] (v12) to (v7);

\draw[line width = 0.3mm, color=red!50, >=latex] (v13) to (v12);
\draw[line width = 0.3mm, color=red!50, >=latex] (v13) to (v7);

\draw[line width = 0.3mm, color=red!50, >=latex] (v5) to (v16);
\draw[line width = 0.3mm, color=red!50, >=latex] (v5) to (v18);
\draw[line width = 0.3mm, color=red!50, >=latex] (v16) to (v18);

\draw[line width = 0.3mm, color=red!50, >=latex] (v4) to (v16);
\draw[line width = 0.3mm, color=red!50, >=latex] (v4) to (v27);
\draw[line width = 0.3mm, color=red!50, >=latex] (v16) to (v27);

\draw[line width = 0.3mm, color=red!50, >=latex] (v10) to (v18);
\draw[line width = 0.3mm, color=red!50, >=latex] (v10) to (v17);
\draw[line width = 0.3mm, color=red!50, >=latex] (v17) to (v18);

\draw[line width = 0.3mm, color=red!50, >=latex] (v15) to (v24);
\draw[line width = 0.3mm, color=red!50, >=latex] (v24) to (v25);
\draw[line width = 0.3mm, color=red!50, >=latex] (v15) to (v25);

\draw[line width = 0.3mm, color=red!50, >=latex] (v22) to (v23);
\draw[line width = 0.3mm, color=red!50, >=latex] (v23) to (v27);
\draw[line width = 0.3mm, color=red!50, >=latex] (v22) to (v27);

\draw[line width = 0.3mm, color=red!50, >=latex] (v10) to (v22);
\draw[line width = 0.3mm, color=red!50, >=latex] (v10) to (v21);
\draw[line width = 0.3mm, color=red!50, >=latex] (v22) to (v21);

\draw[line width = 0.3mm, color=red!50, >=latex] (v21) to (v13);
\draw[line width = 0.3mm, color=red!50, >=latex] (v13) to (v26);
\draw[line width = 0.3mm, color=red!50, >=latex] (v21) to (v26);

\draw[line width = 0.3mm, color=red!50, >=latex] (v20) to (v13);
\draw[line width = 0.3mm, color=red!50, >=latex] (v13) to (v14);
\draw[line width = 0.3mm, color=red!50, >=latex] (v20) to (v14);

\draw[line width = 0.3mm, color=red!50, >=latex] (v9) to (v19);
\draw[line width = 0.3mm, color=red!50, >=latex] (v9) to (v14);
\draw[line width = 0.3mm, color=red!50, >=latex] (v19) to (v14);

\end{tikzpicture}
}}

\centering
\scalebox{0.7}{\subfloat[Links form]{
\begin{tikzpicture}
\def \n {5}
\def \radius {2cm}
\def \margin {8}

\node[draw, circle,fill = blue!25,   minimum size=2pt] at (1,2) (v1){$$};
\node[draw, circle, fill = blue!25, minimum size=2pt] at (3,1.55) (v2){$$};
\node[draw, circle,fill = blue!25,  minimum size=2pt] at (4,2) (v3){$$};
\node[draw, circle,fill = blue!25,   minimum size=2pt] at (5,1.75) (v4){$$};
\node[draw, circle, fill = blue!25, minimum size=2pt] at (6,1.35) (v5){$$};

\node[draw, circle,fill = blue!25,   minimum size=2pt] at (0,1) (v6){$$};
\node[draw, circle, fill = blue!25, minimum size=2pt] at (1.1,0.8) (v7){$$};
\node[draw, circle,fill = blue!25,  minimum size=2pt] at (2.2,0.45) (v8){$$};
\node[draw, circle,fill = blue!25,   minimum size=2pt] at (2.3,1.2) (v9){$$};
\node[draw, circle, fill = blue!25, minimum size=2pt] at (3.3,0.6) (v10){$$};
\node[draw, circle, fill = blue!25, minimum size=2pt] at (4,1) (v11){$$};

\node[draw, circle,fill = blue!25,   minimum size=2pt] at (0,0) (v12){$$};
\node[draw, circle, fill = blue!25, minimum size=2pt] at (1.1,-0.65) (v13){$$};
\node[draw, circle,fill = blue!25,  minimum size=2pt] at (3,-0.875) (v14){$$};
\node[draw, circle,fill = blue!25,   minimum size=2pt] at (4,-0.1) (v15){$$};
\node[draw, circle, fill = blue!25, minimum size=2pt] at (5.2,0.15) (v16){$$};
\node[draw, circle, fill = blue!25, minimum size=2pt] at (6,-0.5) (v17){$$};
\node[draw, circle, fill = blue!25, minimum size=2pt] at (7,0.2) (v18){$$};

\node[draw, circle,fill = blue!25,   minimum size=2pt] at (0,-1.3) (v19){$$};
\node[draw, circle, fill = blue!25, minimum size=2pt] at (1.1,-1.5) (v20){$$};
\node[draw, circle,fill = blue!25,  minimum size=2pt] at (2,-1.95) (v21){$$};
\node[draw, circle,fill = blue!25,   minimum size=2pt] at (3.5,-1.6) (v22){$$};
\node[draw, circle, fill = blue!25, minimum size=2pt] at (4.2,-1) (v23){$$};
\node[draw, circle, fill = blue!25, minimum size=2pt] at (5.5,-1.85) (v24){$$};
\node[draw, circle, fill = blue!25, minimum size=2pt] at (6.5,-1.5) (v25){$$};

\node[draw, circle, fill = blue!25, minimum size=2pt] at (0,-2.5) (v26){$$};
\node[draw, circle, fill = blue!25, minimum size=2pt] at (4.6,-2.25) (v27){$$};

\node[ minimum size=2pt] at (-1,0) (v28){$$};
\node[ minimum size=2pt] at (8,0) (v29){$$};

\draw[line width = 1mm, color=gray!50, >=latex] (v1) to (v2);
\draw[line width = 0.3mm, color=red!50, >=latex] (v1) to (v3);
\draw[line width = 0.3mm, color=red!50, >=latex] (v2) to (v3);

\draw[line width = 0.3mm, color=red!50, >=latex] (v1) to (v12);
\draw[line width = 0.3mm, color=red!50, >=latex] (v1) to (v7);
\draw[line width = 0.3mm, color=red!50, >=latex] (v12) to (v7);

\draw[line width = 0.3mm, color=red!50, >=latex] (v13) to (v12);
\draw[line width = 0.3mm, color=red!50, >=latex] (v13) to (v7);

\draw[line width = 0.3mm, color=red!50, >=latex] (v5) to (v16);
\draw[line width = 1mm, color=gray!50, >=latex] (v5) to (v18);
\draw[line width = 0.3mm, color=red!50, >=latex] (v16) to (v18);

\draw[line width = 0.3mm, color=red!50, >=latex] (v4) to (v16);
\draw[line width = 0.3mm, color=red!50, >=latex] (v4) to (v27);
\draw[line width = 0.3mm, color=red!50, >=latex] (v16) to (v27);

\draw[line width = 0.3mm, color=red!50, >=latex] (v10) to (v18);
\draw[line width = 0.3mm, color=red!50, >=latex] (v10) to (v17);
\draw[line width = 0.3mm, color=red!50, >=latex] (v17) to (v18);

\draw[line width = 1mm, color=gray!50, >=latex] (v15) to (v24);
\draw[line width = 0.3mm, color=red!50, >=latex] (v24) to (v25);
\draw[line width = 0.3mm, color=red!50, >=latex] (v15) to (v25);

\draw[line width = 0.3mm, color=red!50, >=latex] (v22) to (v23);
\draw[line width = 0.3mm, color=red!50, >=latex] (v23) to (v27);
\draw[line width = 0.3mm, color=red!50, >=latex] (v22) to (v27);

\draw[line width = 0.3mm, color=red!50, >=latex] (v10) to (v22);
\draw[line width = 0.3mm, color=red!50, >=latex] (v10) to (v21);
\draw[line width = 0.3mm, color=red!50, >=latex] (v22) to (v21);

\draw[line width = 0.3mm, color=red!50, >=latex] (v21) to (v13);
\draw[line width = 0.3mm, color=red!50, >=latex] (v13) to (v26);
\draw[line width = 0.3mm, color=red!50, >=latex] (v21) to (v26);

\draw[line width = 0.3mm, color=red!50, >=latex] (v20) to (v13);
\draw[line width = 0.3mm, color=red!50, >=latex] (v13) to (v14);
\draw[line width = 0.3mm, color=red!50, >=latex] (v20) to (v14);

\draw[line width = 0.3mm, color=red!50, >=latex] (v9) to (v19);
\draw[line width = 0.3mm, color=red!50, >=latex] (v9) to (v14);
\draw[line width = 1mm, color=gray!50, >=latex] (v19) to (v14);

\draw[line width = 0.3mm, dashed, >=latex] (v2) to (v11);
\draw[line width = 0.3mm, dashed, >=latex] (v3) to (v11);
\draw[line width = 0.3mm, dashed, >=latex] (v4) to (v17);
\draw[line width = 0.3mm, dashed, >=latex] (v6) to (v2);
\draw[line width = 0.3mm, dashed, >=latex] (v6) to (v7);
\draw[line width = 0.3mm, dashed, >=latex] (v7) to (v9);
\draw[line width = 0.3mm, dashed, >=latex] (v1) to (v8);
\draw[line width = 0.3mm, dashed, >=latex] (v2) to (v21);
\draw[line width = 0.3mm, dashed, >=latex] (v11) to (v16);
\draw[line width = 0.3mm, dashed, >=latex] (v8) to (v16);
\draw[line width = 0.3mm, dashed, >=latex] (v8) to (v13);
\draw[line width = 0.3mm, dashed, >=latex] (v8) to (v27);
\draw[line width = 0.3mm, dashed, >=latex] (v19) to (v26);
\draw[line width = 0.3mm, dashed, >=latex] (v17) to (v27);
\draw[line width = 0.3mm, dashed, >=latex] (v18) to (v25);
\draw[line width = 0.3mm, dashed, >=latex] (v22) to (v24);
\draw[line width = 0.3mm, dashed, >=latex] (v14) to (v22);
\draw[line width = 0.3mm, dashed, >=latex] (v21) to (v14);
\draw[line width = 0.3mm, dashed, >=latex] (v21) to (v11);

\end{tikzpicture}
}}
\scalebox{0.7}{\subfloat[Resulting network]{
\begin{tikzpicture}
\def \n {5}
\def \radius {2cm}
\def \margin {8}

\node[draw, circle,fill = blue!25,   minimum size=2pt] at (1,2) (v1){$$};
\node[draw, circle, fill = blue!25, minimum size=2pt] at (3,1.55) (v2){$$};
\node[draw, circle,fill = blue!25,  minimum size=2pt] at (4,2) (v3){$$};
\node[draw, circle,fill = blue!25,   minimum size=2pt] at (5,1.75) (v4){$$};
\node[draw, circle, fill = blue!25, minimum size=2pt] at (6,1.35) (v5){$$};

\node[draw, circle,fill = blue!25,   minimum size=2pt] at (0,1) (v6){$$};
\node[draw, circle, fill = blue!25, minimum size=2pt] at (1.1,0.8) (v7){$$};
\node[draw, circle,fill = blue!25,  minimum size=2pt] at (2.2,0.45) (v8){$$};
\node[draw, circle,fill = blue!25,   minimum size=2pt] at (2.3,1.2) (v9){$$};
\node[draw, circle, fill = blue!25, minimum size=2pt] at (3.3,0.6) (v10){$$};
\node[draw, circle, fill = blue!25, minimum size=2pt] at (4,1) (v11){$$};

\node[draw, circle,fill = blue!25,   minimum size=2pt] at (0,0) (v12){$$};
\node[draw, circle, fill = blue!25, minimum size=2pt] at (1.1,-0.65) (v13){$$};
\node[draw, circle,fill = blue!25,  minimum size=2pt] at (3,-0.875) (v14){$$};
\node[draw, circle,fill = blue!25,   minimum size=2pt] at (4,-0.1) (v15){$$};
\node[draw, circle, fill = blue!25, minimum size=2pt] at (5.2,0.15) (v16){$$};
\node[draw, circle, fill = blue!25, minimum size=2pt] at (6,-0.5) (v17){$$};
\node[draw, circle, fill = blue!25, minimum size=2pt] at (7,0.2) (v18){$$};

\node[draw, circle,fill = blue!25,   minimum size=2pt] at (0,-1.3) (v19){$$};
\node[draw, circle, fill = blue!25, minimum size=2pt] at (1.1,-1.5) (v20){$$};
\node[draw, circle,fill = blue!25,  minimum size=2pt] at (2,-1.95) (v21){$$};
\node[draw, circle,fill = blue!25,   minimum size=2pt] at (3.5,-1.6) (v22){$$};
\node[draw, circle, fill = blue!25, minimum size=2pt] at (4.2,-1) (v23){$$};
\node[draw, circle, fill = blue!25, minimum size=2pt] at (5.5,-1.85) (v24){$$};
\node[draw, circle, fill = blue!25, minimum size=2pt] at (6.5,-1.5) (v25){$$};

\node[draw, circle, fill = blue!25, minimum size=2pt] at (0,-2.5) (v26){$$};
\node[draw, circle, fill = blue!25, minimum size=2pt] at (4.6,-2.25) (v27){$$};

\node[ minimum size=2pt] at (-1,0) (v28){$$};
\node[ minimum size=2pt] at (8,0) (v29){$$};

\draw[line width = 0.3mm, color=blue!25, >=latex] (v1) to (v2);
\draw[line width = 0.3mm, color=blue!25,  >=latex] (v1) to (v3);
\draw[line width = 0.3mm, color=blue!25,  >=latex] (v2) to (v3);

\draw[line width = 0.3mm, color=blue!25,  >=latex] (v1) to (v12);
\draw[line width = 0.3mm, color=blue!25,  >=latex] (v1) to (v7);
\draw[line width = 0.3mm, color=blue!25,  >=latex] (v12) to (v7);

\draw[line width = 0.3mm, color=blue!25,  >=latex] (v13) to (v12);
\draw[line width = 0.3mm, color=blue!25,  >=latex] (v13) to (v7);

\draw[line width = 0.3mm, color=blue!25,  >=latex] (v5) to (v16);
\draw[line width = 0.3mm, color=blue!25, >=latex] (v5) to (v18);
\draw[line width = 0.3mm, color=blue!25,  >=latex] (v16) to (v18);

\draw[line width = 0.3mm, color=blue!25,  >=latex] (v4) to (v16);
\draw[line width = 0.3mm, color=blue!25,  >=latex] (v4) to (v27);
\draw[line width = 0.3mm, color=blue!25,  >=latex] (v16) to (v27);

\draw[line width = 0.3mm, color=blue!25,  >=latex] (v10) to (v18);
\draw[line width = 0.3mm, color=blue!25,  >=latex] (v10) to (v17);
\draw[line width = 0.3mm, color=blue!25,  >=latex] (v17) to (v18);

\draw[line width = 0.3mm, color=blue!25, >=latex] (v15) to (v24);
\draw[line width = 0.3mm, color=blue!25,  >=latex] (v24) to (v25);
\draw[line width = 0.3mm, color=blue!25,  >=latex] (v15) to (v25);

\draw[line width = 0.3mm, color=blue!25,  >=latex] (v22) to (v23);
\draw[line width = 0.3mm, color=blue!25,  >=latex] (v23) to (v27);
\draw[line width = 0.3mm, color=blue!25,  >=latex] (v22) to (v27);

\draw[line width = 0.3mm, color=blue!25,  >=latex] (v10) to (v22);
\draw[line width = 0.3mm, color=blue!25,  >=latex] (v10) to (v21);
\draw[line width = 0.3mm, color=blue!25,  >=latex] (v22) to (v21);

\draw[line width = 0.3mm, color=blue!25,  >=latex] (v21) to (v13);
\draw[line width = 0.3mm, color=blue!25,  >=latex] (v13) to (v26);
\draw[line width = 0.3mm, color=blue!25,  >=latex] (v21) to (v26);

\draw[line width = 0.3mm, color=blue!25,  >=latex] (v20) to (v13);
\draw[line width = 0.3mm, color=blue!25,  >=latex] (v13) to (v14);
\draw[line width = 0.3mm, color=blue!25,  >=latex] (v20) to (v14);

\draw[line width = 0.3mm, color=blue!25,  >=latex] (v9) to (v19);
\draw[line width = 0.3mm, color=blue!25,  >=latex] (v9) to (v14);
\draw[line width = 0.3mm, color=blue!25, >=latex] (v19) to (v14);

\draw[line width = 0.3mm, color=blue!25,  >=latex] (v2) to (v11);
\draw[line width = 0.3mm, color=blue!25,  >=latex] (v3) to (v11);
\draw[line width = 0.3mm, color=blue!25,  >=latex] (v4) to (v17);
\draw[line width = 0.3mm, color=blue!25,  >=latex] (v6) to (v2);
\draw[line width = 0.3mm, color=blue!25,  >=latex] (v6) to (v7);
\draw[line width = 0.3mm, color=blue!25,  >=latex] (v7) to (v9);
\draw[line width = 0.3mm, color=blue!25,  >=latex] (v1) to (v8);
\draw[line width = 0.3mm, color=blue!25,  >=latex] (v2) to (v21);
\draw[line width = 0.3mm, color=blue!25,  >=latex] (v11) to (v16);
\draw[line width = 0.3mm, color=blue!25,  >=latex] (v8) to (v16);
\draw[line width = 0.3mm, color=blue!25,  >=latex] (v8) to (v13);
\draw[line width = 0.3mm, color=blue!25,  >=latex] (v8) to (v27);
\draw[line width = 0.3mm, color=blue!25,  >=latex] (v19) to (v26);
\draw[line width = 0.3mm, color=blue!25,  >=latex] (v17) to (v27);
\draw[line width = 0.3mm, color=blue!25,  >=latex] (v18) to (v25);
\draw[line width = 0.3mm, color=blue!25,  >=latex] (v22) to (v24);
\draw[line width = 0.3mm, color=blue!25,  >=latex] (v14) to (v22);
\draw[line width = 0.3mm, color=blue!25,  >=latex] (v21) to (v14);
\draw[line width = 0.3mm, color=blue!25,  >=latex] (v21) to (v11);

\end{tikzpicture}
}}

\caption{\label{fig:ident} The network that is formed and eventually observed is shown in panel D.
The process comes from forming triangles with probability $\beta_{T}$ as in (B) in red; and forming links, in grey, with probability $\beta_{L}$ as in (C)---all independently.  New links are dashed while
 links that overlap with some link also formed in a triangle are in solid and bold.  We see that there is
both (i) overlap  as some links coincide with links already in triangles, as well as (ii) extra triangles that were generated ``incidentally.''
Given that we only observe the resulting network in panel D, we need to infer the formation of the different subgraphs carefully and not simply by directly counting observed links and triangles.}
\end{figure}

The overlap and incidental generation present a challenge for estimating a parameter related to triangle formation since some of the observed triangles were
\emph{``directly generated''} in the formation process,
and others were {\sl ``incidentally generated;''} and similarly, it presents a challenge to estimating a parameter for link formation since some truly generated links end up as parts of triangles.
We show that despite this difficulty, the parameters can be recovered by careful study of the observed patterns. In particular, we show that a SUGM is {\sl always} identified, and also provide techniques for recovering the parameters.

\subsection{A General Identification Result}

We first show that as the parameters of any SUGM change, so does the distribution over networks, and hence SUGMs are identified models.

Let $\Pr_{\beta}$ denote the probability distribution over a network $g$ on $n$ nodes under a vector of parameters ${\beta}$ describing the probabilities of subgraph types $(G_\ell)_{\ell\in \{1,\ldots,k\}}$.

\begin{theorem}\label{newid}
Every SUGM is identified.   That is, for any finite collection of distinct types of subgraphs $(G_\ell)_{\ell\in \{1,\ldots,k\}}$ on $n$ nodes,
\(
\beta\neq\beta'\implies \Pr_{\beta}
\neq \Pr_{\beta'}.
\)
\end{theorem}

Recalling the general definition of the SUGM, this means that
for every SUGM (even one comprised of subgraphs that could have nodes with varying (discrete)
covariates and allowing for multiplexing, etc.) is identified.

To understand why this holds, for instance in the case of links and triangles, note that as one varies $(\beta_{L},\beta_{T})$, the {\sl relative} rates of overall observed links and triangles change, as do the number of triangles that overlap with each other.
One can calculate the relative rates at which incidental links and triangles are expected to be generated, and there is an invertible relationship between observed counts of links and triangles, and the underlying
rates at which they were expected to be directly formed.  Theorem \ref{newid} shows that this is true not only for links and triangles, but for any
collection of distinct subgraphs.

We emphasize, of course, that identification does not imply that the parameters are easily estimated,
especially on a very small number of nodes.
We provide results on consistency below, which require observation of a sufficiently large network and/or sufficiently many networks.

\subsubsection{Identification from Link and Triangle Counts}

Although  Theorem \ref{newid} shows that SUGMs are always identified---i.e., distinct
parameters yield distinct distributions---it is often convenient to use minimum distance based
estimators  based on simple moments of the network.  Thus, it is useful to show that identification can be achieved from simple statistics.
We illustrate that this can be done with direct counts of the relative frequency of appearances of the subgraphs.  In particular, in Proposition \ref{prop:Links-Triangles-SUGM} we show that a links and triangles SUGM can be identified directly from the counts of links and triangles:  $S(g) = (S_L(g), S_T(g))$.
This does not mean that one can ignore incidental generation, but it does mean that the information one has to use can be simple counts.

Further below, in Theorem \ref{thm:Largenetwork}, we show conditions under which such direct counts not only identify the parameters
for general subgraphs, but can also be used to derive consistent and normally distributed estimators of the parameters.

To understand the identification, consider Figure \ref{fig:id}.
Each configuration involves two triangles, but the graph in  Panel B with
only five links is {\sl relatively} more easily incidentally formed than the one in Panel A.
Thus, by looking at the combination of how many triangles and how likely links there are, we can sort out relative rates of the two parameters.

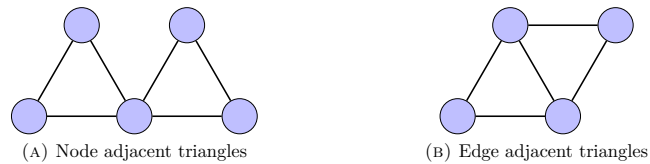
\begin{figure}[h]
\centering
\scalebox{0.7}{\subfloat[Node adjacent triangles]{
\begin{tikzpicture}
\def \n {5}
\def \radius {2cm}
\def \margin {8}
\node[draw, circle, fill = blue!25, minimum size=19pt] at (-2,0) (v1){$$};
\node[draw, circle,fill = blue!25,  minimum size=19pt] at (-1,1.73) (v2){$$};
\node[draw, circle,fill = blue!25,   minimum size=19pt] at (0,0) (v3){$$};
\node[draw, circle,fill = blue!25,  minimum size=19pt] at (1,1.73) (v4){$$};
\node[draw, circle,fill = blue!25,   minimum size=19pt] at (2,0) (v5){$$};
\node[ minimum size=19pt] at (-4,0) (v6){$$};
\node[ minimum size=19pt] at (4,0) (v7){$$};
\draw[line width = 0.3mm, >=latex] (v1) to (v2);
\draw[line width = 0.3mm,  >=latex] (v1) to (v3);
\draw[line width = 0.3mm,  >=latex] (v2) to (v3);
\draw[line width = 0.3mm,  >=latex] (v3) to (v4);
\draw[line width = 0.3mm,  >=latex] (v3) to (v5);
\draw[line width = 0.3mm,  >=latex] (v4) to (v5);
\end{tikzpicture}
}}
\scalebox{0.7}{\subfloat[Edge adjacent triangles]{
\begin{tikzpicture}
\def \n {5}
\def \radius {2cm}
\def \margin {8}
\node[draw, circle, fill = blue!25, minimum size=19pt] at (-2,0) (v1){$$};
\node[draw, circle,fill = blue!25,  minimum size=19pt] at (-1,1.73) (v2){$$};
\node[draw, circle,fill = blue!25,   minimum size=19pt] at (0,0) (v3){$$};
\node[draw, circle,fill = blue!25,  minimum size=19pt] at (1,1.73) (v4){$$};
\node[ minimum size=19pt] at (-3,0) (v5){$$};
\node[ minimum size=19pt] at (2,0) (v6){$$};
\draw[line width = 0.3mm, >=latex] (v1) to (v2);
\draw[line width = 0.3mm,  >=latex] (v1) to (v3);
\draw[line width = 0.3mm,  >=latex] (v2) to (v3);
\draw[line width = 0.3mm,  >=latex] (v2) to (v4);
\draw[line width = 0.3mm,  >=latex] (v3) to (v4);
\end{tikzpicture}
}}
\caption{\label{fig:id} Two different configurations of two triangles; one has a count of 6 total links and the other has a count of 5 links.  (A) is more relatively more likely to come directly from the formation of two triangles, and (B) is relatively more likely to come from a combination of links and triangles.  The likelihoods of links and triangles can thus be deduced via careful deductions from the combination of the counts of links and triangles.}
\end{figure}

\begin{proposition}
\label{prop:Links-Triangles-SUGM} A SUGM of links and triangles is
identified from moments  $S(g) = (S_L(g), S_T(g))$ for any $\beta=\left(\beta_{L},\beta_{T}\right)\in\left[0,1\right)^{2}$.
That is, if $\left(\beta_{L}',\beta_{T}'\right)\neq\left(\beta_{L},\beta_{T}\right)$
then $\E_{\beta'}\left[S\left(g\right)\right]\neq\E_{\beta}\left[S\left(g\right)\right]$.
\end{proposition}

Let us outline the basic ideas behind the proof, with the full proof appearing in the appendix.
Let $\widetilde{q}_L\left(\beta_{L},\beta_{T}\right)$ denote the probability that any given link forms conditional upon exactly one particular triangle that it could be a part of not forming, which depends on the $\beta$s.  For instance, for nodes $ij$ it is the probability that $ij$ is formed
either as a link or as part of a triangle that is {\sl not} triangle $hij$ for some other node $h$.   Although this is not an immediately obvious parameter to define, it allows us to write the probability that a given link forms as $\beta_T + (1-\beta_T) \widetilde{q}_L\left(\beta_{L},\beta_{T}\right)$.
This expression turns to be useful as it helps us to compare the rate at which links form to the rate at which triangles form in a way that shows how they are identified.
In particular:
\begin{equation}
\label{eslt}
\E_{\beta_L,\beta_T}\left[ S_{L}(g), S_{T}(g)  \right] =  \left[  \beta_T + (1-\beta_T) \widetilde{q}_L\left(\beta_{L},\beta_{T}\right)    ,  \beta_T + (1-\beta_T) \left(\widetilde{q}_L\left(\beta_{L},\beta_{T}\right)\right)^3 \right].
\end{equation}
For instance, note that the term $\beta_T + (1-\beta_T) \left(\widetilde{q}_L\left(\beta_{L},\beta_{T}\right)\right)^3$ is the probability that a triangle forms, either directly ($\beta_T$), or does not form
directly $(1-\beta_T)$ but then each of the links then forms on its own $\left(\widetilde{q}_L\left(\beta_{L},\beta_{T}\right)\right)^3$.\footnote{Conditional upon the triangle
not forming directly, the links are then independent.}
  This is helpful in showing how different parameters lead to different rates of formation of links and triangles since we can isolate the difference via the $\widetilde{q}_L\left(\beta_{L},\beta_{T}\right)$ versus $\left(\widetilde{q}_L\left(\beta_{L},\beta_{T}\right)\right)^3$
  expressions.

Analogs of this proposition extend to cases with covariates and multiplexing, simply with more complicated extensions of (\ref{eslt}) accounting for the specific types of triangles or links.
Also, a general version of asymptotic identification is a by-product of Theorem \ref{thm:Largenetwork}, below.

%%%%%%%%%%%% ASYMPTOTICS %%%%%%%%%%%%%%%%%%%%%%%%%
\section{Estimation and Asymptotics}
\label{sec:Asymptotics}

We now provide conditions under which various estimators of the parameters are consistent and describe their asymptotic distributions.  We consider two asymptotic frames, in which at least one of either the size of the network or the number of networks becomes large enough for consistent estimation. We discuss two different estimators for each frame for a total of four estimators.

\subsection{Data and Asymptotic Frames}
Suppose that the researcher observes $R\geq 1$ independently, and identically drawn graphs $(g_1,\ldots,g_R)$, on at least $n$ nodes each, drawn from a SUGM with a list of $k$ subgraphs and parameters $\beta \in [0,1)^k$.
Each of the $k$ subgraphs involves no more than $n$ nodes.
For simplicity in notation,
we work with each network having exactly $n$ nodes, but one can directly extend the results by simply selecting $n$ nodes for each network and applying all of
our estimation to those subgraphs.

The first asymptotic frame, studied in Section \ref{sec:many-networks}, covers situations in which the number of different realizations of networks $R$ tends to infinity. Here researchers have access to many networks and the empirical moments of interest converge to their expectations via observation of independent networks. This applies when a researcher is studying, for instance a number of schools, classrooms, villages, etc.
In this case estimation and inference is straightforward.
There are a growing number of independent draws from the distribution and we have
already proven identification in Theorem \ref{newid}.  Our Theorem \ref{thm:many_networks} shows that the maximum likelihood estimator from $R$ networks---which we denote by $\betaRML$---is consistent and asymptotically normally distributed as $R$ grows.

Given the difficulty in calculating the likelihoods for networks, also consider a second computationally-simpler minimum-distance estimator (presented for the case of links and triangles), denoted by $\betaRMD$.  We show in Proposition \ref{prop:many_networks_LT} that this minimum distance estimator is consistent and asymptotically normally distributed.

The second asymptotic frame is studied in Section \ref{sec:single-large} and it holds the number of networks observed $R$ fixed, without loss of generality at $R = 1$, and then lets the number of nodes grow: $n \rightarrow \infty$.
Examples include when the researcher has detailed information about a large community, friendships on social media platform, citation networks, etc.  Clearly, this extends to cases with large $n$ and more than one network, but we consider $R=1$ for ease of notation.
This is the more challenging perspective as the observations of
various parts of a network are not independent.
Also, the identification result from Theorem \ref{newid} does not
guarantee that the empirical moments converge to
their expectations in a single large network.

There are two cases of interest with a single large network. The first is what we call the sparse case (which we explicitly characterize), and this is a situation in which certain types of incidental generation of subgraphs become asymptotically negligible.
For the sparse case, we prove that identification and asymptotic consistency and normality is possible from an easy variation on direct
counts of observed subgraphs. Namely, one begins with the largest subgraph in the model, count how many of them are present, then remove links associated with them and step down to the next largest and so on. The estimator corresponding to this procedure is what we call a direct count estimator---denoted by $ \betanDC$---as it is essentially directly calculating the linking rate for each subgraph type.
We prove the
 consistency and asymptotic normality of the direct count estimator under suitable sparsity conditions in Theorem \ref{thm:Largenetwork}.

It is possible to verify whether a network is sparse enough to permit the direct estimator in the following way. One can take relevant parameter values for the SUGM (which one can find by a first crude estimation from the data) and then generate a network with those parameter values and then check to see if the direct estimators recover these parameters. If there is too much incidental generation, then the parameters will not be recovered and then
our fourth estimator is needed, as is our new central limit theorem.

In particular, Theorem \ref{thm:Largenetwork} requires a level of sparsity that makes certain kinds of incidental generations rare. For denser graphs (which can still be sparse, but permitting nontrivial incidental generation) we work with a minimum distance estimator that matches the moments of the shares of the subgraphs---which we denote by $\betanMD$. In Proposition \ref{prop:LT_SUGM1-2}, we show  the consistency and asymptotic normality of this minimum distance estimator.  We focus on the links and triangles model since the calculations are idiosyncratic based on the specific SUGM the researcher wants to employ, but the logic extends. The proof  of asymptotic normality in this case of potentially dense SUGMs requires using our new central limit theorem for correlated random variables, Theorem \ref{clt}, which is the focus of Section \ref{sclt}.

 Appendix D% \ref{sec:sim_consistency}
 provides simulations verifying consistency, asymptotic normality, and convergence. We also show how $\betanDC$ and $\betanMD$ both perform well when incidental generation is sufficiently small but that as the networks become denser $\betanDC$ is biased while $\betanMD$ is consistent.

We let $\beta$ possibly depend on $n$ and/or $R$ as described below. We take the list of the types of subgraphs to be analyzed to be fixed.

\subsection{The Many Networks Case}\label{sec:many-networks}

We keep the presentation of this first frame brief since it follows standard statistical arguments (e.g., \cite{newey1994large}).

One has a collection of $R$ networks, each drawn independently according to the same SUGM with the same parameter vector $\beta_0$. We hold the set of nodes $\{1,\ldots,n\}$ (and their covariates) fixed. Theorem \ref{thm:many_networks} states that a maximum likelihood estimator of the parameters is consistent and asymptotically normally distributed.

\begin{theorem}\label{thm:many_networks}
	Consider a SUGM of $k$ distinct types of subgraphs with $\beta_0 \in \text{int}(\mathcal{B})$,
	for $\mathcal{B}$ a  compact subset of $[0,1)^k$. Let $g_r$ for $r=1,\ldots,R$
	denote i.i.d. draws from this distribution. Let $\betaRML$ denote the maximum likelihood estimator
	\(
 \betaRML= \argmax_{\beta \in \mathcal{B}} \frac{1}{R}\sum_r \log \Pr_{\beta}(g_r) .
	\)
	 Then
	\(
	\betaRML
 \cvgto\beta_{0}.
	\)
	If in addition $J :=  \E[\nabla_{\beta} \log \Pr_{\beta_0}(g_r)\nabla_{\beta} \log \Pr_{\beta_0}(g_r)']$ is non-singular, then
\(
	\sqrt{R}\left(\betaRML-\beta_{0}\right)\rightsquigarrow\mathcal{N}\left(0,J^{-1}\right).
	\)
\end{theorem}

Although Theorem \ref{thm:many_networks} demonstrates that a consistent and asymptotically normally distributed estimator exists, calculating the likelihood function of arbitrary
networks as a function of the parameters can be computationally intensive for large networks.
Thus, we also present a result on a minimum distance estimator which is computationally straightforward since it simply involves calculating frequencies of certain subgraphs.  We present it
based on links and triangles as the typical case that researchers will need, but the technique extends as a researcher requires.
As before, let $S_L(g)$ and $S_T(g)$ denote the fraction of links and triangles in the network $g$, with $S = (S_L, S_T)'$.

Let
$$h(g_r,\beta)= S(g_r)-\E_{\beta}\left[S\left(g_r\right)\right], $$
be a moment function comparing observed subgraph statistics to expected ones for a given $\beta$.

Let $\betaRMD$ denote the minimum distance estimator,
	\[
        \betaRMD := \argmin_{\beta \in \mathcal{B}} \left(
	\frac{1}{R}\sum_r h(g_r,\beta)
	\right)^\prime \left(
	\frac{1}{R}\sum_r h(g_r,\beta)
	\right) .
	\]

\begin{proposition}\label{prop:many_networks_LT}
	Consider a SUGM of links and triangles with parameters $\beta_0 \in \text{int}(\mathcal{B})$, a compact subset of $[0,1)^2$. Let $g_r$ for $r=1,\ldots,R$ denote i.i.d. draws from this distribution.
	Then,
	\[
	 \betaRMD \cvgto\beta_{0}
	\text{ and }
	\sqrt{R}\left(\betaRMD-\beta_{0}\right)\rightsquigarrow\mathcal{N}\left(0,(H^\prime \Omega^{-1} H)^{-1}\right)
	\]
	where $H :=  \E [\nabla_{\beta} h(g_r,\beta_0)]$ and $\Omega= \E \left[ h(g_r,\beta_0) h(g_r,\beta_0)^\prime\right]$.
\end{proposition}

\subsection{The Large Network Case}\label{sec:single-large}

Next we turn to the case where researchers have access to at least one large network (modeled as  $n \rightarrow \infty$).   For the exposition, we let $R=1$, but clearly this extends
directly to having observations of more than one network.

This case is considerably more challenging as it involves correlated observations generated within a network.
Network data tend to be sparse, but still have local patterns such as clustering, so that people have relatively few connections compared to the potential number of links, but where a person's neighbors tend to be linked to each other
with much higher than an independent probability (e.g., see the background in \cite{newman2003,jackson2008social}).
Such clustering is the challenging aspect of the asymptotics since subgraphs
are not only directly generated but also
incidentally generated.  Thus,
we need new techniques for our asymptotic results.

\subsubsection{Sequences of Large Random Networks}

To describe how parameter estimators behave as a function of the number of nodes
$n$, it is useful to allow the parameters to also be indexed by $n$. This approach is standard in the random graphs literature (e.g., see the classic book of \cite{bollobas2001random}) as it is needed to accommodate most applications.
Specifically, research on social networks has long observed that parameters need to adjust with the number of nodes.
For example, friendship networks among a small set of agents (say 50 or 100) and large set of agents (thousands or much more)
often have comparable average degrees.\footnote{See \cite{chandrasekhar2015econometrics} for examples networks of varying size ranging from village network data in sub-saharan Africa or India to university dorm friendship network data which all exhibit somewhat comparable number of links per node.}
As a concrete example, consider friendships among high school students in the U.S. based on the Add Health data set
(e.g., see \citep*{currarinijp2009,currarinijp2010}).
There are some high schools with only 30 students and others with around 3000 students.
The average degree ranges between 6 and 8 over the high schools, but this means that the link probability must shrink dramatically with $n$:  average degree $d$ corresponds to a link probability of roughly $d/30$ in the small schools, but only $d/3000$ in the large schools.
Thus, irrespective of the size of their school, students have numbers of friends of the same order of magnitude; and so the true underlying parameters describing friendship formation must decrease with $n$ to match the data.

Thus, we consider
a sequence of SUGMs with subgraphs $(G_1^n,...,G_k^n)$ that form on $n$ nodes that are generated with probabilities $\beta^n = (\beta_1^n,...,\beta_k^n)$.
The superscript on the $\beta^n$ indicates the dependence on $n$ to allow for true subgraph formation rates to vary along the sequence.

\subsubsection{Direct-Count Estimators for Negligible Incidental Generation}\label{sec:no-incidental-estimation}

It is convenient to express the $\beta_\ell^n$s in the form
\[
\beta_\ell^n = \frac{b_\ell}{n^{h_\ell}}
\]
for some $b_\ell>0$ and $ h_\ell > 0$ fixed in $n$.
This allows us to encode the rates at which the parameters vary with $n$,
and is a general way of encoding the rates that could come from meeting, time budgets, costs, or any other constraints that gives rise to sparse networks.

We consider the case in which $m_\ell> h_\ell$ (where recall that $m_\ell$ is the number of nodes in the subgraph of type $\ell$ and is fixed along the sequence),
as otherwise the expected number of subgraphs in the whole network could be bounded as $n$ grows, precluding estimation.

The researcher can make assumptions on $h_\ell$, either its value or the possible range of values that are admissible for their model. In fact, the magnitude may be straightforward to observe with simple subgraph counts. For example, if across networks of varying size, one sees some growing function of links, triangles, and so on, one can infer what values of $h_\ell$ are needed to be consistent with data. In fact, in most models of network formation, such assumptions are implicitly made, knowingly or not.

We show that even without knowing the $b_\ell$ or $h_\ell$,
the parameters $\beta_\ell^n$ can be well-estimated, provided the
network model is not so sparse that subgraphs are never observed, nor so dense so that they scale linearly in $n$.

To develop the estimator, first we need some definitions and notation.

Consider a SUGM and order the classes of the subgraphs, $G_1^n,\ldots, G_\ell^n,\ldots,G_k^n$, from `largest' to `smallest'.
In particular, pick an ordering of $1,\ldots,k$ so that a subgraph in $G^n_\ell$ cannot be a
subnetwork of the subnetworks in $G^n_{\ell'}$ for $k\geq \ell'>\ell \geq 1$:
$$g_\ell\in G^n_{\ell} \ {\rm and } \ g_{\ell'} \in G^n_{\ell'}\ \ {\rm implies\ that} \ \ g_\ell \not\subset g_{\ell'}.$$
There exists at least one such ordering - for instance, any ordering in which subgraphs with more links are counted before subgraphs with fewer links.  In an example with links, 2-stars and triangles:  triangles precede 2-stars which precede links.
Note that this is a partial order: for instance, a `three link line' $ij, jk, kl$ is neither a subgraph nor a supergraph of a `3-star' $ij,ik,il$, which is also a three link subgraph on four nodes.
It is irrelevant in which order subgraphs with the same number of links are counted.

We count subgraphs in this order after having removed links associated with all of the subgraphs already counted.  The resulting counts are
denoted $\widetilde{S}^n_\ell$:\footnote{Note that counting in order from `largest' to `smallest' subnetworks means that
we count things from smallest to largest index $\ell$ since the specification of how we ordered labels moves in the opposite direction of the size of the subgraphs.}
\[
\widetilde{S}_{\ell}^n(g) = | \{ g_{\ell} \in G^n_{\ell} :  g_{\ell}\subset g \ {\rm and \ }  g_{\ell}\not\cap g_{\ell''} {\rm \ for \ any \ }
g_{\ell''} \in G^n_{\ell''} {\rm \ such \ that \ } g_{\ell''}\subset g {\rm \ for \ some \ } {\ell''<\ell}  \} |.
\]

The logic of this is that incidental generation is more often in one direction than another: a triangle incidentally generates three links, while it can be much rarer that three links happen to independently form to make a triangle.
This manner of counting motivates a simple estimator that we call the \emph{direct-count} estimator. We then divide by the number of
possible subgraphs of that variety.

For the direct count estimator $\betanDC$ we presume that, for each $\ell$, $G_\ell^n$ includes all subgraphs that are relabellings of each other.
Thus, we work without demographics on the subgraphs, but these counts can easily be adjusted accordingly by normalizing by $|G_\ell^n|$.
Let $\kappa_{\ell}^n$ denote the (finite number) of relabelings to count different subgraphs in $G_\ell^n$ on a given set of $m_\ell$ nodes.\footnote{For example, note that $\kappa_\ell ^n= 1$ for a triangle but for a $K$-star it is $K$ since each star is different when a different member of the $K$ nodes is the center.}
Then $\kappa_{\ell}^n \binom{n}{m_\ell}$ is the number of possible subgraphs of type $\ell$.

The direct-count estimator $\betanDC$  is
\begin{equation}\label{phat}
\betanDCell = \frac{\widetilde{S}^n_\ell(g)}{\kappa_{\ell} \binom{n}{m_\ell}}.
\end{equation}
As we prove in Section \ref{sec:no-incidental-estimation}, under suitable conditions,  incidental generation is  low and the direct estimators are consistent estimators of the true parameters and are asymptotically normally distributed.

As an illustration, consider Figure \ref{fig:multi_panel} in which links and triangles are formed on 41 nodes.
There are 9 truly generated triangles, but 10 observed overall.   So,
the frequency of triangles, $\widetilde{S}^n_T(g)$, is overestimated by using 10 instead of 9.
The true frequency was $9/10660$ but is estimated as $10/10660$.
With respect to links, there were actually 25 truly directly generated, but one becomes part of an incidentally generated triangle
and two others overlap on existing triangles,
and so $\widetilde{S}^n_L(g)$ becomes 22 instead.
So we estimate $22/820$ while the true frequency was $25/820$.
These errors are already small on a network on just 41 nodes, and as we prove below, the errors disappear completely as $n$ grows.

\begin{figure}[h]
\begin{center}
\subfloat[$n$ nodes]{
\label{fig:image_1}
\includegraphics[width=0.251\textwidth]{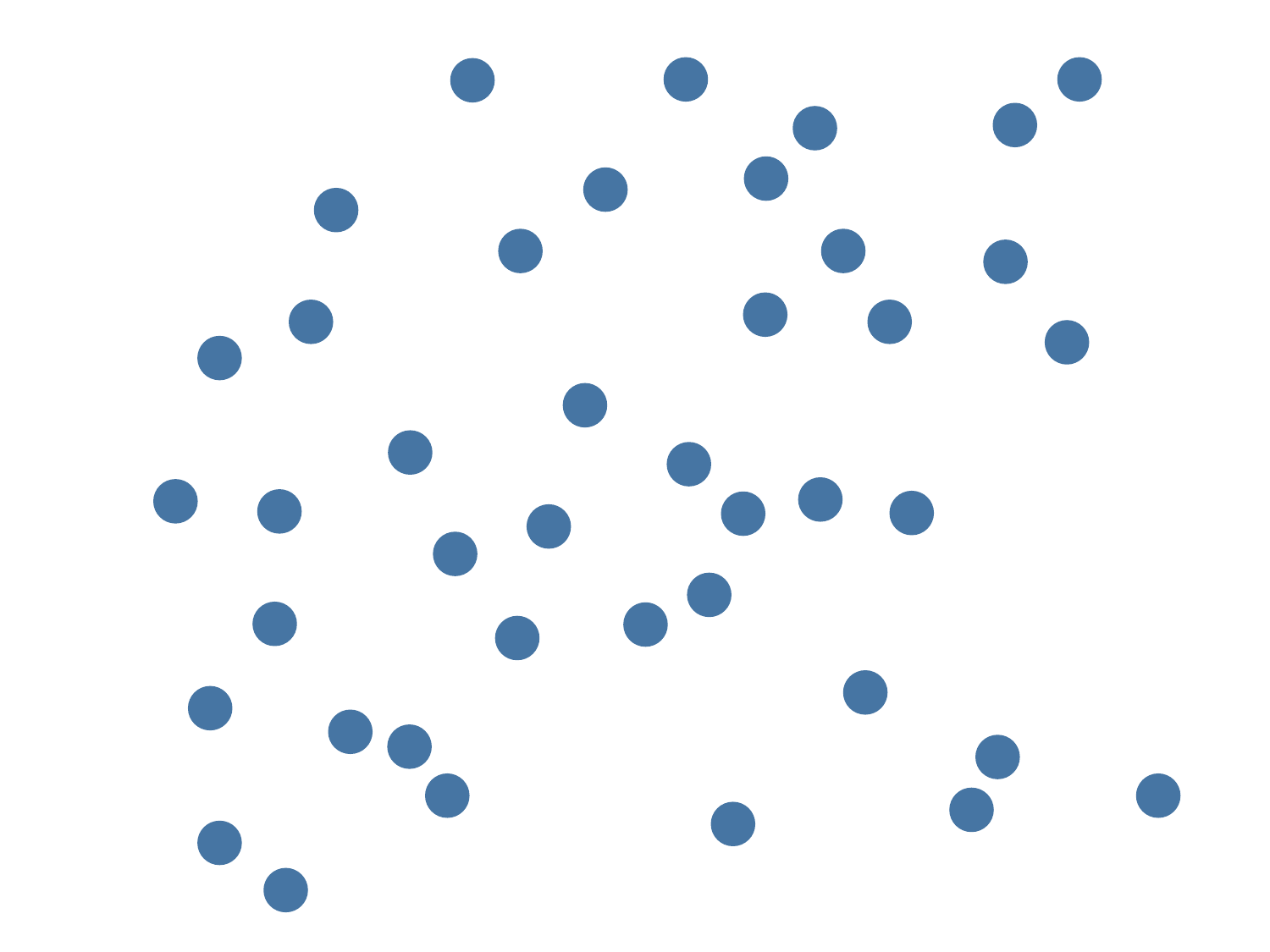}
}
\subfloat[Triangles form]{
\label{fig:image_2}
\includegraphics[width=0.25\textwidth]{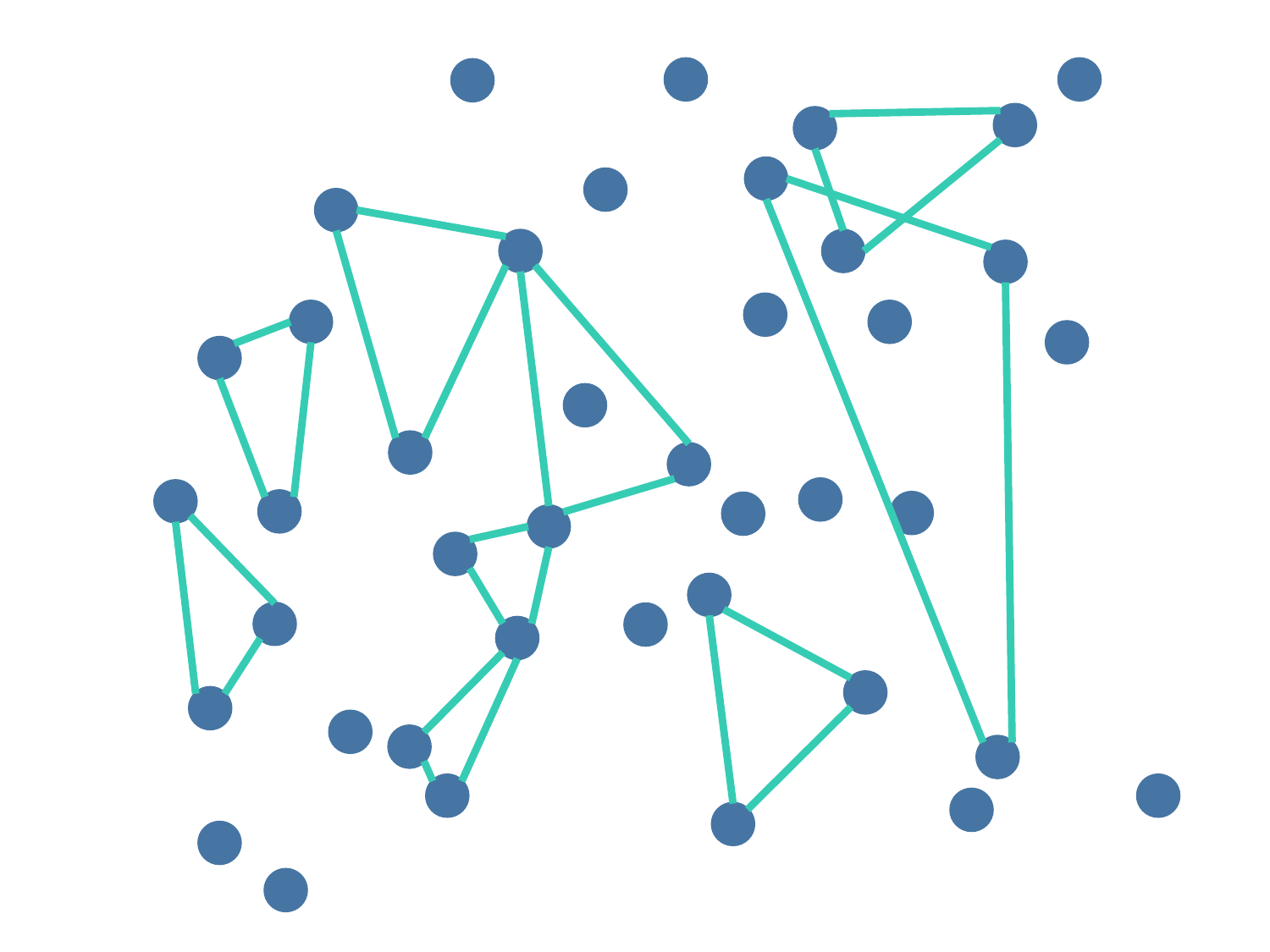}
}
\subfloat[Links form]{
\label{fig:image_1b}
\includegraphics[width=0.25\textwidth]{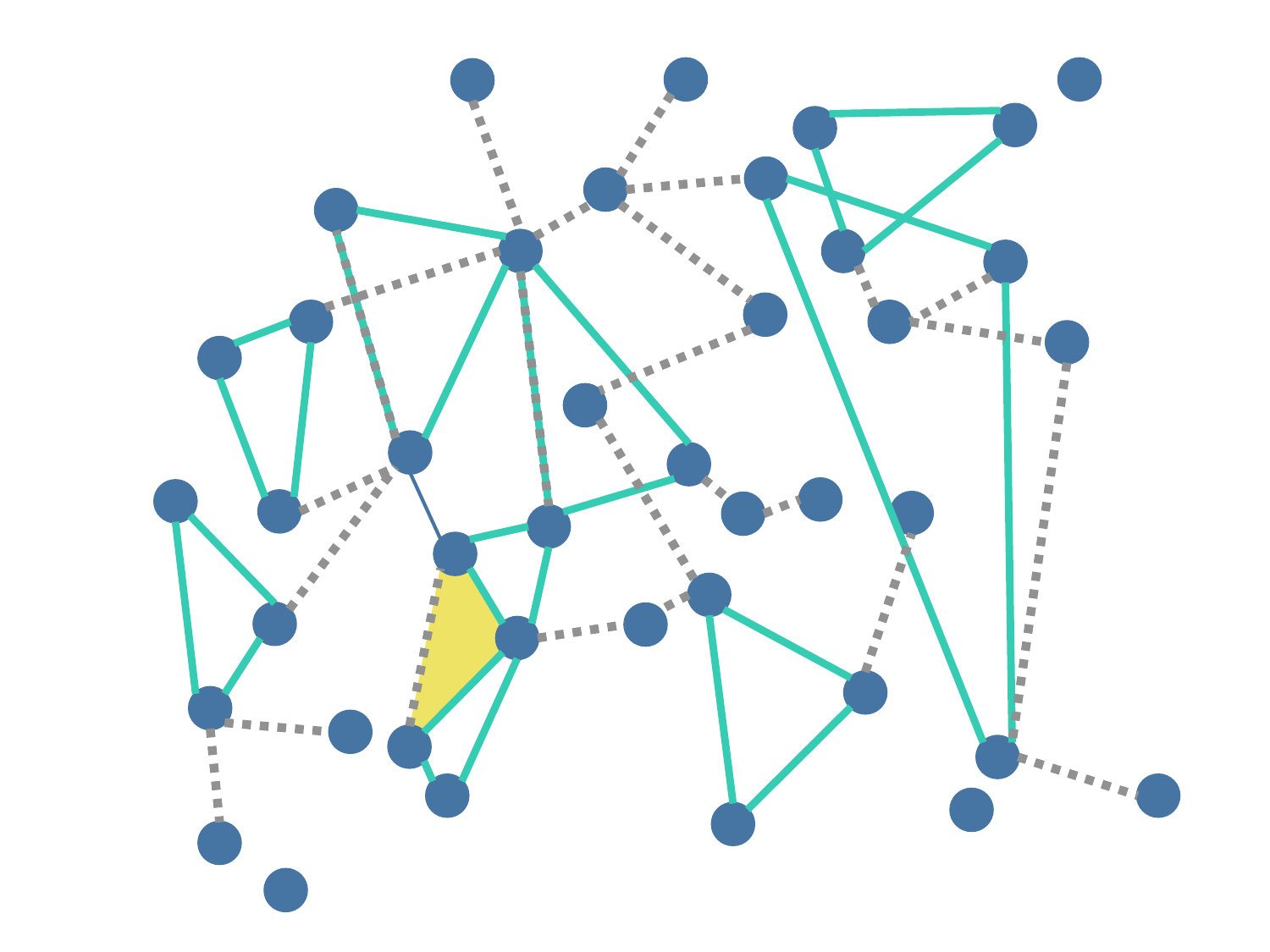}
}
\subfloat[Resulting network]{
\label{fig:image_1c}
\includegraphics[width=0.25\textwidth]{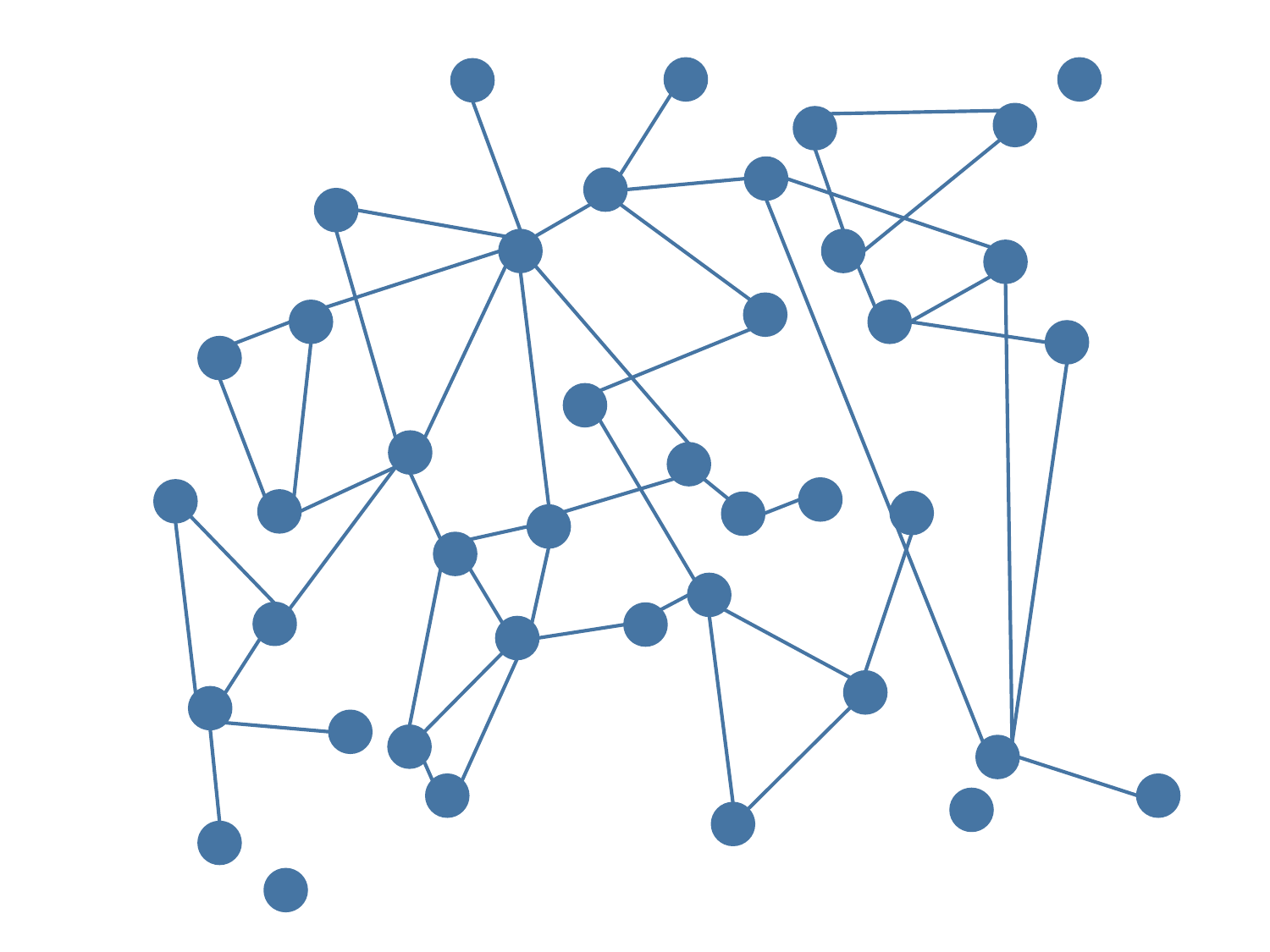}
}
\end{center}
\caption{\label{fig:multi_panel} A network is formed on 41 nodes and is shown in panel D.
The process can be thought of as first forming triangles as in (B), and links  as in (C). Note that two links form on triangles, and a third link incidentally generates an extra triangle.  In this network we would count
$\widetilde{S}^n_T(g)=10$,
and $\widetilde{S}^n_L(g)=22$ from (D), while the true process generated 9 triangles and 23 links directly.
The estimated parameters are $\widehat{\beta}^{DC}_{n,T}=\frac{10}{10660}$,
and $\widehat{\beta}^{DC}_{n,L}=\frac{22}{820}$, while the true frequencies were $\frac{9}{10660}$
and $\frac{25}{820}$. }
\end{figure}

To understand when the direct-count estimator is appropriate,  we need to characterize the rate of incidental subgraph formation.   To do this we track how many
ways a subnetwork  $g' \in G_\ell^n$
could be incidentally generated.

We first provide a precise specification of what it means to be incidentally generated.
We say that a subgraph $g' \in G_\ell^n$ for some $\ell$ can be \emph{incidentally generated} by the subgraphs $\{g^j\}_{j\in J}$, indexed by $J$,
if $g' \subset \cup_{j\in J} g^j$.
For instance a triangle $g'=123$ can be incidentally generated by links $g^1=12$, $g^2=23$, and triangle $g^3=134$; or by link 12 and triangles 234 and 135, etc.
Some of these incidental generations are equivalent to each other (e.g., involve two links and one triangle) and so it is useful to define equivalence classes of generators.

Consider any potential subgraph $g' \in G^n_\ell$  that can be incidentally generated by a set of subnetworks $\{g^j\}_{j\in J}$ with associated indices $\ell_j$ and also by another set
$\{g^{j'}\}_{j'\in J'}$.   We say that $\{g^j\}_{j\in J}$ and $\{g^{j'}\}_{j'\in J'}$ are {\sl equivalent generators} of $g'$ if
there exists a bijection $\pi$ from $J$ to $J'$ such that $\ell_j=\ell_{\pi(j)}$ and $|g_j\cap g'|= |g_{\pi(j)} \cap g'|$.
So the equivalent generating sets have the same configurations in terms of numbers and types of subgraphs, and in terms of how many nodes each of those subgraphs intersects the given network.
For instance a triangle 123 is not only incidentally generated by links 12, 23, and triangle 134; but also by an equivalent generator of links 12, 23, and triangle 135, or links 23, 13; and triangle 128, and so forth.

Given this equivalence relation, we simplify by ignoring the specific labels of subgraphs and defining {\sl generating classes} for any type of subgraph $G_\ell$.  We
just track the number and type of subgraphs needed, as well as
how many nodes each subgraph has intersecting with the given incidentally generated subgraph.

In particular,
each generating class $\mathcal{C}$ of some $G^n_\ell$ is a list  $\mathcal{C} = (\ell_1,c_1,\ldots, \ell_{C},c_{C})$
consisting of a list of types of subgraphs used for the incidental generation and how many nodes each has intersecting with the given
incidentally generated subgraph.
Thus, $\mathcal{C} = (\ell_1,c_1,\ldots, \ell_{C},c_{C})$
is such that  there $\exists g'\in G^n_\ell $ generated by some $  \{g^j\}_{j\in J}$ for which $|J|=C$
and for each $j$:
$g^j\in G^n_{\ell_j}$ and $c_j= |g^j\cap g'|$.
For example, if we consider a triangle, then it can be incidentally generated by two other triangles and a link; and we represent that as
	$(T,2;T,2;L,2)$, where this indicates that two triangles were
	involved and each intersected the subgraph in question in two nodes and then $L,2$ indicates that a link was involved intersecting the subgraph in two nodes.

We order generating classes so that the indices are ordered: $\ell_j\leq \ell_{j+1}$, and lexicographically $c_j\leq c_{j+1}$ whenever
$\ell_j= \ell_{j+1}$.  This ensures that we avoid counting the same class twice.\footnote{
	However, a generating class of two links and a triangle is a different generating class than one link and two triangles - this numbering just avoids the double counting of two links and a triangle separately from a triangle and two links. }

We only need to work with a small set of generating classes, so we restrict attention to the following:
\begin{itemize}
        \item generating classes that only involve smaller subgraphs:   $\ell_j\geq \ell$ for all $j\in J$,
        and
	\item generating classes that are minimal: in the above $J$
	there cannot be $j'$ such that $g'\subset \cup_{j\in J, j\neq j'} g^j$.
\end{itemize}

The first condition
states that we can ignore many generating classes because of our counting convention:
when counting any given subgraph type, we only have to worry about incidental generation by the remaining (weakly smaller) subgraphs.
We do it this way, since we first count the largest subgraphs, and having accounted for them, we worry about the remaining subgraphs, and so forth.  The second condition restricts attention to the smallest generators.  For instance a triangle could be generated by two links and two triangles.   However, in that case either one of the links or one of the triangles can be dropped.  The minimal classes for the triangle only involve three subgraphs:   three links,  two links and one triangle,  one link and two triangles, or three triangles. Under the first condition, there are no generating classes for links to worry about, since they cannot be incidentally generated by themselves and we only count them after removing all triangles.

The following conditions ensure that the direct estimation parameters are arbitrarily accurate for large enough networks.

First, for each $\ell$ let
\begin{equation}
	\label{ln1}
	h_\ell > m_\ell -2.
\end{equation}
This condition ensures that the overall degree of any node  grows more slowly than the size of the graph.  This comes from the fact that any given node can be a part of $\binom{n}{m_\ell -1}$ subgraphs of type $\ell$, each of which forms with probability
$\frac{b_\ell}{n^{h_\ell}} $.  Expecting over these gives an upper bound on the number of links (up to a proportional constant) that a given node is part of from graphs of type $\ell$, and the condition is that this be smaller than $n$.
The average degree can still grow with $n$, but sublinearly.
In particular, this condition ensures that the chance that any given link is part of multiple subgraphs is vanishing.

Next, for each $\ell$
consider any (minimal)\footnote{If the condition is satisfied by minimal classes, it is automatically satisfied
	by larger classes.} generating class with index $J$ of subgraphs no larger than $\ell$.
Let
\begin{equation}
	\label{ln2}
	h_\ell< \sum_{j\in J}  h_{{\ell_{j}}}+ c_{{\ell_{j}}} - m_{{\ell_{j}}}
\end{equation}
and
\begin{equation}
	\label{ln3}
	h_{{\ell_{j'}}} + m_\ell - m_{{\ell_{j'}}}< \sum_{j\in J}  h_{{\ell_{j}}}+ c_{{\ell_{j}}} - m_{{\ell_{j}}}
\end{equation}
for each $j'\in J$.

(\ref{ln2}) is the requirement that a given subgraph is more likely to form directly than indirectly. $h_\ell$ governs the direct formation, and $\sum_{j\in J}  h_{{\ell_{j}}}+ c_{{\ell_{j}}} - m_{{\ell_{j}}}$ governs the rate of incidental generation, and so the exponent on the direction formation must be less than the sum of the exponents of the graphs needed for incidental generation, subtracting off how many variations on each of these there are (captured by the $-(m_{\ell_{j}}-c_{\ell_{j}})$ coming from how many nodes are free to be chosen for each incidentally generating subgraph).
(\ref{ln3}) is the requirement that a given subgraph of some type $\ell_{j'}$  that is part of a generating class of some $\ell$ appear at a fast enough rate to ensure that it is not always becoming part of incidentally generated $\ell$s, but can be distinguished.  This is a similar calculation of rates.

Under these conditions, we prove identification in addition to consistency and
asymptotic normality on a single large network.

Define the variance-covariance matrix
\[
V_{n}={\rm diag}\left\{ n^{2h_\ell}  \frac{\beta^n_{0,\ell}(1-\beta^n_{0,\ell})}{ \kappa_\ell \binom{n}{m_\ell}} \right\}.
\]

We say that a sequence of SUGMs with $k$
types of subgraphs is complete and growing if for each $\ell\in \{1,\ldots,k\}$, $G_\ell^n$ includes all subgraphs that are relabellings of each other and $G_\ell^n\subset G_\ell^{n+1}$.  So, this implies that $\ell$ has the same meaning (e.g., triangles or $k-stars$) across $n$.

\begin{theorem} \label{thm:Largenetwork}
	Consider a growing and complete sequence of SUGMs of $k$ distinct types of subgraphs.  If they have associated true parameters
	$0<b_{0,\ell}$
	such that  $\beta_{0,\ell}^n = \frac{b_{0,\ell}}{n^{h_\ell}}$ and  % for $\mathcal{B}$ a compact subset of $(0,1)^k$.
	 (\ref{ln1})-(\ref{ln3}) hold for each $\ell$ and associated (minimal) generating classes,
	then
	$|\bnDC - b_0| \cvgto 0$ and
	$V^{-1/2}_n  \left(\bnDC-b_0\right)\wkto\mathcal{N}\left(0,I \right)$.
\end{theorem}

Although the conditions may appear hard to understand, they are actually fairly straightforward, and it is easy
to see sufficient conditions that ensure them.

For example,  suppose that each $h_\ell = m_\ell - x$ for some same $x\in (0,2)  $, so that each node has the same order probability of being a part of different sorts of subgraphs.
This is the natural case, as otherwise some subgraphs become infinitely more likely than others.

In that case, all three conditions are automatically satisfied whenever the subgraphs are all cyclic (cliques, or other subgraphs in which all nodes are parts of cycles).
If some of the subgraphs are not cyclic (e.g., lines or stars),  then all three conditions hold if $x\in (0,1)$.

\begin{corollary} \label{cor:Largenetwork}
	Consider a growing and complete sequence of SUGMs of $k$.  If they have associated true parameters
	$0<b_{0,\ell},n^{h_\ell}$
	such that  $\beta_{0,\ell}^n = \frac{b_{0,\ell}}{n^{h_\ell}}$, and
such that  $m_\ell -h_\ell= x$ for each $\ell$ and some $x\in (0,2)  $  and either all subgraphs are cyclic
	or else $x<1$,
		then
	$|\bnDC - b_0| \cvgto 0$ and
	$V^{-1/2}_n  \left(\bnDC-b_0\right)\wkto\mathcal{N}\left(0,I \right)$.
\end{corollary}

In both results, although we state them in terms of $b$s, it is also true that
the ratio of $\betanDCell$ to  $\beta_{0.\ell}$, tends to one. Furthermore, as we show
in the proof, if we normalize
the difference between the estimated probability and the truth by the standard deviation,
this is asymptotically normally distributed.  This is an equivalent
representation of the above result, but is helpful to note as it
does not require knowledge of $h_\ell$s but rather just that
they satisfy the relevant bounds, which is
true of many human networks.

\subsubsection{Minimum Distance Estimator for Non-Negligible Incidental Generation}\label{sec:incidental-estimation}

Theorem \ref{thm:Largenetwork} holds for parameter values for which incidental generation becomes small as a function of the overall
counts of the subgraphs, and works for arbitrary subgraph varieties. However, we may want an estimator that works when incidental generation is not ignorable, even in the limit.

For SUGMs with specific subgraph types, we can explicitly calculate all the incidental rates and account for them,
and develop an estimator that is more
accurate in small samples where there can be
nontrivial incidental generation and also works asymptotically even when there is incidental generation.
In particular, in this section we consider a links and triangles SUGM based
and provide an estimator that
fully accounts for
incidental generation (with extensive details in  Appendix  C.%\ref{sec: networks_normal}).
We prove identification  as well as
consistency and asymptotic normality of a minimum distance estimator.

In order to show the properties of the minimum distance estimator, we show that
the following moments converge:
\[
\frac{S_L^n(g) - \E_{\beta^n_0}[S_L(g)]}{\sigma_L^n} \wkto \mathcal{N}(0,1),
\text{ and }
\frac{S_T^n(g) - \E_{\beta^n_0}[S_T(g)]}{\sigma_T^n} \wkto \mathcal{N}(0,1),
\]
and jointly as well, where $(\sigma_L^n)^2 := \var \left( S_L^n(g) \right)$ and $(\sigma_T^n)^2 := \var \left( S_T^n(g)  \right)$.
Since
\[
S_L^n(g) = \frac{\sum_{i<j}g_{ij}}{\binom{n}{2}} \mbox{ and } S_T^n(g) = \frac{\sum_{i<j<k}g_{ij}g_{ik}g_{jk}}{\binom{n}{3}}
\]
and $g_{ij}^n$ and $g_{ik}^n$ are correlated for any $k$,  $S_L^n$ involves correlated random variables, and since
any two triples in $S^n_T$ that involve a common link are correlated, we need to prove a central limit theorem that shows that such correlation does not cause problems.

Let $S^n(g) = (S_L^n(g), S_T^n(g))'$ be the stacked vector of both shares.
It is useful to define the variance-covariance matrix of the moments
\[
V_{n}=\left(\begin{array}{cc}
\var(n^{h_L} S_L) & \cov(n^{h_L} S_L,n^{h_T} S_T)\\
\cov(n^{h_L} S_L,n^{h_T} S_T) & \var(n^{h_T} S_T)
\end{array}\right).
\]
Finally, let $R_{n}={\rm diag}\left\{ n^{h_{L}},n^{h_{T}}\right\}$. With this defined we can state our result.

Define the minimum distance estimator for a single large network by
\begin{align*}
    \betanMD:= \argmin_{\beta}(S^n(g) - \mathrm{E}_{\beta_0^n}\left[ S^n(g)\right] )'R_n^2(S^n(g) - \mathrm{E}_{\beta}\left[ S^n(g)\right]).
\end{align*}

\begin{proposition} \label{prop:LT_SUGM1-2} Consider a links and triangles SUGM with associated parameters
$\beta_{0,L}^n,\beta_{0,T}^n= \left(\frac{b_{0,L}}{n^{h_L}}, \frac{b_{0,T}}{n^{h_T}} \right)$ with $0\leq \underline{D}<b_{0,L},b_{0,T}<\overline{D}$
such that
$h_L \in (2/3,2) \mbox{ and } h_T \in [h_L + 1, 3h_L], \text{ with }h_T<3$.
Then the minimum distance estimator is consistent,
$|\bnMD - b_0| \cvgto 0$, and \footnote{The expression for $V_n$ is different when $h_T = h_L+1$,
and is given in the proof of the proposition. }
and asymptotically normal,
$V^{-1/2}_n  \left(\bnMD-b_0\right)\wkto\mathcal{N}\left(0,I \right)$.
\end{proposition}

The proof makes use of Theorem \ref{clt}, below. The proof is in Appendix  C.%\ref{sec: networks_normal}.

Proposition \ref{prop:LT_SUGM1-2} covers a wide range of link and triangle densities, ranging from average degree on the order $n^{1/3-\delta}$ to $n^{-1+\delta}$ for any $\delta > 0$. This covers the order constant and logarithmic growth rates of average degree studied in the literature \citep{newman2001random,bollobas2001random,jackson2008social,graham2014empirical}, for instance.

In particular, Proposition \ref{prop:LT_SUGM1-2} covers situations in which the rate of incidental generation (e.g., the proportion of triangles that are generated incidentally) does not vanish asymptotically.  Not only does the estimator have better small sample properties (see the simulations below), but it also works asymptotically in cases that Theorem \ref{thm:Largenetwork} does not.

The restrictions are easily interpretable.  $h_{T}\geq h_{L}+1$ ensures that triangles are not so
numerous that almost all of the links in the network lie in triangles: that $n^{3-h_{T}}$ does not dwarf $n^{2-h_{T}}$.  $h_{T}\leq 3h_{L}$ ensures the opposite: that triangles are not always incidentally formed by links and never formed directly: $n^{3\left(1-h_{T}\right)}$
is not dwarfed by $n^{3-3h_{L}}$.
$h_{L}>2/3$ ensures that links and triangles are disentangled by
imposing a density cap. Finally, $h_L < 2, h_T  < 3$ ensure that there is information in the network---enough links and triangles are present to estimate their formation.

Again we note that although the results are stated in terms of $b$, these are
equivalent statements to saying that ratio of the estimated ($ \betanMD$)
and true  ($\beta_0^n$) frequencies tend to one.  And, that, when self-normalized by the standard
deviations, the empirical frequencies estimated are asymptotically normally distributed. Thus, the result requires no knowledge of $h_\ell$s other than that
they satisfy the relevant bounds.

\subsubsection{Discussion of incidental generation and estimators}

It is instructive to summarize the difference in assumptions and performance of $\betanDC$ and $\betanMD$.
Again, the first requires more sparsity---less incidental generation specifically---than the latter.
Relative to Theorem \ref{thm:Largenetwork}, we can see that Proposition \ref{prop:LT_SUGM1-2} covers cases where incidental generation is not ignorable. Namely, one can check that our result on the $\betanDC$ requires $h_T > 2$ (which means that the probability of a triangle is going to faster at a rate faster than $1/n^2$). But  $\betanMD$ only requires a rate faster than $h_T > 5/3$ or $1/n^{1+2/3}$. This means that triangles (and therefore links, checking the conditions) can appear at a faster rate and still be estimated under the minimum distance estimator but not through direct-counts. We also see evidence of this in our simulations, in Appendix D%\ref{sec:sim_consistency}
, where for very sparse networks both estimators give the same result but the direct-count becomes biased as we increase density whereas the minimum distance estimator remains unbiased.

%%%%%%%%%%%%%%%%%%%%% APPLICATIONS %%%%%%%%%%%%%%%%%%%%%%%%

\medskip
\section{Applications\label{sec:Applications}}

SUGMs are useful for a number of purposes.
First, purely as a statistical modeling tool, SUGMs---even ones with just links and triangles---generate higher-order features of empirically observed social networks that link-based
models (even those accounting for characteristics, unobserved characteristics, geography, and latent locations) cannot.
This is important for prediction.
For example, if one wants to see which networks might form under a hypothetical policy,
a model is only useful if it can generate networks that are likely to occur
at a variety parameter values.  As we demonstrate, our model outperforms stochastic block models, models with node-level fixed effects, latent space models, and ERGMs in generating realistic distributions of networks even with considerably  fewer parameters (e.g., 4 parameter SUGMs versus over 200 (or even 400) parameters in some  alternatives).

Second, a SUGM can be used to test which incentives underlie link formation.  There are many theories (e.g., \citet*{coleman1988closure,jackson2012quilts})
predicting that triangles and other cliques play special roles in maintaining cooperation
in favor exchange.  In order to test such theories, we need a statistical
model that allows us to test whether cliques appear significantly more
often than being randomly generated by links, and whether
they appear in configurations that would be predicted by the game theory.

Third, SUGMs can be used for structural estimation. There are parsimonious
microfoundations---models of mutual consent or search---that
give rise to SUGMs. Structural parameters are useful for welfare analyses, and also aid in examining counterfactuals or policy evaluation. Such parameters are recoverable from SUGM parameters.

We provide three examples. Our first example shows that SUGMs model a myriad of network features much better than other standard models. The other two examples build models of network formation to address specific economic questions.  In both cases, the equilibrium network is a random draw from a SUGM with interpretable parameters.

\subsection{Data}

We use the \citet*{banerjee2013diffusion,banerjee2014gossip} data (https://doi.org/10.7910/DVN/U3BIHX) consisting of a variety of social and economic networks from 75 Indian villages as well
as detailed demographic background.\footnote{See \citet*{banerjee2013diffusion} for more information about the data.}
Having 75 villages allows us to show not only how the model scales with the number of nodes, but also to cover both of our asymptotic frames.

The networks have households as nodes.  There are an average of $n=220$ households per village.
We surveyed adults, asking them about a variety of their daily interactions,
as well as their demographics (caste, education, profession, religion, family size, wealth variables, voting and ration cards, self-help group participation, savings behavior, etc.). We have network data from 89.14 percent of the 16,476 households based on
interviews with 65 percent of all adults between the ages of 18 and 55.
We have data concerning twelve types of interactions: (1) whose houses he or she visits, (2) who visits his
or her house, (3) his or her relatives in the village, (4) non-relatives
who socialize with him or her, (5) who gives him or her medical help,
(6) from whom he or she borrows money, (7) to whom he or she lends
money, (8) from whom he or she borrows material goods (e.g., kerosene,
rice), (9) to whom he or she lends material goods, (10) from whom
he or she gets important advice, (11) to whom he or she gives advice,
(12) with whom he or she goes to pray (e.g., at a temple, church or
mosque).

The answers are aggregated to the household level, but one can also work with the individual-level networks to get similar results to those presented below. How a link is defined varies based on the application. We use undirected,\footnote{Some links are not reciprocated, but that is true at similar rates for
	the questions regarding relatives as compared to the other questions, and so much of the failure of reciprocation may simply be measurement error
	rather than true one-way relationships.  For our purposes here, which are purely to illustrate the ability of the models to
	match data, this distinction is inconsequential.}
unweighted networks that may allow for multiplexing. This also means that we observe 98.8\% of the potential
links between pairs.\footnote{This is a new wave of data relative to our
	original microfinance study that includes more surveys. Note that $1-(1-0.8914)^2 = 0.988$.}

For much of what follows, we work with the borrowing and lending of material goods (questions 8 and 9, with any positive answer indicating a link
being present) that we call ``favor'' links, and the exchange of advice (questions 10 and 11, with any positive answer indicating a link
being present) that we call ``info'' links.

\subsection{Example 1: Matching Features of Empirical Network Data}

A challenge for network formation models has been to capture more than one or two observed features of networks at a time. For instance, many observed social networks are sparse but clustered, which motivates developing models that reflect this \citep{watts1998collective}.
They also have a variety of differing degree distributions (\citep{barabasi1999emergence,jackson2007meeting} and exhibit high levels of homophily \citep*{mcpherson2001,currarinijp2009,currarinijp2010}, which can lead to poverty traps and inequalities \citep*{calvo2007networks,jackson2023inequality}.
There are also features such as the expansion properties of a network
that are described by maximal eigenvalue of the adjacency matrix and governs diffusion processes operating on the network (\cite{bollobas2001random}).  The depth of the max flow min cut speaks to several things such as consensus time in a social learning process \citet{golub2009homophily} as well as the sustainable degree of cooperation  \citep*{karlan2009trust}.

We show that a SUGM fits {economically-relevant} network features in the data
far better than four prominent alternatives.
Importantly, these features were not used
to fit the model. They are the size of the giant component, average path length, and various spectral
properties of the adjacency matrix (e.g., the largest eigenvalue and an eigenvalue measure of homophily).
A simple SUGM outperforms the alternative models despite the fact that the alternative models have many more dimensions such as numerous covariates, $n$ fixed effects, or even $n$ latent space variables, that should give them an advantage in fitting.

Specifically, the alternative models are (a) a standard stochastic block model that includes flexible
controls for continuous covariates that influence edge probabilities; (b) an extension of that model that includes $n$ parameters to
capture node fixed effects (e.g., \cite{graham2014empirical}); (c) a latent space
model \citep*{hoff2002latent} in which nodes have unobserved arbitrary locations in $\mathbb{R}^3$
to be estimated and the probability of linking declines in their latent positions; and (d) an exponential random graph model with links, triangles, and a rich set of covariates.

Before we proceed, we review the features of the graph structure that we examine and why they are interesting.  We look at the first eigenvalue of the adjacency matrix, which is a measure of diffusiveness of a network under a percolation process (e.g., \cite*{bollobas2010percolation,jackson2008social}).   This is intimately related to the expansiveness of the network---namely, for any subset of nodes the number of links leaving the subset relative to the number of links within the subset.
We are also interested in the second eigenvalue of the stochasticized adjacency
matrix.\footnote{The stochasticized adjacency matrix $T$ is defined as
	$T_{ij} = \frac{g_{ij}}{\sum_{k} g_{ik}}$, where either $g_{ii}=1$, or
	$g_{ik}>0$ for some $k\neq i$, as this captures the set of people to whom $i$ listens.}
This is a quantity that is key in local average learning processes and modulates
the time to consensus (\citet*{demarzo2003persuasion,golub2009homophily}), but is also closely related to homophily (\citet*{golub2009homophily}) and is labeled as such in the table below.
Additionally, we look at the fraction of nodes that belong to the giant component
of the network, as well as the number of isolates, as empirical networks are often not completely connected.
Finally, we also consider average path length (in the largest component).

We present the results for favor and info networks.
These networks are reasonably connected (with more than ninety percent of the nodes being in the giant component) and yet also typically sparse.

Our procedure is as follows.  For every village, we estimate six network formation models.

One network formation model
is a link-based model (stochastic block model)
in which the probabilities can depend on geographic distance, caste, the number of rooms
households have, number of beds, quality of electricity provision,
quality of latrines, household ownership status, and squared differences in non-binary variables. The probabilities are estimated using logistic regression and the model has 12 parameters.

The next is the model of  \cite{graham2014empirical}. This is the same formulation of the preceding model, but adds  unobserved heterogeneity in the form of node-fixed effects,
\[
\Pr(g_{ij}=1 \vert X_{ij}) = \Lambda\left(\alpha_i + \alpha_j + \gamma'X_{ij}\right),
\]
where $\Lambda(\cdot)$ is the logit link function
and $X_{ij}$ is the aforementioned vector of demographic
characteristics and polynomials therein.
This model has $n$+12 parameters per network.\footnote{Consistency of all $\alpha_i$
	in addition to $\beta$ has been proven for a dense sequence of graphs (e.g., \cite{chatterjee2010random,graham2014empirical}).}

The third model is a latent space model,
\[
\Pr(g_{ij}=1 \vert w_{ij}) = \Lambda\left(\alpha_i + \alpha_j - \eta \cdot \text{dist}(w_i,w_j) + \gamma'X_{ij}\right),
\]
where now $w_i$ are unobserved positions in $\mathbb{R}^3$.\footnote{We use $\mathbb{R}^3$ as Euclidean is commonly used in the literature, though it is not the only choice. The subject of choice of geometry is addressed in \cite*{lubold2020identifying}, which shows how to check isometric embedding conditions. In this data we find that 25\% of networks are not consistent with \emph{any} latent space from the family of simply connected, complete Riemannian manifolds of constant curvature, lending evidence to the idea that latent space models may not be a universally appropriate device to model correlation.} This has $2n+12$ parameters.

The fourth model is a links and triangles ERGM with covariates. Specifically,
\[
\Pr(g) \propto \exp(\theta_L \cdot S_L(g) + \theta_T \cdot S_T(g) + \gamma'X ).
\]

Turning to SUGMs, in contrast, we consider only low-dimensional models.
One is a the basic SUGM with links and triangles.
Pairs of household are categorized as either being ``close'' or ``far,''
where  ``close'' refers to pairs of nodes that are of the
same caste and ``far'' to those that  differ in caste.
Similarly, we categorize triangles as being ``close'' if
all nodes are of the same caste and ``far'' otherwise.  Thus, we allow for four parameters, close and far link parameters and close and far triangle parameters.
The other model is a slightly richer SUGM in which we allow some
nodes to be isolates, which adds one more parameter.\footnote{With isolates, in a first stage some nodes are randomly chosen to be isolates with a given probability.  For the subsequent formation of other subgraphs, those isolates can be considered as removed from the set of nodes and no subgraph that involves them forms in the subsequent subgraph formation.}
Neither includes any other demographic covariates nor unobserved heterogeneity. We estimate both  via a variation on the minimum distance estimator of Proposition \ref{prop:LT_SUGM1-2}, $\betanMD$, since there appears to be enough incidental generation that needs to be accounted for.\footnote{Specifically, we use a hybrid estimator of first directly (unbiasedly) estimating the  link parameters from the frequency of links among pairs of nodes that have no common neighbors.  We then fix this parameter in the minimum distance estimator to estimate the triangle parameters.  This slightly simplifies the computations.  The code appears in supplementary materials. }

To make the strongest point,  we compare these stark SUGMs that use only same/different caste variables to account for homophily,
to very rich covariate dependent (block) models that can incorporate a large set  of covariates -- including much richer demographics that are
usually available to a researcher as well as node-level fixed effects in the unobserved heterogeneity model and node-level latent locations in the latent space model.
We show that  even though we have considerably more information on the
nodes, such as geographic distance and demographic characteristics, and allow for such unobserved heterogeneities---and we do not make use of this information for the SUGMs---they
recreate networks much more accurately than a link-based model that does
takes advantage of a rich set of node characteristics.
Adding over 12 parameters to the block model to flexibly control for demographic attributes,
{\sl or even $n$+12 parameters with unobserved heterogeneity or $2n+12$ with latent locations},  does not come
close to doing as well as the simple SUGMs.
Moreover, since the specification developed here makes
use of considerably richer data than those used in the
two candidate SUGM models, it suggests that by decomposing
a network into a tapestry of random structures (triangles,
links, and even isolates), considerable value is added in modeling
higher order features of networks in a parsimonious way.

We estimate parameters village-by-village for each model and then
generate random network from each model based on the estimated parameters.
We do 100 such simulations for each village and model.
We then compare the true network characteristics
with those from the simulations for each of the various models.

\begin{table}[!h]
\centering
\caption{Network Properties}\label{tab-emprops}
\resizebox{\ifdim\width>\linewidth\linewidth\else\width\fi}{!}{
\fontsize{10}{12}\selectfont
\begin{tabular}[t]{llllllll}
\toprule
  & Truth & {\footnotesize\makecell[l]{Links/\\Triangles\\SUGM}} & {\footnotesize\makecell[l]{Links/\\Triangles/\\Isolates\\SUGM}} & {\footnotesize\makecell[l]{Covariates\\(Block\\Model)}} & {\footnotesize\makecell[l]{Covariates +\\Unobserved\\Heterogeneity\\(Latent Block\\Model)}} & {\footnotesize\makecell[l]{Latent Space\\Model\\(with\\Covariates)}} & {\footnotesize\makecell[l]{ERGM\\(Links/Triangles\\with Covariates)}}\\
\midrule
\addlinespace[0.3em]
\multicolumn{8}{l}{\textbf{Panel A: Information}}\\
\hspace{1em}Degree & 8.0960 & 8.2166 & 8.2064 & 8.8111 & 9.5860 & 13.2540 & 13.8812\\
\hspace{1em} & (0.2607) & (0.2754) & (0.2746) & (0.3126) & (0.3615) & (0.1257) & (0.1337)\\
\hspace{1em}Clustering & 0.2198 & 0.1599 & 0.1478 & 0.0506 & 0.0742 & 0.0834 & 0.1287\\
\hspace{1em} & (0.0057) & (0.0034) & (0.0032) & (0.0030) & (0.0045) & (0.0007) & (0.0010)\\
\hspace{1em}Isolates & 10.9718 & 3.3361 & 13.4597 & 0.5313 & 0.8369 & 10.7658 & 12.8454\\
\hspace{1em} & (0.8410) & (0.3853) & (0.9924) & (0.0977) & (0.1336) & (0.1448) & (0.1381)\\
\hspace{1em}\% in Giant & 0.9503 & 0.9844 & 0.9397 & 0.9977 & 0.9964 & 0.9434 & 0.9205\\
\hspace{1em} & (0.0030) & (0.0016) & (0.0033) & (0.0004) & (0.0005) & (0.0008) & (0.0013)\\
\hspace{1em}Maximal Eigenvalue & 11.9138 & 10.6260 & 11.0178 & 10.3737 & 12.5418 & 16.2470 & 18.3364\\
\hspace{1em} & (0.3741) & (0.3293) & (0.3466) & (0.3222) & (0.4325) & (0.1293) & (0.1323)\\
\hspace{1em}Homophily & 0.8865 & 0.8156 & 0.8029 & 0.6869 & 0.6795 & 0.8743 & 0.7921\\
\hspace{1em} & (0.0065) & (0.0090) & (0.0093) & (0.0104) & (0.0097) & (0.0009) & (0.0024)\\
\hspace{1em}Average Path Length & 3.0273 & 2.9406 & 2.8580 & 2.7602 & 2.6399 & 3.1017 & 3.1163\\
\hspace{1em} & (0.0485) & (0.0422) & (0.0393) & (0.0403) & (0.0374) & (0.0106) & (0.0172)\\
\addlinespace[0.3em]
\hline
\addlinespace[0.3em]
\multicolumn{8}{l}{\textbf{Panel B: Favor}}\\

\hspace{1em}Degree & 7.0579 & 7.2192 & 7.2185 & 7.7614 & 8.5145 & 13.1301 & 16.6176\\
\hspace{1em} & (0.2611) & (0.3048) & (0.3052) & (0.3232) & (0.3835) & (0.1571) & (0.1347)\\
\hspace{1em}Clustering & 0.2895 & 0.1894 & 0.1764 & 0.0467 & 0.0641 & 0.0724 & 0.1506\\
\hspace{1em} & (0.0054) & (0.0034) & (0.0032) & (0.0032) & (0.0040) & (0.0008) & (0.0008)\\
\hspace{1em}Isolates & 10.0704 & 7.2859 & 16.0939 & 1.0423 & 3.4777 & 19.3338 & 15.8796\\
\hspace{1em} & (0.7670) & (0.6645) & (1.1103) & (0.1429) & (1.7734) & (0.2763) & (0.2046)\\
\hspace{1em}\% in Giant & 0.9510 & 0.9632 & 0.9245 & 0.9955 & 0.9820 & 0.8695 & 0.9108\\
\hspace{1em} & (0.0032) & (0.0031) & (0.0041) & (0.0005) & (0.0107) & (0.0021) & (0.0012)\\
\hspace{1em}Maximal Eigenvalue & 10.0654 & 9.8417 & 10.1727 & 9.4778 & 11.3075 & 16.0454 & 21.5754\\
\hspace{1em} & (0.3337) & (0.3702) & (0.3826) & (0.3382) & (0.4236) & (0.1637) & (0.1201)\\
\hspace{1em}Homophily & 0.9412 & 0.8716 & 0.8636 & 0.7325 & 0.7189 & 0.9074 & 0.7895\\
\hspace{1em} & (0.0044) & (0.0082) & (0.0085) & (0.0107) & (0.0111) & (0.0010) & (0.0025)\\
\hspace{1em}Average Path Length & 3.5158 & 3.1591 & 3.0739 & 2.9140 & 2.7799 & 3.8148 & 2.8126\\
\hspace{1em} & (0.0659) & (0.0479) & (0.0441) & (0.0442) & (0.0443) & (0.0212) & (0.0149)\\
\bottomrule
\end{tabular}}
\end{table}

Table \ref{tab-emprops} presents the results, averaged across villages for each of the models.
We use 71 villages out of the 75 since 4 villages have only one caste group.   The ERGM is only estimated for 68 villages as it did not converge for 3 villages.
Both of the SUGMs match the features of the networks
substantially better than the conditional edge independent models (with
and without node fixed effects).
Including isolates in the SUGM
further improves the fits not only for isolates,
but also for fraction in the giant
component and the maximum eigenvalue.
This suggests that there are more isolated households in a village
for a reason beyond randomness in network formation.

An obvious thing to note is that the link-based and also latent space models do extremely poorly when it comes to
matching clustering while the SUGM does much better, and here adding unobserved dimensions to generate
unconditional link correlations (e.g., clustering) does worse than a SUGM that allows correlated link formation directly. The ERGM performs better on clustering but a the cost of generating excessive density, diffusiveness, the spectral cut (homophily), connectedness, and average path length.

Including triangles in the SUGM is
enough to deliver better matches on all   dimensions, and the difference on homophily is perhaps most interesting, since one would imagine
that the block models or even latent space models could get that right given that they include many covariates.
This tells us that triangles and correlation between links play a subtle but important role in homophily---something that is better picked up by a SUGM than an independent link model even when that model
includes rich demographics
and unobserved heterogeneity.

It is important that SUGMs do a much better job at recreating a multitude of features of
observed network structures that standard link-based models,
especially with rich demographic information, models with unobserved heterogeneity, latent space models, and ERGMs. It suggests that there is a substantial
value added of modeling the formation of triangles and isolates.
Knowing that our model is better able to capture the realistic correlation
of links within observed networks should make us more confident in trusting the results
of some other empirical applications. For example, when we look at links across social boundaries,
we can be comfortable that to a first order, thinking about a SUGM with links and triangles
across and within caste groups can do a good job of matching patterns in the data, and thus
tracing them back to model parameters.

\subsection{Example 2: Do incentives for risk sharing drive network formation?}\label{sec:ex1b}

\subsubsection{A model of mutual consent}

Consider a simple model in which individuals get utility from being
in bilateral relationships,   denoted
by $L$, as well as trilateral relationships, denoted by $T$.
The value of a partner $j$ to $i$ in a bilateral relationship
is a function of their demographics (given by vector $X_i$) is given by $u_{i}^{L}$:
\[
u^{L}\left(X_i; X_j\right)=X_{i}'{\gamma}_{L1}+ X_{j}'{\gamma}_{L2}+ {\gamma}_{L3}d_L(X_i,X_{j})-\epsilon_{ij} =: \phi_L(X_i;X_j) - \epsilon_{ij}.
\]
where $d_L(X_i,X_{j})$ is a distance or other function comparing the demographics---for instance to allow for homophily.
Similarly, the value of a triangle  of relationships $jk$ to $i$ is
given by $u_{i}^{T}$:
\[
u^{T}\left(X_i; X_j,X_k\right)= X_{i}'{\gamma}_{T1}+ f(X_{j},X_{k})'{\gamma}_{T2}  + {\gamma}_{T3}d_T(X_i;X_{j},X_k)-\varepsilon_{ijk} =: \phi_T(X_i; X_j,X_k) - \varepsilon_ijk,
\]
where $f(\cdot,\cdot)$ is a function that is symmetric in arguments, and $d_T(X_i;\cdot,\cdot)$ is a function that is symmetric in the last two arguments.
The value of the relationships depend on the characteristics of the people involved, as well as
some idiosyncratic values to the relationships, $-\epsilon_{ij}$ and $-\varepsilon_{ijk}$, which may capture personalities, compatibilities, etc.,  distributed according to some distributions $F_L$ and $F_T$ respectively.

Forming relationships requires mutual consent (e.g., the pairwise stability of \cite{jackson1996strategic}), so the net utility must be positive to all agents.
The probability that a subgraph $ij$ forms
is
\[
 {\beta}_{L}\left(X_{ij}, {\gamma}_{L}\right)=F_L\left(
 \phi_L(X_i;X_j)
 \right) \times
F_L\left(
\phi_L(X_j;X_i)
\right)
\]
and similarly the probability that subgraph $ijk$ forms  is
\begin{align*}
 {\beta}_{T}\left(X_{ijk}, {\gamma}_{T}\right)&=F_T\left(
 \phi_T(X_i;X_j,X_k)
 \right)  \times F_T\left(
 \phi_T(X_j;X_i,X_k)
 \right) \times F_T\left(
 \phi_T(X_k;X_i,X_j)
 \right).
\end{align*}
The products capture that a link
requires two consents and a triangle requires three.

By estimating the probabilities of subgraphs forming
(${\beta}_{T}\left(\cdot\right)$ and ${\beta}_{L}\left(\cdot\right)$), under suitable assumptions described below, one can
recover the marginal effects of changes in covariates on preferences
for being in various configurations (${\gamma}_{T}$ and ${\gamma}_{L}$). Since we have finite support for covariates, we label the subgraph formation probabilities  $\beta_{T,X_{T}}$
and $\beta_{L,X_L}$ for pair and node covariate combination $X_T$ and $X_L$
respectively.

\subsubsection{Incentives for Risk-Sharing}

\citet*{jackson2012quilts} show that whether or not a link is supported
can play an key role in maintaining informal favor exchange when it would not be self-sustaining without social pressure. It
characterizes renegotiation-proof and robust pairwise stable networks
and shows that (in the homogenous parameter case) all networks that incentivize exchange
are quilts ({a union of cliques with no cycle involving more than the minimal clique-size number of nodes}), and in the inhomogenous parameter case every link must
be supported ({if $i,j$ are linked then there exists $k$ such that $g_{ik}=g_{jk}=1$}).

Consider a variation on this model wherein
now there are multiple link types: favors and information. We
can use this to study the question raised by \citet*{jackson2012quilts}.
To make this simple ignore covariates, so all nodes are
identical. Preferences are described by a random utility framework \citep{mcfadden1973conditional}, with the value of a link between $i$ and $j$
to $i$ given by
\[
u_{i}^{L,favor}\left(i\right)= {\gamma}_{L,favor}-\epsilon_{ij,favor},\ u_{i}^{L,info}\left(i\right)= {\gamma}_{L,info}-\epsilon_{ij,info},
\]
and the value of a triangle given by
\[
u_{i}^{T,favor}\left(jk\right)= {\gamma}_{T,favor}-\epsilon_{ijk,favor},\ u_{i}^{T,info}\left(jk\right)={\gamma}_{T,info}-\epsilon_{ijk,info}.
\]

In this case, due to mutual consent, $\beta_{L,favor} = F(\gamma_{L,favor})^2$ and $\beta_{T,favor} = F(\gamma_{T,favor})^3$. It is analogous for information.
By the arguments of  \citet*{jackson2012quilts}, we can test the hypothesis that fraction of links that are supported
is higher in favor exchange than in information links, which would be consistent with exchanging of favors needing to be incentivized while sharing of information not needing network incentives. In the
language of this model the hypothesis is expressed in terms of parameters as follows:\footnote{It is without loss of generality to take $F(\gamma) = \gamma$ which is just a bijection and is convenient to work with.}

\begin{lemma} \label{lem:favor-info}Under the above assumptions,  %a test for
	 \[
  \frac{\gamma_{T,favor}}{\gamma_{L,favor}}>\frac{\gamma_{T,info}}{\gamma_{L,info}} \text{  corresponds to    } \frac{\beta_{T,favor}/\beta_{L,favor}^{3/2}}{\beta_{T,info}/\beta_{L,info}^{3/2}}>1.
  \]
\end{lemma}

 The (joint) hypothesis that we are testing is that exchanging material goods is more costly and/or happens less frequently for agents, and so requires more incentives and
supporting enforcement than exchanging information which is less costly and/or more frequent.

Given that triangles can be incidentally generated, one cannot test this simply by examining the ratio of supported links to unsupported ones. If  ${\gamma_{L,info}}$ was very high, then it could be that there are many incidentally generated information triangles, and few unsupported links, and so we need to estimate the underlying parameters using our techniques to account for incidental generation.
To keep the illustration in this first example clear, we abstract from covariates.  We illustrate the incorporation of covariates in the examples below.

\begin{table}[!h]
\centering
\caption{\small{Parameter estimates by network type\label{tab:infofavors}}}
\centering
\fontsize{10}{12}\selectfont
{\begin{threeparttable}
\begin{tabular}[t]{lll}
\toprule
 & $\betaRMDL$  & $\betaRMDT$\\
\midrule
Information & 0.0119 & 0.0001\\
 & (0.0191) & (0.0001)\\ \\
Favor & 0.0109 & 0.0002\\
 & (0.0307) & (0.0002)\\
\bottomrule
\end{tabular}
			\begin{tablenotes}
 				\item {\footnotesize Notes: Standard errors computed
 				using the results of Proposition \ref{prop:many_networks_LT}.}
 			\end{tablenotes}
\end{threeparttable}}
\end{table}

First, we estimate the four parameters in question under the many independent network ($n$ fixed, $R\rightarrow \infty$) framework (Proposition \ref{prop:many_networks_LT}).
Table \ref{tab:infofavors} presents the parameter estimates and standard errors.
Although the point estimates are in line with the theory,
the standard errors estimates are large and we cannot reject the null hypothesis
that there
is no difference in the support of favor relationships
compared to information relationships.\footnote{Specifically,
	the $p$-value is computed for a test of the null
	hypothesis
	$\frac{\beta_{T,favor}}{\beta_{T,info}}
	= \frac{\beta_{L,favor}^{3/2}}{\beta_{L,info}^{3/2}}$, where the
	parameters are held to be common across all villages in the
	sample.
If instead of using the conservative standard errors from Proposition \ref{prop:many_networks_LT}, we bootstrap them then the hypothesis is rejected at the 1 percent level.
}

We cannot conclude that the data are consistent with the
theory that incentives for favor exchange matters in network
formation in these data, but in part because the parameters actually vary nontrivially across villages.  This leads to high standard errors and also suggests that the more appropriate approach is to examine villages separately.

Thus, we push this further by
estimating the parameters separately for each
village $v$, with the large single network ($n \rightarrow \infty$, $R = 1$) paradigm for each village. This allows for heterogeneity in the parameters across villages by assuming they are drawn from entirely different distributions. Again, a we use the same variation on the minimum distance estimator of Proposition \ref{prop:LT_SUGM1-2} used in Example 1.

We see the results in Figure \ref{fig:FavorInfo}, with
standard errors omitted for visual clarity.
We see that for almost all of villages, the
favor over info ratios are higher for triangles compared to links (more than would occur at random at the 1 percent level if the two had the same distribution).

\begin{figure}[!h]
    \includegraphics[trim = 0.5in 2.5in 0.5in 2.25in, clip,scale = 0.4]{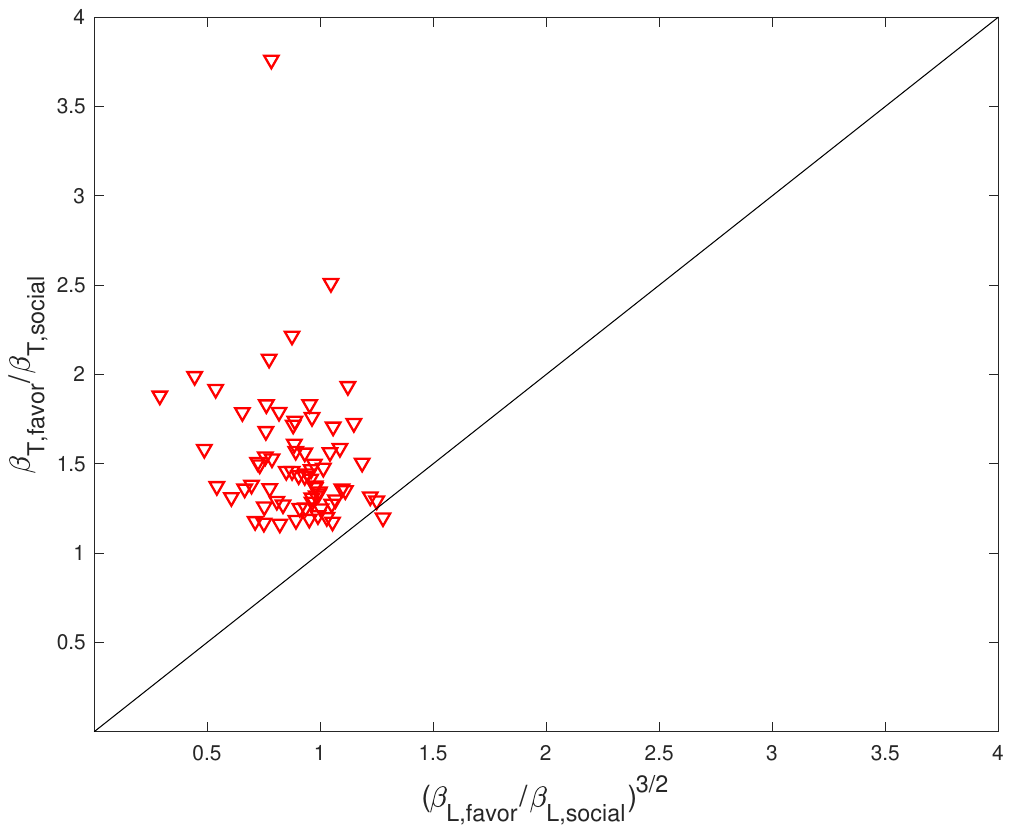}
    \caption{\label{fig:FavorInfo} Plots of estimates of $\frac{\beta_{T,favor}}{\beta_{T,info}}$ against $\frac{\beta_{L,favor}^{3/2}}{\beta_{L,info}^{3/2}}$ by village.}
\end{figure}

\subsection{Example 3: Links across Social Boundaries}\label{sec:caste-cont}

Our next example shows how a SUGM can be used to investigate
whether there are norms governing link-formation across different
social groups. Identities  can lead to strong
social norms---prescriptions and proscriptions---concerning
interactions across groups.  For instance, in much of India there are strong forces that influence if and when individuals can form relationships across castes, particularly among ``upper caste'' Hindus and the ``lower caste'' communities, comprised of Dalits (Scheduled Castes, SC) and scheduled tribes (ST). The SC and ST communities are those defined by the Indian government as being disadvantaged.  This is a fundamental  distinction over which the strongest cultural forces are likely to focus.  Additional norms are at work with finer caste or subcaste distinctions, but those norms are more varied depending on the particular castes in question while this provides a clear  barrier \citep{munshi2006traditional}.

Among many, one natural question concerns the norms around forming public versus private cross-caste group relationships. Namely, are members of upper and lower caste more likely to form cross-group relationships  when those links are unsupported (without any friends in common) compared to when those links are supported with at least one friend in common (and thus have a witness to the relationship)?

To answer this we need a model that accounts for link dependencies; cliques of three or more may exhibit greater adherence to a norm prohibiting certain inter-caste relationships, while the norm may be circumvented in isolated bilateral relationships. In particular, we can get at this hypothesis by testing whether the relative frequency of triangles compared to links is higher when the relationships are  within caste than across caste.

This example is instructive because it is more subtle than the previous example   and it demonstrates that a SUGM can be used for a hypothesis test even when preference parameters are not identifiable without additional restrictive assumptions.  Consider a process in which individuals may meet in pairs or triples and then decide whether to form a given link or triangle.
The link is formed if and only if both individuals prefer to form the link, and a triangle is formed if and only if all three individuals prefer to form it. This minimally complicates an independent-link model enough to include link interdependencies.

Individuals' probabilities to have opportunities to form links or triads can depend  on the
composition of castes of those involved. So let  $\pi_L(diff), \pi_L(same)$ denote the probabilities that a given link has an {\sl opportunity} to form (i.e., the pair meets and can choose to form the relationship) that depend on the pair of individuals being of different castes or of the same caste, respectively. Analogously define  $\pi_T(diff), \pi_T(same)$. Notice these are unlikely to be observed by the researcher.

As noted above, individual $i$'s utility of having a relationship with $j$ can by influenced by whether they share caste ($x_{ij}$ a dummy variable for same caste) and is given by
\[
u_i^L(j) = \alpha_{0,L} + \gamma_{0,L} x_{ij} - \epsilon_{L,ij}
\]
and similarly for a triad,
\[
u_i^T(jk) = \alpha_{0,L} + \gamma_{0,T} x_{ijk} -  \epsilon_{T,i,jk},
\]
where $x_{ijk}$ is a dummy for whether all three individuals are members of the same caste.\footnote{This is a simplified model for
	illustration, but one can clearly consider preferences conditional on any string of covariates.  This extends a model such
	as that of \cite*{currarinijp2009,currarinijp2010} to allow for additional link dependencies. We could also be interested in higher order relationships.}  The probability of an individual consenting to a subgraph of type $z \in \{L,T\}$   among the $m_z$ nodes is
\[
p_{z,same} = F(\alpha_{0,z} + \gamma_{0,z}) \text{ and } p_{z,diff} = F(\alpha_{0,z} ).
\]
The hypothesis that we explore is that $$  \frac{p_{T,diff}}{p_{T,same}} < \frac{p_{L,diff}}{p_{L,same}}$$ so that people are more reluctant to involve themselves in cross-caste relationships when those are ``public'' in the sense that other individuals observe those relationships.

The researcher does not observe either the meeting probabilities nor the probabilities within the mutual consent process. Rather, the researcher observes the compositions $\beta_\ell$ for $\ell \in \{L,T\} \times \{same,diff\}$ which are precisely SUGM parameters:
\begin{enumerate}
	\item $\beta_{L,same}  = p_{L,same}^2 \pi_L(same)$  and  $\beta_{L,diff} = p_{L,diff}^2 \pi_L(diff)$, and
\item $\beta_{T,same} = p_{T,same}^3 \pi_T(same)$ and $\beta_{T,diff} = p_{T,diff}^3 \pi_T(diff)$.
\end{enumerate}

There are two challenges. Recall the difference in the exponents reflects that it is more difficult to
get a triangle to form than a link.  Hence, to perform
a proper  test, we have to adjust for the exponents as otherwise
we would just uncover a natural bias due to the exponent that would end up favoring cross-caste links. Further, identifying a preference bias is confounded by the meeting bias.  Thus, we first model the meeting process $\pi_{z}(x)$ more explicitly and show that we still have identification as the meeting bias makes triangles relatively more likely to be cross-caste than links.

Consider a meeting process where people spend a fraction $f$ of their time mixing in the community that is predominantly of their own types and a fraction $1-f$ of their time mixing in the other caste's community. Then at any given snapshot in time, a community would have $f$ of its own types present and $1-f$ of the other type present.\footnote{Variations on this sort of biased meeting process appear in \citet*{currarinijp2009,currarinijp2010,bramoulle-etal}.}
This  generates a conservative test  in the sense that if we find cross-caste links relatively more likely, that is evidence for a preference bias.

\begin{lemma} \label{lem:geomodel} A sufficient condition for
	$\frac{p_{T,diff}}{p_{T,same}} < \frac{p_{L,diff}}{p_{L,same}}$ is that
	$\frac{\beta_{T,diff} }{\beta_{T,same}}<\left(\frac{\beta_{L,diff} }{\beta_{L,same} }\right)^{3/2}.$
\end{lemma}

Turning to the data, we link two households if members of either engaged in favor exchange with each other: i.e.,  they borrowed or lent goods such as kerosene or rice in times of need.

\begin{table}[!h]
\centering
\caption{Parameter estimates by network type\label{tab:casteboundaries}}
\centering
\fontsize{10}{12}\selectfont
\begin{tabular}[t]{lllll}
\toprule
 & $\hat{\beta}_{R,L,same}^{\mathrm{MD}}$  & $\hat{\beta}_{R,T,same}^{\mathrm{MD}}$ & $\hat{\beta}_{R,L,diff}^{\mathrm{MD}}$  & $\hat{\beta}_{R,T,diff}^{\mathrm{MD}}$\\
\midrule
Information & 0.016891 & 0.000309 & 0.006535 & 0.000002\\
& (0.020001) & (0.000156) & (0.008739) & (0.000048)\\ \\
Favor & 0.012652 & 0.000302 & 0.004281 & 0.000003\\
   & (0.024433) & (0.000194) & (0.006796) & (0.000038)\\
\bottomrule
\end{tabular}
\end{table}

Table \ref{tab:casteboundaries} presents the parameter estimates using the estimator from Proposition \ref{prop:many_networks_LT}, for the estimation in which we assume that all 75 networks are independent draws from the same distribution.
Similar to the example above, the large standard errors on the triangle parameters lead us to fail to   reject the null
hypothesis that $\frac{p_T(diff)}{p_T(same)}=\left(\frac{p_L(diff)}{p_L(same)}\right)^{3/2}$.
Again, this suggests that since the parameters vary by village, that we should work with estimation by village.

\begin{figure}[!h]
	
	\subfloat[Information]{
		\includegraphics[trim = 0.5in 2.25in 0.5in 2.25in, clip,scale=0.35]{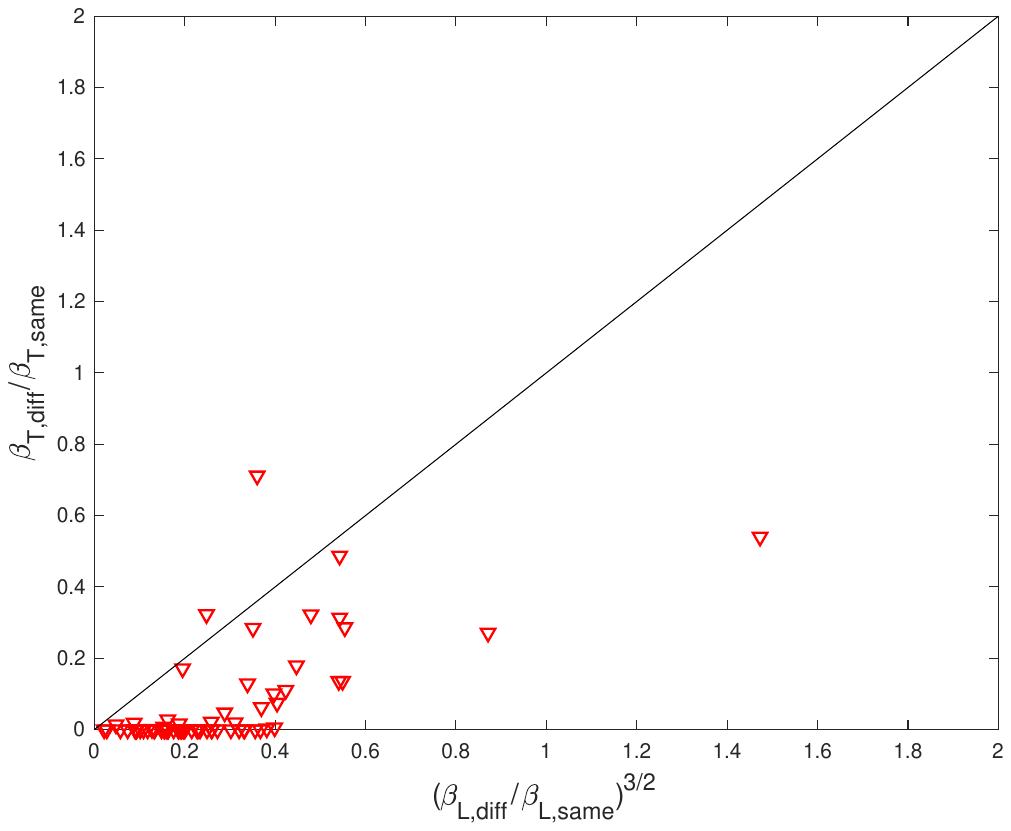}
	}
	\subfloat[Favors]{
		\includegraphics[trim = 0.5in 2.25in 0.5in 2.25in, clip,scale=0.35]{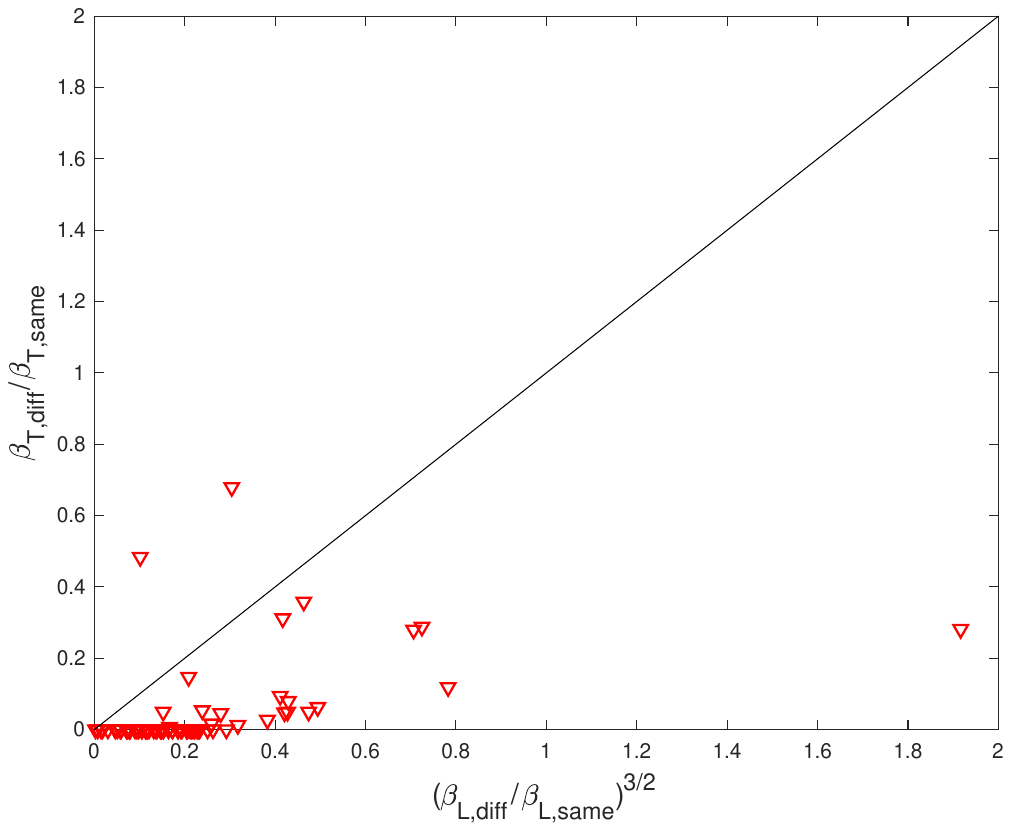}
	}
	\caption{\label{fig:Caste} Plots of estimates of $\frac{\beta_{T,diff}}{\beta_{T,same}}$ against $\frac{\beta_{L,diff}^{3/2}}{\beta_{L,same}^{3/2}}$ by village for information and favors networks.}
\end{figure}

Finally, Figure \ref{fig:Caste} shows the results when
we allow the parameter estimates to vary by village, again using the variation on the estimator of Proposition \ref{prop:LT_SUGM1-2} that we used in Example 1.
For the bulk of villages, cross-caste relationships relative
to within-caste relationships are more frequent as
isolated links compared to being embedded in triangles,
for both information and favor networks.
This is a new observation in the literature and begs the question as to the foundations as to why members of two groups which meet less frequently and may have less affinity for each other may nonetheless be able to sustain interactions privately. They appear to  give up support for the sake of not having the interaction be public.

\medskip

%%%%%%%%%%%% CLT %%%%%%%%%%%%%%%%%%%%%%%%%%%%%%%%%%%%%%%%%
\section{A Central Limit Theorem for Correlated Random Variables}\label{sclt}

We now state a new central limit theorem that applies to a variety of settings in which all variables may be correlated (well-beyond network settings) but the total amount of covariance is bounded. We require it to prove Proposition \ref{prop:LT_SUGM1-2}, but the result is considerably more general. Our result is derived from Stein's method, which provides the foundations for many CLTs with dependent variables, but as will become clear, our treatment delivers more general conditions that also allow us to cover SUGMs \citep{stein1986approximate}.

To understand the necessity of such a theorem at a high intuitive level,  begin with the simplest example where $n=3$ and the only possible subgraphs are a triangle and three links.   Here if there is a triangle, then there are necessarily three links, and if there is no triangle there are necessarily fewer than three links.
There is a strong dependence in the subgraph counts, and this carries over in the estimates of parameters.    The overall dependence of different subgraph counts is not so extreme as $n$ grows, but there is always a correlation and adjacent subgraphs remain nontrivially correlated all along the sequence.  Thus, we need a theorem that covers moments that are based on variables, some of which are highly correlated all along the sequence.

There are essentially two approaches used to prove CLTs relevant that allow for correlated random variables and the existing results in neither  setting apply to SUGMs.

The first approach to CLTs with dependence we can call ``geometric'': random variables carry indices locating them in some embedding space. This covers time series, spatial data, mixing random fields. Random variables are embedded in some space where there are ``close'' and ``far'' random variables and the further they are, the less correlated they are.  With enough moment conditions, using distance-based limits on correlation allows researchers to leverage Stein's lemma under mixing conditions to derive a CLT (e.g., \citep{bolthausen1982,jenishp2009} just to name a few). In fact, a literature evolved focusing on mixing random fields leveraging such conditions \citep{froot1989consistent,conley1999gmm,driscoll1998consistent}.\footnote{See also \cite{kuersteiner2019limit}, which studies conditional mixingale types of assumptions so that nodes that are ``far'' in characteristic space have decaying dependence.}    Some researchers   working on network formation (e.g., \cite{boucher2012my,leung2014}) exploit such spatial techniques by embedding nodes in some space so that only ``nearby'' nodes can link and ``distant'' nodes  cannot link  (e.g., following the logic in \cite{penrose2003rgg,hoff2002latent}) in order to satisfy mixing conditions and apply a central limit theorem like \cite{bolthausen1982}. As $n \rightarrow \infty$ most nodes get further and further apart and therefore essentially never link.

SUGMs do not lend themselves to  ordering the indices of random variables into time, space, lattices, or more general random fields. The reason that this does not work as it imposes specific structure on the adjacency matrix.  For example, consider the simple case where nodes live on a line. Then in the adjacency matrix, only nodes within some limited distance to the left or right of any given node tend to be linked to that node.  While this is fine for certain contexts, it is not an adequate description of a village network where there is no natural space on which some households in a village should be considered, ex ante, to be infinitely far apart (or students in a university who are, ex ante, infinitely unlikely to link to each other, etc.).  We can prove a CLT without this structure and that nests such previous theorems, and so it is worthwhile to do so.

A second approach to CLTs with dependence is to use dependency graphs (\citet*{baldi1989normal,goldstein1996multivariate,chen2004normal,ross2011fundamentals}).\footnote{\cite{aronow2017estimating} use \cite{chen2004normal} assuming conditional independence in  treatment effects with spillovers,  where certain interferences are ruled out. But this means one cannot study spillovers due to information, for example, because in principle information can flow to any node in the network, so treatment of one node leaves non-zero exposure to all others. Another application of \cite{chen2004normal} is in \cite{LeungMoon22}. They characterize dependence through radii of stability which correspond to changes in network features when dropping observed links.}
There is another graph (not a SUGM but a graph among indices of random variables), where edges between random variable indices indicates conditional dependence (and independence if there is no edge).\footnote{Some work, e.g.,  \cite{chen2004normal}, focuses on exact finite sample Berry-Essen inequalities.} While that literature leverages \cite{stein1986approximate} to prove CLTs despite not forcing a geometric structure, they require most entries of the dependency graph to be zero. That is, they impose a sparse (dependency) graph with independence across many pairs of random variables (and usually do not consider triangular arrays).  SUGMs violate both impositions since at each $n$, all links can have non-zero correlations with each other.

Our insight is
that if the overall covariances satisfy some bounds, then one can still prove a CLT no matter how that dependency is arranged and even without numerous conditional independencies.

Note that both the geometric or dependency graph approaches   limit
the total correlation of the $\binom{n}{2}$ random variables with specific structure. We develop a CLT in which all random variables can have non-zero correlation by controlling the total size but not forcing zeros. We show that if one can ``collect`` what we call \emph{affinity sets} for each index---that is a set of other random variables that can have larger correlation with the reference random variable---then as long as three covariance conditions hold, we can limit the overall correlation and apply Stein's method.
Our conditions are that on average (1) within an affinity set, most of the correlation comes from the dependence between the reference variable and its members rather from than between two members; (2) the amount of correlation between members of two different affinity sets is small; (3) the correlation between the reference variable and all those outside its affinity sets is appropriately small. We argue that these are relatively intuitive (weighted covariance conditions), interpretable, easy to check, and more easily micro-founded than potentially a more complex but specific assumption on things like mixing random fields. It is in this sense that we believe this is of  independent interest.  In a follow-up paper \citep*{chandrasekhar2023general} we provide a number of other (non-network) applications and additional discussion of the literature.

We require some notation.

Consider a triangular array of (real-valued) random variables   $X_\alpha^N$
 with a set of labels $\alpha \in \Lambda^N$
such that  $|\Lambda^N|=N$.
For instance, in our SUGM setting the $X_\alpha^N$ may be an indicator of the appearance of some particular subgraph,
such as a link or triangle, and $\alpha$ would track the pairs of nodes involved in a potential link ($ij$) or triples of nodes in a triangle ($ijk$). $N$ captures the
$\binom{n}{2}$ possible links or
$\binom{n}{3}$ possible triangles. So when considering link counts $\alpha$ would track pairs of nodes involved in links and when considering triangles $\alpha$ would track triples.

Let us normalize the variables by their means:
\[
Z_{\alpha}^N= X_{\alpha}^N-{\rm E}\left[X_{\alpha}^N\right].
\]
We presume that the $
Z_\alpha$s are such that the $\E[|Z_\alpha^N|^3]/\E[(Z_\alpha^N)^2]^{3/2}$ is bounded above for all $\alpha,N$.\footnote{This condition holds, for instance, if the $Z$s are from Bernoulli random variables, or if  $\E[|Z_\alpha^N|^3]$ is bounded above and $\E[(Z_\alpha^N)^2]$ is bounded below across $\alpha,N$, but also in cases where such individual bounds do not hold.  Violations of this condition are, for example, random variables with finite second moments but infinite third moments.}
We provide conditions under which a normalized statistic, as $N\rightarrow \infty$, converges to the standard normal distribution,
\[
\frac{\sum_{\alpha \in \Lambda^N} Z_{\alpha}^N}{\sqrt{a_N}}\rightsquigarrow\mathcal{N}(0,1),
\]
where the normalizer, $a_N$, is a measure of the variance of the sum, defined below.

\subsection{{Affinity Sets}}

For each $\alpha,N$, we partition the index set
$\Lambda^N$ into two pieces: an \emph{affinity set} and its complement.
In particular, we define an affinity set, for each $\alpha,N$, as
\[
\mathcal{A}
\left(\alpha,N\right)\subset\Lambda^{N} \ \ {\rm such \ that \ \ }  \alpha\in \mathcal{A}
\left(\alpha,N\right).
\]
The conditions for $\eta \in \mathcal{A}
(\alpha,N)$ are  defined below.
It is crafted in a specific manner to satisfy a few sufficient conditions for a CLT.

$\mathcal{A}
\left(\alpha,N\right)$ includes indices $\eta$ where the corresponding  $X_{\eta}$'s  have relatively ``high'' correlation with $X_\alpha$,
and its complement includes the indices $\eta$ where the corresponding  $X_{\eta}$'s that have relatively ``low'' correlation with $X_\alpha$.
There is substantial freedom in defining these sets, but an easy rule to applying them to (non-sparse) SUGMs is to set the $\mathcal{A}
\left(\alpha,N\right)$ sets to include the other tuples of nodes with which the reference tuple of nodes shares some potential edges and therefore
could face incidental generation.

We show that under conditions on the relative correlations inside and outside of the
 {affinity sets} a central limit theorem applies.

\subsection{The Central Limit Theorem}

Let
\[
a_N := \sum_{\alpha;\eta \in \mathcal{A}
(\alpha,N) } \cov \left(Z_{\alpha}, Z_{\eta}\right),
\]
be the total sum of variance-covariances across all the pairs of variables in each other's  {affinity sets}
,
and recall this was used in normalizing the total sum above.
In what follows, we maintain the assumption that $a_N\rightarrow \infty$, as otherwise there is insufficient variation to obtain a central limit theorem.

Finally, let ${\bf Z}_{-\mathcal{A}(\alpha,N)} := \sum_{\eta \notin \mathcal{A}%\Delta
 (\alpha,N)}Z_\eta$ be the sum over all random variables not in the reference index $\alpha$'s affinity set.

The following are the key conditions for the theorem:
\begin{equation}\label{one}
	\sum_{\alpha; \eta,\gamma \in\mathcal{A}
 (\alpha,N)} \E  \left[\, |Z_\alpha| Z_{\eta} Z_{\gamma} \right]= o\left(a_N^{3/2}\right),
\end{equation}
\begin{equation}\label{two}
	\sum_{\alpha,\alpha',\eta\in\mathcal{A}
 \left(\alpha,N\right),\eta'\in\mathcal{A}
 \left(\alpha',N\right)} \cov\left(Z_{\alpha}Z_{\eta},Z_{\alpha'}Z_{\eta'}\right)
	=o\left( a_N^{2}\right),
\end{equation}
\begin{equation}\label{three}
	\sum_{\alpha} \E \left[   \left\vert \E \left[Z_\alpha
{\bf Z}_{-\mathcal{A}(\alpha,N)}
  \left\vert
 {\bf Z}_{-\mathcal{A}(\alpha,N)}
  \right. \right]\right\vert\right]
 =
 \sum_{\alpha} \E \left[   \left\vert {\bf Z}_{-\mathcal{A}(\alpha,N)}\E \left[Z_\alpha
  \left\vert
 {\bf Z}_{-\mathcal{A}(\alpha,N)}
  \right. \right]\right\vert\right]
	= o\left(a_N \right).
\end{equation}

 Condition (\ref{one}) captures the idea that most of the covariance between random variables in
{an affinity set}  $\alpha$ comes from covariances between the reference random variable $X_\alpha$ and one member of the neighborhood $X_\eta$, rather than from covariance between two other members $X_\eta$ and $X_\gamma$.  Some of them can have high covariance,
but in total they cannot. The term on the left-hand-side consists of an integral over $Z_\alpha$ of covariances between $Z_\gamma$ and $Z_\eta$, weighted by $|Z_\alpha|$, so the covariances cannot be too large exactly with large values of $Z_\alpha$. So in constructing our normalizer $a_N$ we need only consider the covariance terms between reference variables and
members of their {affinity sets}.

Condition (\ref{two}) is similar but it looks at the covariance between two members ($\eta, \eta'$) of   different {affinity sets}
 of two distinct reference nodes ($\alpha, \alpha'$). It says, again, that the {\sl total} amount of covariance across members of different {affinity sets},
 when considering any two pairs of reference nodes, is small relative to the total sum of variances.

Condition (\ref{three}) states that covariances between reference nodes and all members outside of its {affinity set}
 are relatively small. This is intuitive and motivates the strategy in defining
 {affinity sets}
  in the first place.
Note that if, for instance, $\E[Z_\alpha {\bf Z}_{-\mathcal{A}(\alpha,N)} \vert {\bf Z}_{-\mathcal{A}(\alpha,N)} ] \geq 0$, then Condition (\ref{three}) is simply that $\sum_{\alpha; \eta \notin \mathcal{A}
(\alpha,N)} \cov \left(Z_\alpha  Z_{\eta}\right)
= o\left(a_N\right)$.

Although these conditions may seem complex, in \cite*{chandrasekhar2023general} we show how they can be directly verified in a number of applications.

\begin{theorem}\label{clt}
	If (\ref{one})-(\ref{three}) are satisfied, then
	${\sum_{\alpha \in \Lambda^N} Z_{\alpha}^N} / {\sqrt{a_N}}\rightsquigarrow\mathcal{N}(0,1)$.
\end{theorem}

It is useful to consider the special case in which $\mathcal{A}
(\alpha,N) = \{\alpha\}$,
which  extends but nests many standard central limit theorems.
Here we use the notation ${\bf Z}_{-\alpha} $ to denote ${\bf Z}_{-\mathcal{A}(\alpha,N)} $
This is useful when we get to the case of sparse networks, where incidental
networks are unlikely and the correlation between different subgraphs becomes small.

\begin{corollary}
	\label{sparseCLT}
	If  $\E\left[Z_\alpha {\bf Z}_{-\alpha} \left\vert {\bf Z}_{-\alpha} \right. \right] \geq 0$ for every
 $\alpha$, and\footnote{If $\E\left[Z_\alpha {\bf Z}_{-\alpha}\left\vert {\bf Z}_{-\alpha} \right. \right] \geq 0$
		does not hold,  then (ii) can just be substituted by (\ref{three}). }
	\begin{itemize}
		\item[(i)] $\sum_{\alpha, \eta}  \cov ( Z_\alpha^2, Z_{\eta}^2 ) =o\left(a_N^{2}\right)$, and
		\item[(ii)]   $\sum_{\alpha \neq \eta}  \cov (Z_\alpha ,Z_{\eta}) = o\left(a_N\right)$,
	\end{itemize}
	then $ {\sum_{\alpha \in \Lambda^N} Z_{\alpha}^N} / {\sqrt{a_N}}\rightsquigarrow\mathcal{N}(0,1)$.

	Moreover, if the $X_\alpha$s are Bernoulli random variables with
	$\E[X_\alpha]\rightarrow 0$ (uniformly), then (ii)  implies (i).
\end{corollary}

Note that (i) is often satisfied whenever (ii) is,  so this is an easy corollary  based on one intuitive condition:
the overall sum of covariances between different variables cannot be too large relative to the sum of their variances.

\medskip

\section{Concluding Remarks}\label{sec:conclusion}

\

We have developed a new class of models---SUGMs---in which networks are formed via a basis set of subgraphs. The parameters are always identified and we study conditions when the parameters have estimators that are consistent and asymptotically normally distributed.  We provide four such estimators to cover various data settings.
En route, we develop a new central limit theorem for dependent random variables which extends the dependency graph literature and also does not require a geometric (lattice-like) ordering of covariances of the kind used in the time series and spatial literatures. We believe this is of independent interest.

Our model is useful for empirical work. We show that it models economically relevant features of
real-world network data  better than the standard alternatives: stochastic block models,
unobserved heterogeneity models, latent space models, and ERGMs. Further, we have illustrated that
it is easy to microfound a SUGM to test important hypotheses such as whether a network provides incentives to sustain informal contracts or whether people are willing to interact across caste publicly.

Future research can explore, among other things, richer inclusion of covariates in subgraphs, a data-driven approach to select subgraphs for inclusion in the model,   statistical properties of other specific empirically-relevant SUGMs not studied here, and systematic bootstrap techniques for inference for use in complex implementations of these models.

Data Availability Statement:  The code and data needed to reproduce all figures and tables in the paper are available at  http://doi.org/10.5281/zenodo.14218442

\bibliographystyle{ecta}
\bibliography{23289}

\newpage
\appendix

\section{
Proofs}\label{sec:proofs}

\begin{proof}[{\bf Proof of Lemma \ref{lem:favor-info}}]
	Note that for $z \in \{favor,info\}$,
	\(
	\frac{\gamma_{T,z}}{\gamma_{L,z}} = \frac{\beta_{T,z}^{1/3}}{\beta_{L,z}^{1/2}}
	\)
	and so the condition becomes
	\(
	\frac{\beta_{T,favor}^{1/3}}{\beta_{L,favor}^{1/2}} >  \frac{\beta_{T,info}^{1/3}}{\beta_{L,info}^{1/2}}
	\)
	from which the result directly follows.
\end{proof}

\

\begin{proof}[{\bf Proof of Lemma \ref{lem:geomodel}}]
	Having two randomly picked nodes bump into each other within a community, there is a $f^2+(1-f)^2$ probability of the nodes being of the same type, and a $1-(f^2+(1-f)^2)$ probability of them being of different types.\footnote{ To keep things simple, we consider equal-sized groups, but the argument extends with some adjustments to asymmetric sizes.}     Thus, the relative meeting frequency of different type links compared same type links is
	\[
	\frac{\pi_L(diff)}{\pi_L(same)}=\frac{1-(f^2+(1-f)^2)}{f^2+(1-f)^2}.
	\]
	For triangles, picking three individuals out of the community at any point in time would lead to a $f^3+(1-f)^3$ probability that all three are of the same type, and $1-(f^2+(1-f)^2)$ of them being of mixed types,
	and so
	\[
	\frac{\pi_T(diff)}{\pi_T(same)}=\frac{1-(f^3+(1-f)^3)}{f^3+(1-f)^3}.
	\]
	It follows directly that for $f\in (0,1)$:
	\begin{equation}
		\label{ratios}
		\frac{\pi_T(same)}{\pi_T(diff)}<\frac{\pi_L(same)}{\pi_L(diff)}.
	\end{equation}
	So different type triangles are more likely to have opportunities to form under this random mixing model than different type links.
	In particular, note that
		$\frac{p_{T,diff}}{p_{T,same}} < \frac{p_{L,diff}}{p_{L,same}}$ if and only if $\left( \frac{\beta_{T,diff} }{\beta_{T,same}}  	\frac{\pi_T(diff)}{\pi_T(same)}  \right) ^{1/3}<\left(\frac{\beta_{L,diff} }{\beta_{L,same} }   \frac{\pi_L(same)}{\pi_L(diff)}  \right)^{1/2}.$ 	In summary, given (\ref{ratios}), sufficient condition for
	$\frac{p_{T,diff}}{p_{T,same}} < \frac{p_{L,diff}}{p_{L,same}}$ is that
	$\frac{\beta_{T,diff} }{\beta_{T,same}}<\left(\frac{\beta_{L,diff} }{\beta_{L,same} }\right)^{3/2}.$
\end{proof}

\

\begin{proof}[{\bf Proof of Theorem \ref{newid}}]
Order subgraph types so that the number of links in a subgraph of type $\ell$ is nondecreasing in $\ell$. Let $\ell^*$ be the smallest $\ell$ for which $\beta_\ell\neq \beta_\ell'$.

Consider a particular subgraph $g'\in G_{\ell^*}$ with labeled nodes.   Let $p_\beta(g')$ denote the probability that the subgraph $g'$ (without any extra links within the subgraph nor any links anywhere else in the network) forms from some collection of subgraphs in $G_\ell$ for $\ell<\ell^*$.
We can then write the probability of forming the subgraph $g'$ (and no other links anywhere) as
\[
p_\beta(g')  + (1 - p_\beta(g') )  \beta_{\ell^*},
\]
where recall that $\beta_{\ell^*} $ is the probability that $g'$ forms directly.
Let $no_\beta(g')$ denote the probability that all $g'' \in G_\ell$ for $\ell<\ell^*$ such that  $g''\subset g'$ {\sl do not} form.
Then the probability that none of the links in $g'$ are present as parts of subgraphs that do not extend beyond $g'$ is then
\[
no_\beta(g')(1- \beta_{\ell^*}).
\]
Let $\emptyset$ denote the empty network.
It then follows that
\[
\frac{\Pr_\beta (g')}{\Pr_\beta (\emptyset) } =   \frac{ p_\beta(g')  + (1 - p_\beta(g') )  \beta_{\ell^*}}{no_\beta(g')(1- \beta_{\ell^*}) }.
\]
So the probability that the realized network is exactly $g'$ compared to the probability that it is the empty network (which has positive probability given that each $\beta_\ell<1$), depends only on the probability that $g'$ forms directly or incidentally from subgraphs of it,
over the probability that no subgraph of $g'$ (including itself) forms.

Note that this expression is strictly increasing in $ \beta_{\ell^*}$ since  $p_\beta(g')<1$ and $no_\beta(g')>0$.
By the definition of $\ell^*$:   $p_\beta(g') = p_{\beta'}(g')<1$ and $no_\beta(g')=no_{\beta'}(g')$.
It then follows that
\[
\frac{\Pr_\beta (g')}{\Pr_\beta (\emptyset) }
\neq
\frac{\Pr_{\beta'} (g')}{\Pr_{\beta'} (\emptyset) },
\]
which establishes the claim.\end{proof}

\

\begin{proof}[{\bf Proof of Proposition \ref{prop:Links-Triangles-SUGM}}]
First, note that $ 1- (1-\beta_{T})^{x}$ is the probability that some link is formed as part of at least one triangle out of $x$ possible triangles that could have it as an edge (independently of whether it also forms directly).

Next, note that the probability that a link forms conditional on some particular triangle that it could be a part of {\sl not forming} is\footnote{That is,
consider a given pair of nodes $i,j$ and a third node $k$. Consider the probability that link $ij$ is formed conditional on triangle $ijk$ not forming directly as a triangle.}
\begin{equation}
\label{tildeq}
\widetilde{q}_L\left(\beta_{L},\beta_{T}\right) =   \beta_L  + (1 -\beta_L )\left(1- (1-\beta_{T})^{n-3}\right).
\end{equation}
To see the derivation, note that:
\begin{align*}
\widetilde{q}_L \left(\beta_{L},\beta_{T}\right)=\Pr(g_{ij=1} \mid T_{ijk}=0) &= \frac{\Pr(g_{ij}=1 \cap T_{ijk}=0)}{\Pr(T_{ijk}=0)} \\
&= \frac{\Pr(g_{ij}=1 \cap T_{ijk}=0)}{1-\beta_T} \text{ by definition }\\
&= \frac{\underbrace{\beta_L(1-\beta_T)}_{(i)}+\underbrace{(1-\beta_L)(1-(1-\beta_T)^{n-3})(1-\beta_T)}_{(ii)}}{1-\beta_T}\\
\end{align*}
 where $T_{ijk}$ is an indicator for the direct triangle on $i,j,k$ forming. Here $(i)$ captures when the link forms directly and the indicated direct triangle does not and $(ii)$ captures when the link does not form directly and some other triangle generates it and the indicated directed triangle does not.

Given this, note that the probability that a link forms can be written as
\begin{equation}
\label{ELdef2}
\E_\beta(S_L(g)) = q_L= \beta_T + (1-\beta_T) \widetilde{q}_L\left(\beta_{L},\beta_{T}\right),
\end{equation}
noting that a link could form as part of a triangle that it is part of, or else form conditional upon that triangle not forming. Note that this is equivalent to saying a link can be formed as part of some triangle, or else that it could form by other means that exclude that triangle (but then are either direct or involve other triangles, which is the $\widetilde{q}_L\left(\beta_{L},\beta_{T}\right)$ .

We can write the probability of some triangle forming as
\begin{equation}
\label{ETdef}
\E_\beta(S_T(g)) = q_T = \beta_T+ (1-\beta_T) (\widetilde{q}_L\left(\beta_{L},\beta_{T}\right))^3,
\end{equation}
where the first expression $\beta_T$ is the probability that the triangle is directly generated, and then the second expression $(1 -\beta_T )(\widetilde{q}_L\left(\beta_{L},\beta_{T}\right))^3$ is the probability that it was not generated directly, but instead all three of the edges formed on their own (which happen independently, conditional on the triangle not forming, which has
probability $(\widetilde{q}_L\left(\beta_{L},\beta_{T}\right))^3$). The result follows from Lemma \ref{lem:Unique}, with $x_{1}=\beta_{L}$,
$x_{2}=\beta_{T}$, $q_{L}=a_{1}\left(x\right)$, $q_{T}=a_{2}\left(x\right)$
and $\widetilde{q}_{L}\left(\beta_{L},\beta_{T}\right)=f\left(x\right)$.\end{proof}

\

\begin{lem}
\label{lem:Unique}Let $x=\left(x_{1},x_{2}\right)\in\left[0,1\right)^{2}$
and $a\left(x\right)=\left(a_{1}\left(x\right),a_{2}\left(x\right)\right)$
be two real-valued functions
\begin{eqnarray*}
a_{1}\left(x\right) & = & x_{2} + \left(1-x_2\right) f\left(x\right)\\
a_{2}\left(x\right) & = & x_{2}+ \left(1-x_2\right)f\left(x\right)^{3},
\end{eqnarray*}
with
\[
f\left(x\right)=x_{1}+\left(1-x_{1}\right)\left[1-\left(1-x_{2}\right)^{N}\right] =
 1- (1-x_1)\left(1-x_{2}\right)^{N}
\]
for some integer $N\geq 0$. Then
\(
x\neq x'  \implies a\left(x\right)\neq a\left(x'\right).
\)
\end{lem}

\begin{proof}[{\bf Proof}]
Suppose the contrary.
Then
\[
x_{2}' + \left(1-x_2'\right) f\left(x'\right)  = x_{2} + \left(1-x_2\right) f\left(x\right) \text{ and }
x_{2}' + \left(1-x_2'\right) f\left(x'\right)^{3} =  x_{2}+ \left(1-x_2\right)f\left(x\right)^{3}.
\]

First, note that if $x_2'=x_2$, then since these are both less than one, the first equation above implies that $f(x')=f(x)$.  However, that is not possible since  $f$ is increasing in $x_1$ and $x_1'\neq x_1$  - recalling
that $x'\neq x$ and  $x_2'=x_2$ .
Thus, $x_2'\neq x_2$, and so without loss of generality consider the case in which $x_2'<x_2$.
This implies that both
\[
f\left(x'\right) = b f\left(x\right) + c
\]
and
\[
f\left(x'\right)^3 = b f\left(x\right)^3 + c,
\]
where
$b=  \frac{1-x_2 }{1-x_2'}\in (0,1)$ and $c=  \frac{ x_2-x_2'}{1-x_2'}\in (0,1)$, and $b+c=1$.

This implies that
\[
 b f\left(x\right)^3 + 1-b
 =(b f\left(x\right) + 1-b)^3.
 \]
 This as an equation of the form
 \[b y^3 + 1-b = (by +1-b)^3\]
where $b\in (0,1)$ and $y\in [0,1)$.
Note that the left hand side is larger when $y=0$ and the two are equal when $y=1$, and that the derivative of the difference is
\[
3by^2 - 3b (b y + 1-b)^2 =   3b \left[y^2 -  (b y+ 1-b)^2\right] <0.
\]
 The difference is decreasing over the entire interval, and hits 0 at the end.
Thus, the difference is always positive in $[0,1)$ and there is no solution, meaning  our supposition was incorrect.
\end{proof}

\

\begin{lem}\label{lem:analytic}
Any event (in the discrete $\sigma$-algebra generated by all possible realizations of all subgraphs) associated with any SUGM has a probability that is an analytic function (and so it is in $C^\infty$), and has derivatives and cross partials at all levels being uniformly continuous and bounded on the whole parameter space of $[0,1]^k$.
\end{lem}

\begin{proof}[{\bf Proof}]
An `outcome' is a specification of exactly which subgraphs form and which do not - so a complete specification of what happens.
Any event then corresponds to a set of outcomes, and so its probability is a sum of probabilities of the outcomes. Each outcome's probability is of the form
\[
\prod_\ell \beta_\ell^{z_\ell} (1-\beta_\ell)^{m_\ell-z_\ell}
\]
where $z_\ell$ indicates how many subgraphs of type $\ell$ are present in the outcome.
As each of these functions is analytic (and hence in $C^\infty$), all of the derivatives and partials, cross partials, etc., are continuous and bounded on $[0,1]^k$ and hence uniformly continuous on $[0,1]^k$.
Any event is then a finite sum of analytic functions and so the result follows directly.
\end{proof}

\

\begin{proof}[{\bf Proof of Theorem \ref{thm:many_networks}}]
	We verify the conditions of Theorem 2.5 of \cite{newey1994large} for consistency.
	Assumption (i)  holds by Theorem \ref{newid} and we assume compactness of the parameter space
	(Assumption (ii)). Continuity of $\log \Pr_\beta(g)$ at each $\beta$ with probability one is
	mechanical in our model since subgraph probabilities  are analytic functions of the parameters (Lemma \ref{lem:analytic}). Finally, the uniform bound of assumption (iv) holds since $n$ is fixed, there are only a finite number of graphs in consideration, each with assigned probabilities in a compact set of parameters, and there is a positive probability of seeing any graph in $\mathcal{G}^n$. Therefore, the supremum must be finite.
	
 We verify the conditions of Theorem 3.3 of \cite{newey1994large} for asymptotic normality. We have assumed (i), interiority of the parameter, and our model by construction places positive mass on all of $\mathcal{G}^n$.  We have assumed (iv). Lemma \ref{lem:analytic} implies Assumptions (ii), (iii), and (v). Because all events have probabilities that are analytic functions of parameters, with all derivatives and cross-partials being uniformly continuous and bounded in the parameter space, the norms of the maximal derivative ($\lVert \nabla_{\beta}\Pr_\beta(g) \rVert $) and second derivatives ($\lVert \nabla_{\beta \beta}\Pr_\beta(g) \rVert $) of the probability functions, as well as the log likelihood ($\lVert \nabla_{\beta}\log \Pr_\beta(g)\rVert $), have uniform and finite upper bounds.
\end{proof}

\

\begin{proof}[{\bf Proof of Proposition \ref{prop:many_networks_LT}}]
First we check consistency by the conditions of Theorem 2.6 of \cite{newey1994large}.  Here each observation is
an independently drawn network.
For Assumption (i) let $\widehat{W}$ be the identity matrix and then apply Proposition \ref{prop:Links-Triangles-SUGM}.  For (ii), we have assumed that the parameter space $\mathcal{B}$ is compact. (iii) follows from the fact that   $\E_{\beta}\left[S\left(g_{r}\right) \right]$ is continuous
at each $\beta$ with probability one since it composes continuous functions of parameter entries. Finally (iv) follows from the fact that since  both $S_{\ell}$   are shares, they
are strictly less than 1.

Next we check asymptotic normality by the conditions of Theorem 3.4
of \cite{newey1994large}.\footnote{See also \cite{greene2000econometric} for the simplification of the variance term in the case in which there is just-identification as we have.}

Since $\beta_0$ is in the interior of the compact
parameter space, so (i) is met. We see (iii)  holds since by
definition the subgraph counts are fractions between 0 and 1.
Both (ii) (that the empirical moment function is continuously
differentiable in a neighborhood of the true parameter) and
(iv) (that the gradient of the moment function is continuous at the
true parameter and that it satisfies a ULLN) follow from Lemma \ref{lem:analytic}.
Analytic functions are $C^\infty$, so there are arbitrarily many derivatives.  The bounds follow the expressions in (\ref{dql}) and following, which provide the rows of $H$, and the expressions are bounded  in magnitude by $3(n-3)+1$ (and note the $n$ is fixed and it is $R$ that is growing).
Finally, for (v), that $HH'$ is non-singular follows
from the linear
independence of rows of $H$ (see the expressions in (\ref{dql}) and following, which provide the rows of $H$, which dividing through by $(1-\beta_T)^{n-2}$ are clearly independent for interior parameters).
\end{proof}

\

\begin{proof}[{\bf Proof of Theorem  \ref{thm:Largenetwork}}]
	When obvious, we omit superscript $n$'s to simplify notation, but they are implicit.
	It follows that,
	\begin{equation}
		\label{conway}
  \betanDCell
		=\left( \frac{{S}^{true}_\ell}{\kappa_\ell \binom{n}{m_\ell}} +  \frac{\widetilde{S}^{true}_\ell- {S}^{true}_\ell}{\kappa_\ell \binom{n}{m_\ell}} +  \frac{\widetilde{S}_\ell(g) - \widetilde{S}^{true}_\ell}{\kappa_\ell \binom{n}{m_\ell}} \right)
	\end{equation}
	where ${S}^{true}_\ell$ is the number of truly generated such subgraphs (unobserved) on the whole network,
	and $\widetilde{S}^{true}_\ell$ is the number of truly generated such subgraphs (unobserved) on the networks that the after removing the links in $D_\ell(g)= \{ ij : \  ij\in g', g'\subset g,  g'\in G_{\ell'},  \ell'<\ell\}$, and $\binom{n}{m_\ell}$ counts the number of ways to pick $m_\ell$ nodes out of $n$.
	
	We show below that $|\widetilde{S}^{true}_\ell- {S}^{true}_\ell |= o_p( {S}^{true}_\ell)$ and  $|\widetilde{S}_\ell(g)- \widetilde{S}^{true}_\ell |=o_p( \widetilde{S}^{true}_\ell) $;
	which then also implies that $\widetilde{S}_\ell(g)- \widetilde{S}^{true}_\ell = o_p( {S}^{true}_\ell)$.
	Together with (\ref{conway}), these tell us that
	\begin{equation}
		\label{approx}
\betanDCell
  = \left( \frac{{S}^{true}_\ell}{\kappa_\ell \binom{n}{m_\ell}}  \right)(1+o_p(1)).
	\end{equation}
	
	Note that ${S}^{true}_\ell(g)$ has a binomial distribution with parameter $\beta^n_{0,\ell}$.
	From this and (\ref{approx}), it then follows that
	\[
	\frac{
 \betanDCell- \beta^n_{0,\ell}}{\sigma_{n,\ell}} \wkto \mathcal{N}(0, 1)
	\]
	where
	$\sigma_{n,\ell} =\left(\frac{\beta^n_{0,\ell}(1-\beta^n_{0,\ell})}{ \kappa_\ell \binom{n}{m_\ell}}\right)^{1/2}$. {This is because $\frac{{S}^{true}_\ell}{\kappa_\ell \binom{n}{m_\ell}}$ is a self-normalized sum of $\kappa_\ell \binom{n}{m_\ell}$ Bernoulli($\beta_{0,\ell}^n$)  independent random variables with standard deviation $\sigma_{n,\ell}$ by definition. Convergence in distribution in the sense of the Central Limit Theorem is a consequence of our own Corollary \ref{sparseCLT}, (though one could appeal to other CLTs on triangular arrays for independent random variables). So the result follows from Slutsky's Theorem since $1+o_p(1) \cvgto 1$.}

Next, note that the ${S}^{true}_\ell(g)$  are independent across $\ell$.
From  (\ref{approx}) it then follows that
\[
\sum_\ell  \alpha_{\ell}
\betanDCell
=
\sum_\ell  \alpha_{\ell}  \frac{{S}^{true}_\ell}{\kappa_\ell \binom{n}{m_\ell}} (1+o_p(1))
\]
for any  $\alpha \in [0,1]^k$, with $\sum_\ell  \alpha_{\ell}=1$.
Given the independence of ${S}^{true}_\ell$  across $\ell$,  it then follows that the random variable on the right hand side converges to being normal.\footnote{Note that under the assumption that $m_\ell> h_\ell$ there are a growing number
of observations of each subgraph.}
Then, by the Cram\'er-Wold Theorem, this implies that the
$
\betanDCell
$ are jointly normally distributed in the limit,
and so
\[
\Sigma_n^{-1/2}(\betanDC - \beta^n_0) \wkto \mathcal{N}(0, I)
\]
where $\Sigma_{n,\ell \ell} = \frac{\beta^n_{0,\ell}(1-\beta^n_{0,\ell})}{ \kappa_\ell \binom{n}{m_\ell}}$
and the off-diagonals are all 0.

Thus, to complete the proof we show that
$|\widetilde{S}^{true}_\ell- {S}^{true}_\ell |= o_p( {S}^{true}_\ell)$ and  $|\widetilde{S}_\ell(g)- \widetilde{S}^{true}_\ell | =o_p( \widetilde{S}^{true}_\ell) $.
	
To establish these claims, we establish two facts.  One is that the probability that some observed subgraph of type $\ell$ was incidentally generated (by subgraphs that are no larger than it in the ordering) is  $o(1)$.
This establishes that  $|\widetilde{S}_\ell(g)- \widetilde{S}^{true}_\ell | =o_p( \widetilde{S}^{true}_\ell) $.
The other is that the probability that a truly formed subgraph of type $\ell$ becomes part of an incidentally generated subgraph of type $\ell'<\ell$ is $o(1)$.
This establishes that $|\widetilde{S}^{true}_\ell- {S}^{true}_\ell |= o_p( {S}^{true}_\ell)$.
	
Let $z_{\ell}^{n}$  denote
the probability that any given $g' \in G^n_\ell$ is incidentally generated.
We now show that $z^n_\ell/\beta^n_{0,\ell} = o(1)$, which establishes the first claim.
Consider $g_{\ell}\in G_{\ell}^{n}$ and a (minimal, ordered) generating
subclass $\mathcal{C}=\left(\ell_{j},c_{j}\right)_{j\in J}$, and for which $\ell_j\geq \ell$ for all $j$.
	
We show that the probability $z_{\ell}^{n}$ that it is generated
by this subclass goes to zero relative to $\beta_{0,\ell}^{n}$,
and since there are at most $M_{\ell}\leq k^{m_{\ell}}$ such generating
classes, this implies that $z_{\ell}^{n}/\beta_{0,\ell}^{n}\rightarrow0$.
	
Consider a subnetwork in $G_{\ell_{j}}^{n}$. The probability of getting
at least one such network that has the $c_{j}$ nodes out of the $m_{\ell}$
in $g_{\ell}$ is no more than
\[
\kappa_{\ell_{j}}\binom{n}{m_{\ell_{j}}-c_{j}}\beta_{0,\ell_j}^{n}\leq\kappa_{\ell_{j}}n^{m_{\ell_{j}}-c_{j}}\beta_{0,\ell_j}^{n}.
\]
Then, we can bound the desired ratio by				
\begin{align*}
\frac{z_{\ell}^{n}}{\beta_{0,\ell}^{n}}  \leq  \frac{\prod_{j\in J}n^{m_{\ell_{j}}-c_{j}}\kappa_{\ell_{j}}\beta_{0,\ell_j}^{n}}{\beta_{0,\ell}^{n}}\leq
  \frac{n^{\sum_{j\in J} m_{\ell_{j}}- h_{\ell_{j}}- c_{j}}\prod_{j\in J}\kappa_{\ell_{j}}b_{0,\ell_j}}{n^{-h_\ell}b_{0,\ell}}\rightarrow0,
\end{align*}
where the last convergence is guaranteed by (\ref{ln2}).
	
The second claim follows from a similar calculation:  It is sufficient to show that the probability that some subgraph of type $\ell_{j'}$ becomes part of a subgraph of type $\ell< \ell_{j'}$ (where
$j'\in J$ is part of a generating class of some $\ell<\ell_{j'}$), compared to the likelihood of the formation of a subgraph of type $\ell_{j'}$, is of vanishing order.
Again, as there are a finite number of larger subgraphs, and a finite number of generating classes, it is sufficient to show this for a generic $\ell< \ell_{j'}$ and generic generating class.
In the following, the numerator is on the order of the expected number of incidentally formed subgraphs of type $\ell$ from this type of generating class, while the denominator is the expected number of
the subgraphs of type $\ell$.\footnote{We use Bachmann-Landau notation so $f(n)=\Theta(g(n))$ means that $f$ is bounded above and below asymptotically by $g$. That is, $\exists k_1>0, \exists k_2>0, \exists n_0$ such that $\forall n>n_0$, $k_1 g(n) \leq f(n) \leq k_2 g(n)$.}
\begin{align*}
\frac{\kappa_\ell\binom{n}{m_\ell}\prod_{j\in J}n^{m_{\ell_{j}} - c_{\ell_{j}}}\kappa_{\ell_{j}}\beta_{0,\ell_j}^{n}}{\kappa_{\ell_{j'}} \binom{n}{m_{\ell_{j'}}} \beta_{0,\ell_{j'}}^n}
		= \Theta\left( \frac{n^{m_\ell} n^{\sum_{j\in J} m_{\ell_{j}}-c_{\ell_{j}}-h_{\ell_{j}}}}{n^{m_{\ell_{j'}}-h_{\ell_{j'}}}}\right)\rightarrow0,
\end{align*}
where the convergence to 0 follows from (\ref{ln3}).

Finally, by multiplying and dividing by $n^{h_\ell}$ and collecting terms, 	it follows that
$|\bnDC - b_0| \cvgto 0$ and
$V^{-1/2}_n  \left(\bnDC-b_0\right)\wkto\mathcal{N}\left(0,I \right)$. To see this, observe
that
$\Sigma_n^{-1/2}(\betanDC - \beta^n_0) = V^{-1/2}_n  \left(\bnDC-b_0\right) $.
\end{proof}

\

\begin{proof}[{\bf Proof of Corollary  \ref{cor:Largenetwork}}]
Note that $\sum_j c_j \geq m_\ell + (|C| - 1) z$ for some  $z\geq 1$, where $z\geq 2$ if subgraphs are acyclic   (each subgraph in the incidental set overlaps the others with at least one node, and at least
two if the subgraphs are acyclic).  The conditions then simplify directly.
\end{proof}

\section{Proof of Central Limit Theorem \ref{clt} and Corollary \ref{sparseCLT} }\label{sec:normality}

\subsection{Stein's Lemma}

Our proof uses a lemma from \cite{stein1986approximate}. We review it here, both to be self-contained and also to explain why this approach to proving asymptotic normality is useful and distinct from other approaches in the networks literature.
The key observation of \cite{stein1986approximate} is that if a random variable satisfies
\[\E [f'(Y) - Y f(Y)] = 0\]
for every $f(\cdot)$ that is  continuously differentiable, then it must have a standard normal distribution.

This observation leads to a useful lemma, that allows one to characterize the Kolmogorov distance between a random variable $Y$ and a
standard normally distributed $Q$, denoted $d_K(Y,Q)$. We can bound this from
above by (a constant times) the Wasserstein distance,
$d_W(Y,Q)$, which itself
is bounded by the below expression. Convergence in Wasserstein distance implies convergence in distribution. Let $\left\Vert f \right\Vert$ denote the sup norm over the domain of $f$.

\begin{lem}[\cite{stein1986approximate,ross2011fundamentals}]\label{lem:Stein}
	If $Y$ is a random variable and $Q$ has the standard normal distribution, then
	\[
	d_W(Y,Q) \leq  \sup_{\{f :  ||f||, ||f''||\leq  2, ||f'||\leq \sqrt{2/\pi} \}
	}
	\left|
	\E [f'(Y) - Y f(Y)]\right|.
	\]
	Further
	\(
	d_K(Y,Q)\leq (2/\pi)^{1/4} (d_W(Y,Q))^{1/2}.
	\)
\end{lem}

 Define
 \[
 {\bf Z}^N :=\sum_{\alpha \in \Lambda^N} Z_\alpha^N \text{ and } \overline{{\bf Z}}^{N}={\bf Z}^{N}/a_{N}^{1/2}.
 \]
	
	For ease of notation, we omit the superscript $N$s below. 	
	Recall that
\[ {\bf Z}_{-\mathcal{A}(\alpha,N)} =  \sum_{\eta \notin \mathcal{A}%\Delta
 (\alpha,N)}Z_\eta
\]
and let
\[
\bar{{\bf Z}}_{-\mathcal{A}(\alpha,N)} := {\bf Z}_{-\mathcal{A}(\alpha,N)}/a^{1/2}.
\]

By this lemma, if we show that a normalized sum of random variables satisfies
\[
\sup_{\{f :  ||f||, ||f''||\leq  2, ||f'||\leq \sqrt{2/\pi} \}}
\left|
\E [f'(\overline{{\bf Z}}^{N}) - \overline{{\bf Z}}^{N} f(\overline{{\bf Z}}^{N})]\right| \rightarrow 0,
\]
then $d_W(\overline{{\bf Z}}^{N},Q)\rightarrow 0$, and so it must be asymptotically normally distributed.

\subsection{Proofs of Theorem \ref{clt} and Corollary \ref{sparseCLT}}

The following lemmas are useful in the proof.

\begin{lem} \label{lem:h_max}
	A solution to $\max_h \E[Zh(Y)]   \mbox{ s.t. }  |h(\cdot)|\leq 1$ (so the absolute value of $h$ is bounded by 1, where $h$ is measurable) is $h(Y) = \sign( \E[Z|Y] ),$
	where we break ties, setting $\sign( \E[Z|Y] )=1$ when $ \E[Z|Y] =0$.
\end{lem}

\begin{proof}[{\bf Proof}]  This can be seen from direct calculation:
	\[
	\E[Zh(Y)]  = \int_Y   \E[Z|Y]h(Y) d\Prob(Y)
	\]
	Maximizing \( \E[Z|Y]h(Y) \) pointwise when \( |h|\leq 1 \) is achieved by setting  \( h(Y)= {\rm sign} ( \E[Z|Y] ) \),
	and we
	break ties by setting $\sign( \E[Z|Y] )=1$ when $ \E[Z|Y] =0$, as that makes no difference in the integral.
\end{proof}

\

\begin{lem} \label{lem:XY_h}\( \E[XYh(Y)] \)  when \( h(\cdot) \) is measurable and bounded by \( \sqrt{\frac{2}{\pi}} \)  satisfies
	\[ \E[XYh(Y)] \leq   \sqrt{\frac{2}{\pi}} \E\left[XY \cdot \sign(\E[X|Y] Y)  \right].\]
\end{lem}
\begin{proof}[{\bf Proof}] This follows from Lemma \ref{lem:h_max}, setting $Z = XY$.
\end{proof}

\

\begin{proof}[{\bf Proof of Theorem \ref{clt}}]
	By Lemma \ref{lem:Stein}, it is sufficient to show that the appropriate sequence of random variables $\overline{{\bf Z}}^{N}$ satisfies
	\[
	\sup_{\{f :  ||f||, ||f''||\leq  2, ||f'||\leq \sqrt{2/\pi} \}
	}
	\left|
	\E [f'(\overline{{\bf Z}}^{N}) - \overline{{\bf Z}}^{N} f(\overline{{\bf Z}}^{N})]\right| \rightarrow 0.
	\]

	Observe that
	\begin{align*}
		\E\left[\overline{{\bf Z}} f\left( \overline{{\bf Z}} \right)\right] & =  \E\left[\frac{1}{a^{1/2}}\sum_{\alpha}Z_{\alpha}\cdot f\left( \overline{{\bf Z}}  \right)\right]  \\
		&= \E\left[\frac{1}{a^{1/2}}\sum_{\alpha}Z_{\alpha}\left(f\left( \overline{{\bf Z}}\right)-f\left(\bar{{\bf Z}}_{-\mathcal{A}(\alpha,N)} \right)\right)\right] \\
		&+ \E\left[\frac{1}{a^{1/2}}\sum_{\alpha}Z_{\alpha}\cdot f\left( \bar{{\bf Z}}_{-\mathcal{A}(\alpha,N)} \right)\right].
	\end{align*}

	The first step is to show that
	\[
	\left\vert \E\left[\frac{1}{a^{1/2}}\sum_{\alpha}Z_{\alpha}\cdot f\left( \bar{{\bf Z}}_{-\mathcal{A}(\alpha,N)} \right)\right] \right\vert = o(1),
	\]
	by employing condition \eqref{three}.

	In order to do this, we can expand the term to
	\begin{align*}
		\left|{\rm E}\left[\frac{1}{a_{N}^{1/2}}\sum_{\alpha\in\Lambda}Z_{\alpha}\cdot f\left(    \bar{{\bf Z}}_{-\mathcal{A}(\alpha,N)}   \right)\right]\right| & =  \left|{\rm E}\left[\frac{1}{a_{N}^{1/2}}
		\sum_{\alpha\in\Lambda}Z_{\alpha}\cdot f\left(\frac{1}{a_{N}^{1/2}}\sum_{\eta\notin\mathcal{A}
  \left(\alpha,N\right)}Z_{\eta}\right)\right]\right|\\
		& \leq  \underbrace{\left|{\rm E}\left[\frac{1}{a_{N}^{1/2}}\sum_{\alpha\in\Lambda}Z_{\alpha}\cdot f\left(0\right)\right]\right|}_{=0\mbox{ since }{\rm E}\left[Z_{\alpha}\right]=0.}\\
		&   +\left|{\rm E}\left[\frac{1}{a_{N}^{1/2}}\sum_{\alpha\in\Lambda}Z_{\alpha}\cdot\left(\frac{1}{a_{N}^{1/2}}\sum_{\eta\notin\mathcal{A}
  \left(\alpha,N\right)}Z_{\eta}\right)\cdot f'\left( \widehat{\bar{{\bf Z}}}_{-\mathcal{A}(\alpha,N)} \right)\right]\right|
	\end{align*}
	where $\widehat{\bar{{\bf Z}}}_{-\mathcal{A}(\alpha,N)}$ is an intermediate value between
	$\bar{{\bf Z}}_{-\mathcal{A}(\alpha,N)}$ and $0$.

	To bound the second term, we apply Lemma \ref{lem:XY_h} to conclude that
	\begin{align*}
	\left|\frac{{\rm E}\left[\sum_{\alpha\in\Lambda;\eta\notin\mathcal{A}
 \left(\alpha,N\right)}Z_{\alpha}Z_{\eta} f'\left(  \widehat{\bar{{\bf Z}}}_{-\mathcal{A}(\alpha,N)}   \right)\right]}{a_{N}}\right|
	\leq
	\sqrt{\frac{2}{\pi}}\left|\frac{{\rm E}\left[\sum_{\alpha\in\Lambda;\eta\notin\mathcal{A}
 \left(\alpha,N\right)}Z_{\alpha}Z_{\eta}\cdot \sign\left({\rm E}\left[Z_{\alpha}Z_{\eta}| \ Z_{\eta}\right]\right)\right] }{a_{N}}\right|.
	\end{align*}
	Thus, it is sufficient that
	\begin{equation}
		\label{threeprime}
		{\rm E}\left[\sum_{\alpha\in\Lambda;\eta\notin\mathcal{A}
  \left(\alpha,N\right)}Z_{\alpha}Z_{\eta}\cdot \sign\left({\rm E}\left[Z_{\alpha}Z_{\eta}| Z_{\eta}\right]\right)\right] =o(a_{N})
	\end{equation}
 or
 $$
 {\rm E}\left[\sum_{\alpha\in\Lambda;\eta\notin\mathcal{A}
  \left(\alpha,N\right)}\left\vert{\rm E}\left[Z_{\alpha}Z_{\eta}| Z_{\eta}\right]\right\vert\right] =o(a_{N})
	$$
	to ensure that
	\[
	\left|\frac{\E\left[\sum_{\alpha\in\Lambda;\eta\notin\mathcal{A}
 \left(\alpha,N\right)}Z_{\alpha}\cdot Z_{\eta}\cdot f'\left( \widehat{\bar{{\bf Z}}}_{-\mathcal{A}(\alpha,N)}  \right)\right]}{a_{N}}\right| = o(1),
	\]
	which is ensured by (\ref{three}) (noting that $ \widehat{\bar{{\bf Z}}}_{-\mathcal{A}(\alpha,N)} $ is a function of ${\bf Z}_{-\mathcal{A}(\alpha,N))}$).

	Next, the second step of the proof
	is to apply a similar reasoning as in \cite*{ross2011fundamentals} with an $o(1)$ adjustment (from the first step above), to write
	\begin{align*}
		\left| \E \left[f'(\overline{{\bf Z}}) - \overline{{\bf Z}} f(\overline{{\bf Z}})\right]\right|
		& \leq   \left| \E  \left[ \frac{1}{a^{1/2}} \sum_{\alpha} Z_\alpha (f(\overline{{\bf Z}}) - f( \bar{{\bf Z}}_{-\mathcal{A}(\alpha,N)}  ) - (\overline{{\bf Z}}- \bar{{\bf Z}}_{-\mathcal{A}(\alpha,N)}) f'(\overline{{\bf Z}})\right]\right|\\
		& + \left| \E  \left[ f'(\overline{{\bf Z}} ) \left( 1- \frac{1}{a^{1/2}}\sum_{\alpha} Z_\alpha (\overline{{\bf Z}} -  \bar{{\bf Z}}_{-\mathcal{A}(\alpha,N)})\right)\right]\right| + o(1),
	\end{align*}
	and then to show that the right hand side of this expression goes to 0.
	
	By a Taylor series approximation and given the bound on the derivatives of $f$,
	it follows that
	\begin{align*}
		\left| \E \left[f'(\overline{{\bf Z}} ) -  \overline{{\bf Z}} f(\overline{{\bf Z}})\right]\right|
		& \leq    \frac{||f''||}{2a^{1/2}} \sum_{\alpha} \E  \left[ \left|Z_\alpha \right| \left( \overline{{\bf Z}} - \bar{{\bf Z}}_{-\mathcal{A}(\alpha,N)}\right)^2\right] \\ &+ \left| \E  \left[ f'(\overline{{\bf Z}} ) \left( 1- \frac{1}{a^{1/2}}\sum_{\alpha} Z_\alpha (\overline{{\bf Z}} - \bar{{\bf Z}}_{-\mathcal{A}(\alpha,N)})\right)\right]\right| + o(1).
	\end{align*}
	
	Let us denote the first two terms on the right hand side as $A_1$ and $A_2$ respectively.
	We bound each, and show that each is $o(1)$, which then completes the proof. 	
	\begin{align*}
		A_1
		&=    \frac{||f''||}{2a^{3/2}} \sum_{\alpha} \E  \left[ \left|Z_\alpha\right| \left(\sum_{\eta\in \mathcal{A}
  (\alpha,N)} Z_\eta\right)^2\right]
		= \frac{||f''||}{2a^{3/2}} \sum_{\alpha; \eta\in \mathcal{A}
  (\alpha,N),\gamma\in\mathcal{A}
  (\alpha,N)} \E  \left[|Z_\alpha| Z_\eta Z_\gamma \right] = o(1),
	\end{align*}
	where the last equality follows from (\ref{one}).

	Next,
	\begin{align*}
		A_2
		& = \left| \E  \left[ f'(\overline{{\bf Z}}) \left( 1- \frac{1}{a^{1/2}}\sum_{\alpha} Z_\alpha (\overline{{\bf Z}}-\bar{{\bf Z}}_{-\mathcal{A}(\alpha,N)})\right)\right]\right|  = \frac{1}{a} \left| \E  \left[ f'(\overline{{\bf Z}})\left( a - \sum_{\alpha, \eta\in \mathcal{A}
  (\alpha,N)} Z_\alpha Z_\eta\right)\right]\right|\\
		& \leq  \frac{||f'||}{a}  \E  \left| \left( a - \sum_{\alpha, \eta\in \mathcal{A}
  (\alpha,N)} Z_\alpha Z_\eta \right)\right|  =  \frac{||f'||}{a}  \E  \left| \left( \sum_{\alpha, \eta\in \mathcal{A}
  (\alpha,N)} Z_\alpha Z_\eta - \E\left[ Z_\alpha Z_\eta \right]\right)\right|\\
		& \leq   \frac{\sqrt{2}}{a\sqrt{\pi}}  \left( \var \left[\sum_{\alpha, \eta\in \mathcal{A}%\Delta
  (\alpha,N)}Z_\alpha Z_\eta \right]\right)^{1/2}  =   \frac{\sqrt{2}}{a\sqrt{\pi}} \left( \sum_{\alpha,\alpha',\eta\in\mathcal{A}
  \left(\alpha,N\right),\eta'\in\mathcal{A}
  \left(\alpha',N\right)} \cov\left(Z_{\alpha} Z_{\eta},Z_{\alpha'}Z_{\eta'}\right)\right)^{1/2},
	\end{align*}
	where the last inequality follows by Cauchy-Schwarz.  The final expression is $o(1)$ by (\ref{two}).\end{proof}

\

\begin{proof}[{\bf Proof of Corollary \ref{sparseCLT}}]
	We apply Theorem \ref{clt} to the case in which $\mathcal{A}(\alpha, N)=\{ \alpha \}$.   \eqref{one} becomes
	\[
	\sum_{\alpha} \E  \left[|Z_\alpha| ^3 \right]= o\left(\left(\sum_{\alpha } \var \left(Z_{\alpha}\right)\right)^{3/2}\right)
	\]
	which becomes\footnote{Recall that it is assumed that $\E  \left[|Z_\alpha| ^3 \right]/\E \left[Z_{\alpha}^2\right]^{3/2}$ is bounded above (and necessarily below via Jensen's Inequality).}
	\[
	\sum_{\alpha} \var \left(Z_{\alpha}\right)^{3/2} = o\left(\left(\sum_{\alpha } \var \left(Z_{\alpha}\right)\right)^{3/2}\right),
	\]
	which is satisfied directly, given that $\sum_{\alpha } \var \left(Z_{\alpha}\right)$ is growing without bound.
	
	(i) and (ii) imply \eqref{two} and (\ref{three}), respectively, (noting that the sign is always nonnegative by the supposition of the corollary).

We now show that for Bernoulli random variables with uniformly vanishing means, (i) holds whenever  (ii)  holds.  Observe that \begin{align*}
	{\rm cov}\left(Z_{\alpha}^{2},Z_{\eta}^{2}\right) & ={\rm cov}\left(\left(X_{\alpha}-\mu_{\alpha}\right)^{2},\left(X_{\eta}-\mu_{\eta}\right)^{2}\right)\\
	& ={\rm cov}\left(X_{\alpha}^{2}-2X_{\alpha}\mu_{\alpha}+\mu_{\alpha}^{2},X_{\eta}^{2}-2X_{\eta}\mu_{\eta}+\mu_{\eta}^{2}\right)\\
	& ={\rm cov}\left(X_{\alpha}^{2},X_{\eta}^{2}\right)-2\mu_{\alpha}{\rm cov}\left(X_{\alpha},X_{\eta}^{2}\right)-2\mu_{\eta}{\rm cov}\left(X_{\alpha}^{2},X_{\eta}\right)+4\mu_{\alpha}\mu_{\eta}{\rm cov}\left(X_{\alpha},X_{\eta}\right).
\end{align*}
Because they are Bernoulli, ${\rm cov}\left(X_{\alpha}^{k},X_{\eta}^{k'}\right)={\rm cov}\left(X_{\alpha},X_{\eta}\right)$
for any $k,k'>0$. Since the means tend to zero, this means
\[
{\rm cov}\left(Z_{\alpha}^{2},Z_{\eta}^{2}\right)={\rm cov}\left(X_{\alpha},X_{\eta}\right)\left(1+o\left(1\right)\right).
\]
Therefore satisfying (ii) implies (i) (noting also that  $a_N\geq 1$ so $a_N^2\geq a_N$).%
\end{proof}

\clearpage
\newpage

\setcounter{page}{1}

\begin{center}
{\bf {\large {\sc
Supplementary Online Appendix for:    \\
``A Network Formation Model Based on Subgraphs,''  by Chandrasekhar and Jackson}}}
\end{center}

\section{Proof of Proposition \ref{prop:LT_SUGM1-2}}\label{sec: networks_normal}
Let the moment (normalized) be
\[
\widehat{M}\left(\beta\right)=R_{n}S\left(g\right)-\E_{\beta}\left[R_{n}S\left(g\right)\right],
\]
where $R_{n}={\rm diag}\left\{ n^{h_{L}},n^{h_{T}}\right\} $ properly
normalizes the moments. So, for example, for links we have
\begin{align*}
\widehat{M}_{L}^{n}\left(\beta\right) & =\frac{n^{h_{L}}}{\binom{n}{2}}\sum_{i<j}\left\{ g_{ij}-{\rm E}_{\beta}g_{ij}\right\} =\frac{n^{h_{L}}}{\binom{n}{2}}\sum_{i<j}\left\{g_{ij}-q_{L}\left(\beta\right)\right\}.
\end{align*}
The objective function is
\[
\widehat{Q}_{n}\left(g,\beta\right):=\widehat{M}^{n}\left(\beta\right)'\widehat{M}^{n}\left(\beta\right).
\]
And we need
\[
\bar{Q}_{n}\left(\beta\right)=\E\left[\widehat{M}^{n}\left(\beta\right)\right]'\E\left[\widehat{M}^{n}\left(\beta\right)\right],
\]
which is the non-stochastic analogue.

\subsection{Identification}

We   prove identification for sequences of parameters, in the
 sense of identifiable uniqueness (Lemma 3.1 of \cite{potsherp1997}). See also ``Assumption ID'' in  \cite{andrews1990generic} and \cite{potscher1991basic}. The parameters $\beta^{n}_0$ are   identifiably unique in the sense
that for any $\varepsilon>0$
\[
{\rm lim}{\rm inf}_{n\rightarrow 0}\left[{\rm inf}_{\beta\in\mathcal{B}:\ \delta\left(\beta,\beta^n_0\right)>\varepsilon}\left|\overline{Q}^{n}\left(\beta\right)-\overline{Q}^{n}\left(\beta^n_0\right)\right|\right]>0.
\]

We take the usual Euclidean metric $\lVert \bnMD - b_0 \rVert$ to calculate the distance between two vectors $\bnMD, b_0$. Note that, since in our setting while $b_0$ is constant\footnote{It is possible to further allow $b_0 = b_0^n$ to depend on $n$ with restrictions on the uniform boundedness of these parameters from above and below. We do not pursue this here for clarity's sake.}  in $n$, $\beta^n_0$---the subgraph probability vector---has entries that tend to zero at hypothesized rates.

It is useful to note that in our setting, not only will we show that $\lVert \bnMD - b_0 \rVert \cvgto 0$ but in fact for a metric $\delta(\cdot,\cdot)$, defined below, we have $\delta(\betanMD, \beta_0^n) \cvgto 0$. In fact, the former follows from the latter mechanically as long as the parameter space is compact which we discuss below.

To see why this is useful, first consider the degenerate estimator $\betanMD = 0$ and observe $\lVert 0 - \beta^n_0 \rVert \cvgto 0$. That is, for a sequence of models in which the probability of any given subgraph tends to zero---mechanically true in any sparse random graph model---by definition the zero vector is a consistent estimator for the probability parameters, though this is uninformative.

The right metric for this sequence is to
set\footnote{We take $0/0=0$.}
\begin{equation}
\label{delta}
\delta(x,y) := \max_\ell \left[\frac{|x_\ell-y_\ell |}{\max(|x_\ell |,|y_\ell |)}\right],
\end{equation}
then the requirement becomes
\[
\delta(\betanMD, \beta^n_0) = \max_\ell \frac{ \left\vert \betanMDell-\beta^{n}_{0,\ell} \right\vert }{\max(\left\vert \betanMDell\right\vert ,\left\vert \beta^{n}_{0,\ell} \right\vert )} \cvgto 0.
\]
This requires that $\betanMDell$ and $\beta^n_{0,\ell}$ be proportional to each other far enough along the sequence.  Thus, if $\beta^n_0$ approaches 0,
saying that $\betanMDell$ is a good estimate of it under this metric also requires that $\betanMDell$ approach 0 at the same rate,  which is a much stronger
conclusion than just requiring that the two parameters converge in the usual Euclidean metric.

Returning to our degenerate estimator $\betanMDell = 0$, note
\[
\delta(0,\beta^n_0) = \max_\ell \frac{n^h_\ell \cdot \left\vert0-\beta^n_{0,\ell}\right\vert}{b_{0,\ell}}
 = \max_\ell \frac{\left\vert 0 - b_{0,\ell} \right\vert}{b_{0,\ell}} = 1
 \] which does not tend to zero; so in   the $\delta$ metric this an inconsistent estimator.

 Finally, if  $\delta(\betanMD,\beta^n_0) \cvgto 0$ then  $|\bnMD - b_0| \cvgto 0$ so any proof of consistency in $\beta$-space implies  and so the results in the paper follow as corollaries to the results below. To see this observe
\begin{align*}
\delta\left(\betanMD,\beta^{n}_0\right) & =\max_{\ell}\frac{\left|\betanMDell-\beta_{0,\ell}^{n}\right|}{\max\left(\left|\betanMDell\right|,\left|\beta_{0,\ell}^{n}\right|\right)}=\max_{\ell}\frac{n^{h_{\ell}}\left|\betanMDell-\beta_{0,\ell}^{n}\right|}{n^{h_{\ell}}\max\left(\left|\betanMDell\right|,\left|\beta_{0,\ell}^{n}\right|\right)}\\
 & =\max_{\ell}\frac{\left|\bnMDell-b_{0,\ell}\right|}{\max\left(\left|\bnMDell\right|,\left|b_{0,\ell}\right|\right)}\geq\max_{\ell}\frac{\left|\bnMDell-b_{0,\ell}\right|}{\overline{D}}.
\end{align*}
since by assumption $b_{\ell}$ lives in a compact set with maximum
$\overline{D}$. Since $\delta\left(\betanMD,\beta^{n}_0\right)\cvgto0$
then so must $\max_{\ell}\frac{\left|\bnMDell-b_{0,\ell}\right|}{\overline{D}}$,
proving the result.

\begin{prop}
\label{prop:identifiablyunique}
Consider a links and triangles SUGM with associated parameters
$\beta_{0,L}^n,\beta_{0,T}^n= \left(\frac{b_{0,L}}{n^{h_L}}, \frac{b_{0,T}}{n^{h_T}} \right)$   with  \(	h_{L}\in\left(\frac{1}{2},2\right)\text{ and }h_{T}\in\left[h_{L}+1,3h_{L}\right], \text{ with } h_T<3\). Then $\beta^n_{0,L},{\beta}^n_{0,T}$ are identifiably unique.
\end{prop}

\begin{proof}[{\bf Proof of Proposition \ref{prop:identifiablyunique}}]

Write\footnote{We allow the constants to depend on $n$ to capture that some applications have both rates and constants that adjust with scale, and we may want to fit across data of networks
of varying sizes.  But this is largely semantic, as estimating any particular network has only one $b$, and one can ignore the superscripts on the $b$s if one likes.}
\[
\beta^n = \left(\frac{b_L}{n^{h_L}}, \frac{b_T}{n^{h_T}}\right) \ \ \ \ \  \beta^n_0 = \left(\frac{b_{0,L}}{n^{h_L}}, \frac{b_{0,T}}{n^{h_T}}\right),
\]
where $b_L,b_T,b_{0,L},b_{0,T}$ lie in $[\underline{D}, \overline{D}]$

Let $r^n_L=1/n^{h_L}$ and $r^n_T=1/n^{h_T}$.

First, note that $ 1- (1-\beta_{T}^n)^{x}$ is the probability that some link is formed as part of at least one triangle out of $x$ possible triangles that could have it as an edge (independently of whether it also forms directly).

Next, note that the probability that a link forms conditional on some particular triangle that it could be a part of {\sl not forming} is\footnote{That is,
consider a given pair of nodes $i,j$ and a third node $k$. Consider the probability that link $ij$ is formed conditional on triangle $ijk$ not forming directly as a triangle.}
\begin{equation}
\label{tildeq2}
\widetilde{q}_L^n \left(\beta_{L}^n,\beta_{T}^n\right)=   \beta_L^n  + (1 -\beta_L^n )\left(1- (1-\beta_{T}^n)^{n-3}\right).
\end{equation}
In what follows, we omit the notation indicating the dependence of $\widetilde{q}_L^n $ on $\left(\beta_{L}^n,\beta_{T}^n\right)$.
So, we can write the probability of some triangle forming as
\begin{equation}
\label{ETdef2}
q_T^n:= \E_{\beta_L^n,\beta_T^n}\left[ S_{T}(g)  \right]   = \beta_T^n + (1-\beta_T^n) (\widetilde{q}_L^{n})^3,
\end{equation}
where the first expression $\beta_T^n$ is the probability that the triangle is directly generated, and then the second expression $(1 -\beta_T^n )(\widetilde{q}_L^{n})^3$ is the probability that it was not generated directly, but instead all three of the edges formed on their own (which happen independently, conditional on the triangle not forming, which has
probability $(\widetilde{q}_L^{n})^3$).

It is useful to note that since $\beta_L^n= o(1)$, $(1 -\beta_L^n )\rightarrow 1$ and
since $h_T>1$,
 $|(1-\beta_{T}^n)^{n-3} - (1- \frac{b_{0,T}}{n^{h_T-1}})|\rightarrow 0$.
Thus,
\[
\widetilde{q}_L^n = \Theta\left(\frac{1}{n^{h_L}} + \frac{1}{n^{h_T-1}}\right) = \Theta\left( \frac{1}{n^{h_L}}\right)
\]
where the second equality follows since $h_T\geq h_L+1$.

Next, note that the probability that a link forms is
\begin{equation}
\label{ELdef}
q_L^n:= \E_{\beta_L^n,\beta_T^n}\left[ S_{L}(g)  \right] = \beta_L^n  + (1 -\beta_L^n )\left(1- (1-\beta_{T}^n)^{n-2}\right),
\end{equation}
where the first expression $\beta_L^n$ is the probability that the link is directly generated, and then the second expression $(1 -\beta_L^n )\left(1- (1-\beta_{T}^n)^{n-2}\right)$ is the probability that it was not generated directly, but instead appeared as an edge in some triangle (and there are $n-2$ such possible triangles).

It is also useful to write this in a different way:
\begin{equation}
\label{ELdef3}
q_L^n:= \E_{\beta_L^n,\beta_T^n}\left[ S_{L}(g)  \right] = \beta_T^n + (1-\beta_T^n) \widetilde{q}_L^{n},
\end{equation}
noting that a link could form as part of a triangle that it is part of, or else form conditional upon that triangle not forming.

The following derivative expressions allow us to calculate rates on changes in expected counts of links and triangles as a function of parameters. In turn, this is useful in an approximation calculation below.
\begin{equation}
\label{dql}
\frac{ \partial \widetilde{q}_L^{n}}{\partial \beta_L^n}   = (1-\beta_T^n)^{n-3} \ \ \ \ \ \ \  \frac{ \partial \widetilde{q}_L^{n}}{\partial \beta_T^n}   =(n-3)(1-\beta_L^n) (1-\beta_T^n)^{n-4}.
\end{equation}

\[
\frac{ \partial q_L^n}{\partial \beta_L^n}   =  (1-\beta_T^n)^{n-2}.
\]

\[
\frac{ \partial q_T^n}{\partial \beta_L^n}   =  3(1-\beta_T^n) (\widetilde{q}_L^{n})^2   \frac{ \partial \widetilde{q}_L^{n}}{\partial \beta_L^n} = 3 (\widetilde{q}_L^{n})^2 (1-\beta_T^n)^{n-2}.
\]

\[
\frac{ \partial q_L^n}{\partial \beta_T^n}   = 1 -\widetilde{q}_L^{n} + (1-\beta_T^n)  \frac{ \partial \widetilde{q}_L^{n}}{\partial \beta_T^n} = 1 -\widetilde{q}_L^{n} + (n-3)(1-\beta_L^n) (1-\beta_T^n)^{n-3}.
\]

\[
\frac{ \partial q_T^n}{\partial \beta_T^n}   = 1 -(\widetilde{q}_L^{n})^3 + 3(1-\beta_T^n) (\widetilde{q}_L^{n})^2    \frac{ \partial \widetilde{q}_L^{n}}{\partial \beta_T^n} = 1 -(\widetilde{q}_L^{n})^3 + 3 (\widetilde{q}_L^{n})^2(n-3)(1-\beta_L^n) (1-\beta_T^n)^{n-3}.
\]

Given that $\beta_L^n=o(1)$ (since $h_L>0$), $\beta_T^n=o(1/n)$ (since $h_T>1$), and $\widetilde{q}_L^{n}=\Theta\left(\frac{1}{n^{h_L}}\right)$
the above expressions imply that:

\begin{equation}
\label{dqll}
\frac{ \partial q_L^n}{\partial \beta_L^n}   =  1-o(1),
\end{equation}

\begin{equation}
\label{dqtl}
\frac{ \partial q_T^n}{\partial \beta_L^n}   =  \Theta\left( \frac{ 1}{n^{2h_L}}\right),
\end{equation}

\begin{equation}
\label{dqlt}
\frac{ \partial q_L^n}{\partial \beta_T^n}   = n-2-o(1),
\end{equation}

\begin{equation}
\label{dqtt}
\frac{ \partial q_T^n}{\partial \beta_T^n}   = \Theta\left( \max[1,n^{1-2h_L}]\right).
\end{equation}

Note that (\ref{dqll})-(\ref{dqtt}) hold for any parameters $h_L>0$ and $3h_L> h_T\geq h_L+1$, and thus uniformly for any compact set of $\beta^n$ that have $h_L,h_T$ satisfying these inequalities. So
as long as we restrict attention to $\beta^n$ in that compact set, we have the same order derivatives and so then we approximate:
\begin{equation}
\label{almost2}
\frac{\E_{\beta^n}\left[ {S}_{L}(g) \right] - \E_{\beta^n_{0}}\left[ {S}_{L}(g)  \right]}{r^n_L} \approx  n^{h_L}\left[\frac{b_L-b_{0,L}}{n^{h_L}} + (n-2) \frac{b_T-b_{0,T}}{n^{h_T}}  \right]
\end{equation}
\[
\approx  {b_L-b_{0,L}}  +  ({b_T-b_{0,T}}) \Theta(n^{h_L+1-h_T}) ,
\]
and
\begin{equation}
\label{almost}
\frac{\E_{\beta^n}\left[ {S}_{T}(g) \right] - \E_{\beta^n_{0}}\left[ {S}_{T}(g)  \right]}{r^n_T} \approx  n^{h_T}\left[\frac{b_L-b_{0,L}}{n^{h_L}} \Theta(1/n^{2h_L}) +  \frac{b_T-b_{0,T}}{n^{h_T}} \Theta\left( \max[1,n^{1-2h_L}]\right)\right]
\end{equation}
\[
\approx (b_L-b_{0,L})\Theta(n^{h_T-3h_L}) + ({b_T-b_{0,T}}) \Theta\left( \max[1,n^{1-2h_L}]\right).
\]

To establish identifiable uniqueness (given the additive separability of $\overline{Q}^n(\beta)$ across $L,T$) it is sufficient to argue that for any $\varepsilon>0$ there exists $\phi>0$ such that   for large enough $n$, if
$\delta ((\beta^n_L, {\beta}_T^n), (\beta^n_{0,L}, {\beta}_{0,T}^n)) >\varepsilon$, then at least one of
the following inequalities holds:
\begin{equation}
\label{phiL}
\left|\frac{\E_{\beta^n}\left[ {S}_{L}(g) \right] - \E_{\beta^n_{0}}\left[ {S}_{L}(g)  \right]}{r^n_L} \right|>\phi
\end{equation}
or
\begin{equation}
\label{phiT}
\left|\frac{\E_{\beta^n}\left[ {S}_{T}(g) \right] - \E_{\beta^n_{0}}\left[ {S}_{T}(g)  \right]}{r^n_T} \right|> \phi.
\end{equation}

Note that $\delta ((\beta^n_L, {\beta}_T^n), (\beta^n_{0,L}, {\beta}_{0,T}^n)) >\varepsilon$
translates into
$|b_L- b_{0,L}|> c\varepsilon$ and/or
$|b_T- b_{0,T}|> c\varepsilon$ for some $c>0$.
If the second inequality holds, then by (\ref{almost}) it follows that (\ref{phiT}) holds.
If (\ref{phiT}) does not hold for any $\phi$, then by (\ref{almost}) it must be that $|b_L- b_{0,L}|> c\varepsilon$ while
$|b_T- b_{0,T}|<\delta^n$ for a sequence $\delta^n\rightarrow 0$.  In that case, noting that since $h_T\geq h_L-1$ (and so the second term of (\ref{almost2}) is of order at most 1 times $\delta^n$ while the first term is at least $c\varepsilon$ in magnitude), then by (\ref{almost2}) it follows that (\ref{phiL}) holds.
\end{proof}

\

\subsection{Consistency}

\begin{prop}\label{lem:SUGM-consistency} Consider a links and triangles SUGM with associated parameters
	$\beta_{0,L}^n,\beta_{0,T}^n= \left(\frac{b_{0,L}}{n^{h_L}}, \frac{b_{0,T}}{n^{h_T}} \right)$   with \(	h_{L}\in\left(\frac{1}{2},2\right)\text{ and }h_{T}\in\left[h_{L}+1,3h_{L}\right], \text{ with } h_T<3\). Then $\delta(\betanMD,\beta^n_0) \cvgto 0$ and therefore $\lVert \betanMD - b_0\rVert \cvgto 0 $.
\end{prop}

\begin{proof}[{\bf Proof of Proposition \ref{lem:SUGM-consistency}}] The proof follows from checking the conditions of Lemma 3.1 of \cite{potsherp1997} (see also \cite{jenishp2009}) or equivalently \cite{andrews1990generic}, Lemma 6.
In terms of those conditions here: $\mathcal{B}$ is compact, the weighting function is the identity matrix so it is positive semi-definite, and the moment function is continuous in $\beta$, and identifiable uniqueness was demonstrated in Proposition \ref{prop:identifiablyunique}. Uniform convergence is what remains to be verified.

Observe is that this requires showing
\[
\sup_{\beta}\left|\widehat{M}^{n}\left(\beta\right)-\E\widehat{M}^{n}\left(\beta\right)\right|=o_{p}\left(1\right)
\]
as
\begin{align*}
\sup_{\beta}\left|\widehat{Q}_{n}\left(g,\beta\right)-\bar{Q}_{n}\left(\beta\right)\right| & \leq\sup_{\beta}\left|\widehat{M}^{n}\left(\beta\right)'\widehat{M}^{n}\left(\beta\right)-\E\left[\widehat{M}^{n}\left(\beta\right)\right]'E\left[\widehat{M}^{n}\left(\beta\right)\right]\right|\\
 & \leq\sup_{\beta}\left|\left\{ \widehat{M}^{n}\left(\beta\right)-\E\left[\widehat{M}^{n}\left(\beta\right)\right]\right\} '\widehat{M}^{n}\left(\beta\right)\right|  \\
 &+\sup_{\beta}\left|\E\left[\widehat{M}^{n}\left(\beta\right)\right]'\left\{ \widehat{M}^{n}\left(\beta\right)-\E\left[\widehat{M}^{n}\left(\beta\right)\right]\right\} \right|\\
 & \leq2K\cdot\sup_{\beta}\left|\widehat{M}^{n}\left(\beta\right)-\E\left[\widehat{M}^{n}\left(\beta\right)\right]\right|
\end{align*}
for a constant $K$,
recalling we have assumed $\underline{D}_{L}<b_{L}<\overline{D}_{L}$ and $\underline{D}_{T}<b_{T}<\overline{D}_{T}$.

Thus, we show that
\(
\sup_{\beta}\left|\widehat{M}^{n}\left(\beta\right)-\E\widehat{M}^{n}\left(\beta\right)\right|=o_{p}\left(1\right).
\)
It is enough to show pointwise convergence coupled with stochastic equicontinuity.   Pointwise convergence follows along the lines of the proof of asymptotic normality in Section \ref{sec:app-normal} below, as that shows that the Central Limit Theorem applies, which then implies convergence.

Stochastic equicontinuity requires that for any $\epsilon > 0$, there exists $\eta > 0$ such that
\[
{ \limsup_n} \ \Pr\left\{ \sup_{\beta}\sup_{\beta'\in\delta\left(\beta,\beta'\right)<\eta}\left|\widehat{M}^{n}\left(\beta\right)-\widehat{M}^{n}\left(\beta'\right)\right|>\epsilon\right\} <\epsilon
\]
as in \cite{andrews1990generic}. A sufficient condition is a Lipschitz
condition: for every $\beta, \beta'$,
\[
\left|\widehat{M}\left(\beta\right)-\widehat{M}\left(\beta'\right)\right|=O_{p}\left(1\right)\cdot\delta\left(\beta,\beta'\right),
\]
which we verify next.

Let $\Delta=h_{T}-h_{L}$.  It is also useful to note (see the proof of Proposition \ref{prop:identifiablyunique}) that
\begin{align*}
\left|q_{L}^{n}\left(\beta\right)-q_{L}^{n}\left(\beta'\right)\right| & \leq\left(1+o\left(1\right)\right)\left|\beta_{L}^{n}-\beta_{L}^n{}'\right|+\Theta\left(n\right)\left|\beta^n_{T}-\beta_{T}^n{}'\right|\\
\end{align*}
and
\begin{align*}
\left|q_{T}^{n}\left(\beta\right)-q_{T}^{n}\left(\beta'\right)\right| & \leq\Theta\left(n^{-2h_{L}}\right)\left|\beta_{L}-\beta_{L}'\right|+\Theta\left(1\right)\left|\beta_{T}-\beta_{T}'\right|.
\end{align*}
Returning to the moments computation:
\begin{align*}
\left|\widehat{M}\left(\beta\right)-\widehat{M}\left(\beta'\right)\right| &= \left| (R_n S(g) - \E_\beta [R_n S(g)]) - (R_n S(g) - \E_{\beta'} [R_n S(g)])\right| \\
& = \left|   \E_\beta [R_n S(g)]  - \E_{\beta'} [R_n S(g)]\right|  \\
& \leq n^{h_{L}}\left|q_{L}\left(\beta\right)-q_{L}\left(\beta'\right)\right|+n^{h_{T}}\left|q_{T}\left(\beta\right)-q_{T}\left(\beta'\right)\right|\\
 & \leq\left(1+o\left(1\right)\right)\left|\beta_{L}-\beta'_{L}\right|n^{h_{L}}+\Theta\left(n\right)\left|\beta_{T}-\beta_{T}'\right|n^{h_{L}}\\
 & +\Theta\left(n^{-2h_{L}}\right)\left|\beta_{L}-\beta_{L}'\right|n^{h_{T}}+\Theta\left(1\right)\left|\beta_{T}-\beta_{T}'\right|n^{h_{T}}\\
 & \leq\Theta\left(1\right)\delta_{L}\left(\beta_{L},\beta_{L}'\right)+\Theta\left(n^{1-\Delta}\right)\delta_{T}\left(\beta_{T},\beta_{T}'\right)\\
 & +\Theta\left(n^{h_{T}-3h_{L}}\right)\delta_{L}\left(\beta_{L},\beta_{L}'\right)+\Theta\left(1\right)\delta_{T}\left(\beta_{T},\beta_{T}'\right)\\
 & \leq\Theta\left(1\right)\delta\left(\beta,\beta'\right),
\end{align*}
where the last inequality uses that $\Delta\geq 1$ and $h_{T}\leq 3h_{L}$.  Thus, the result follows. \end{proof}

\subsection{Asymptotic Normality}\label{sec:app-normal}
In what follows, begin with the restrictions required for identification and consistency which are  \(	h_{L}\in\left(\frac{1}{2},2\right)\text{ and }h_{T}\in\left[h_{L}+1,3h_{L}\right], \text{ with } h_T<3\). Asymptotic normality will require further  tightening of the restriction as will be seen below.

\subsubsection{Asymptotic Normality of Link and Triangle Shares}

Let $Y_L := \sum_{i<j}g_{ij}$ denote the sum of links which takes the place of the sum of the random variables in our general CLT.

\begin{lem}\label{lem:links-normality}
	Assume the rate requirements for identifiable uniqueness and consistency.
	Then
	\[
	a_{L}^{-1/2}\left(Y_{L}-{\rm E}\left[Y_{L}\right]\right)\rightsquigarrow\mathcal{N}\left(0,1\right)
	\]
	if	
	\(
	h_{L}\in\left(\frac{2}{3},2\right)\text{ and }h_{T}\in\left[h_{L}+1,3h_{L}\right], \text{ with } h_T<3.
	\)
	
\end{lem}

\begin{proof}
We apply the main theorem where $\alpha$ indexes a link $ij$. We
define the affinity set $\mathcal{A}
\left(\alpha,N\right):=\left\{ \eta:\ \eta\cap\alpha\neq\emptyset\right\} $,
so $\mathcal{A}
\left(ij,\binom{n}{2}\right)=\left\{ ij\right\} \cup\left\{ ik:\ k\neq i,j\right\} \cup\left\{ kj:\ k\neq i,j\right\} $.
Therefore dependency neighborhoods include $ij$ and all links that involve $i$ or $j$.

Condition \eqref{three} is obvious from the definition of $\mathcal{A}%\Delta
\left(ij,\binom{n}{2}\right)$,
because if $ij$ and $kl$ do not share nodes, no triangle nor link
can generate both. Thus they are independent and the left-hand side
term is 0.

Next, we verify Condition \eqref{one} as follows,
First, we note that with nonnegative and positvely associated Bernoulli random variables, instead of working with
$\E  \left[|Z_\alpha| Z_{\eta} Z_{\gamma} \right]$,  we
can substitute $\E  \left[X_\alpha X_{\eta} X_{\gamma} \right]$.\footnote{\label{fn:new} First, note that with nonnegative random variables,  $\E  \left[|Z_\alpha| Z_{\eta} Z_{\gamma} \right] \leq \E  \left[|Z_\alpha| X_{\eta} X_{\gamma} \right]$.  For Bernoulli random variables, $\E  \left[|Z_\alpha| X_{\eta} X_{\gamma} \right]$ becomes
$(1-\mu_\alpha) \Pr  \left[X_\alpha = 1 | X_{\eta} X_{\gamma} =1\right]\Pr  \left[ X_{\eta} X_{\gamma} =1\right] + \mu_\alpha \Pr  \left[X_\alpha = 0 | X_{\eta} X_{\gamma} =1\right]\Pr  \left[ X_{\eta} X_{\gamma} =1\right] $.
This becomes
$\Pr  \left[X_\alpha = 1 | X_{\eta} X_{\gamma} =1\right]\Pr  \left[ X_{\eta} X_{\gamma} =1\right] + \mu_\alpha \Pr  \left[ X_{\eta} X_{\gamma} =1\right] \left( \Pr  \left[X_\alpha = 0 | X_{\eta} X_{\gamma} =1\right] - \Pr  \left[X_\alpha = 1 | X_{\eta} X_{\gamma} =1\right] \right) $  or
$\E  \left[X_\alpha  X_{\eta} X_{\gamma} \right] +\mu_\alpha \Pr  \left[ X_{\eta} X_{\gamma} =1\right] \left( \Pr  \left[X_\alpha = 0 | X_{\eta} X_{\gamma} =1\right] - \Pr  \left[X_\alpha = 1 | X_{\eta} X_{\gamma} =1\right] \right) $.
Given the positive association and Bernoulli structure,
$\mu_\alpha \Pr  \left[ X_{\eta} X_{\gamma} =1\right] \leq
\E  \left[X_\alpha  X_{\eta} X_{\gamma} \right] $.
Thus, the whole expression is no more than
$2\E  \left[X_\alpha  X_{\eta} X_{\gamma} \right] $.
}
\[
{\rm E}\left|X_{\alpha}X_{\eta}X_{\gamma}\right|={\rm P}\left(X_{\alpha}X_{\eta}X_{\gamma}=1\right)\ \text{ s.t. }\ \eta,\gamma\in\mathcal{A}
\left(\alpha,N\right).
\]
We have three cases where all indices are distinct and two cases where
at least two indices are identical. Enumerating them, we have
\begin{enumerate}
	\item $ij,jk,il$ (a line) - there are $O\left(n^{4}\right)$ of these.
	\item $ij,ik,il$ (a star) - there are $O\left(n^{4}\right)$ of these.
	\item $ij,jk,ik$ (a triangle) - there are $O\left(n^{3}\right)$ of these.
	\item $ij,ij,ik$ or $ij,jk,jk$ (two repeat) - there are $O\left(n^{3}\right)$
	of these.
	\item $ij,ij,ij$ (all repeat) - there are $O\left(n^{2}\right)$ of these.
\end{enumerate}
From the proof of identification, recall $q_{L}^n$ is the probability
of a link forming in the graph, which can be due to a link forming
directly or as a part of a triangle. Also recall that $\widetilde{q}_{L}^n$
is the probability of a link forming if a particular triangle that
it could be a part of does not form. Finally, let $\widetilde{q}^n_{L}{}'$
denote the probability that a link forms conditional on two triangles,
that it could be part a part of, not forming. Note that we have, and
will continue to, suppress the dependence on $n$ unless explicitly
needed.

We can construct loose upper bounds on the probabilities of the various
structures:
\begin{enumerate}
	\item Line: $(\beta_{0,T}^n)^{2}+2\left(1-\beta_{0,T}^n\right)\beta_{0,T}^n\widetilde{q}_{L}^n+\left(1-\beta_{0,T}^n\right)^{2}(\widetilde{q}_{L}^n)^{2}\widetilde{q}^n_{L}{}'\leq(\beta_{0,T}^n)^{2}+2\beta_{0,T}^nq_{L}^n+(q_{L}^n)^{3}.$

	\item Star:
	\[
	(\beta_{0,T}^n)^{3}+3\left(1-\beta_{0,T}^n\right)(\beta_{0,T}^n)^{2}+3\left(1-\beta_{0,T}^n\right)^{2}\beta_{0,T}^n \widetilde{q}^n_{L}{}'+\left(1-\beta_{0,T}^n\right)^{3}\left(\widetilde{q}^n_{L}{}'\right)^{3}\leq4(\beta_{0,T}^n)^{2}+3\beta_{0,T}^nq_{L}^n+(q_{L}^n)^{3}.
	\]
	\item Triangle: $\beta_{0,T}^n+\left(1-\beta_{0,T}^n\right)\left(\widetilde{q}^n_{L}{}'\right)^3\leq\beta_{0,T}^n+(q_{L}^n)^{3}.$
	\item Two repeat: $\beta_{0,T}^n+\left(1-\beta_{0,T}^n\right)\left(\widetilde{q}^{n}_{L}{}'\right)^{2}\leq\beta_{0,T}^n+(q_{L}^n)^{2}.$
	\item All repeat: $q_{L}^n$.
\end{enumerate}

Then it follows from (1), (4), and (5), that
\begin{align}\label{eq:triple-bound}
{\rm E}\left|X_{\alpha}X_{\eta}X_{\gamma}\right|
\leq
\Theta\left(n^{4}\left((\beta_{0,T}^n)^{2}+\beta_{0,T}^nq_{L}^n+(q_{L}^n)^{3}\right)+n^{3}\left(\beta_{0,T}^n+(q_{L}^n)^{2}\right)+n^{2}q_{L}^n\right),
\end{align}
where we omit the dominated term from triangles.

For $k\neq i$, by binomial approximation
and bounds on lower order terms, the dominant term is the order of triangle formation and so it follows that
\[
{\rm cov}\left(X_{ij},X_{jk}\right)=\beta_{0,T}^n\left(1-\beta_{0,T}^n\right)\left(1-\widetilde{q}_{L}^n\right)^{2}\leq O \left(\beta_{0,T}^n\right).
\]
Then it follows that
\[
a_{N}=\Theta\left(n^{2}q_{L}^n+n^{3}\beta_{0,T}^n\right).
\]

For the sufficient condition for  \eqref{one} we need to compare this to the bound on ${\rm E}\left|X_{\alpha}X_{\eta}X_{\gamma}\right|$ from  \eqref{eq:triple-bound}
and show
\[
{\rm E}\left|X_{\alpha}X_{\eta}X_{\gamma}\right| = o\left(a_{N}^{3/2}\right)=o\left(n^{3}(q_{L}^n)^{3/2}+n^{9/2}(\beta_{0,T}^n)^{3/2}\right),
\]
and so we need to show
\[
n^{4}\left((\beta_{0,T}^n)^{2}+\beta_{0,T}^n q_{L}^n+(q_{L}^n)^{3}\right)+n^{3}\left(\beta_{0,T}^n+(q_{L}^n)^{2}\right)+n^{2}q_{L} = o\left(n^{3}(q_{L}^n)^{3/2}+n^{9/2}(\beta_{0,T}^n)^{3/2}\right).
\]

This imposes a number of constraints, and omitting the parts that are obviously
satisfied (e.g., $n^{4}(\beta_{0,T}^n)^{2}+n^{3}(q_{L}^n)^{2}=o\left(n^{9/2}(\beta_{0,T}^n)^{3/2}+n^{3}(q_{L}^n)^{3/2}\right)$),
this reduces to
\begin{equation}
	\label{newone}
	n^{4}\left(\beta_{0,T}^nq_{L}^n+(q_{L}^n)^{3}\right)+n^{3}\left(\beta_{0,T}^n\right)+n^{2}q_{L}^n
 =o\left(n^{3}(q_{L}^n)^{3/2}+n^{9/2}(\beta_{0,T}^n)^{3/2}\right).
\end{equation}
Recall in addition $h_{L}>1/2$ and $h_{T}\in[h_{L}+1,3h_{L}]$ from
the identification and consistency requirements. Noting that, as in
the proof of identification, (working there with $\widetilde{q}_{L}^n$
which is of the same order)
\[
q_{L}^{n}=\Theta\left(n^{-h_{L}}+n^{-\left(h_{T}-1\right)}\right)=\Theta\left(n^{-h_{L}}\right)
\]
since $h_{T}-1\geq h_{L}$.
The condition (\ref{newone}) is then satisfied if
we can show that
\begin{equation}
	\label{newtwo}
	n^{ 4-h_{T}-h_{L}} + n^{4-3h_{L}} + n^{3-h_{T}}+ n^{2-h_{L}} < n^{3-\left(3/2\right)h_{L}} + n^{9/2-\left(3/2\right)h_{T}}.
\end{equation}
Since   $h_{T}\geq h_{L}+1$,
\[
n^{ 4-h_{T}-h_{L}} + n^{4-3h_{L}} + n^{3-h_{T}}+ n^{2-h_{L}} \leq  n^{ 3-3h_{L}} + n^{4-3h_{L}} + n^{2-h_{L}}+ n^{2-h_{L}}
\]
and thus to show (\ref{newtwo}) it is enough to show that
\[
\max\left\{4-3h_{L},2-h_{L}\right\} <
3-\left(3/2\right)h_{L},
\]
which holds since $2/3< h_L < 2$ (which exactly correspond to the crossing points).

Next we turn to Condition \eqref{two}. We will show that this is implied by the above
restrictions. To do this, we compute terms of the form
\[
{\rm cov}\left(\left(g_{ij}-q_{L}^n\right)\left(g_{ik}-q_{L}^n\right),\left(g_{rs}-q_{L}^n\right)\left(g_{st}-q_{L}^n\right)\right)
\]
since $\eta\in\mathcal{A}
\left(\alpha,N\right)$ and $\eta'\in\mathcal{A}
\left(\alpha',N\right)$;
here we allow for the cases that $k=i$ and $r=t$. Via a detailed expansion, and given the nonnegative correlation of the presence of links,
one can show that
\[
{\rm cov}\left(\left(g_{ij}-q_{L}^n\right)\left(g_{ik}-q_{L}^n\right),\left(g_{rs}-q_{L}^n\right)\left(g_{st}-q_{L}^n\right)\right)\leq{\rm E}\left[g_{ij}g_{jk}g_{rs}g_{st}\right].
\]
It is easy to see that if $\left\{ i,j,k\right\} \cap\left\{ r,s,t\right\} =\emptyset$
then the covariance is zero since the events are independent. Thus,
we are summing over the cases in which the intersection is non-empty.
The cases with intersection of two or more nodes are handled as we
already did above, noting that the condition here is less restrictive
($a_{N}>1$ so $a_{N}^{2}>a_{N}^{3/2}$).

So we restrict attention to the case where there is only one node
of intersection. In this case the intersection could come from:
\begin{enumerate}
	\item $s=j$, so two-stars joined at the center,
	\item $r=i$, so a line,
	\item $s=i$, so the center of one star is attached to the leaf of the other.
\end{enumerate}
These exhaust all configurations up to a relabeling.

Consider the event $g_{ij}g_{jk}g_{rs}g_{st}=1$. Assume we are in
case 1. This has the highest probability relative to the other two
cases, so we can construct a crude bound using this to finish the
result. The probability is of order no more than
\[
(\beta_{0,T}^n)^{3}+(\beta_{0,T}^n)^{2}+(\beta_{0,T}^n)^{2}q_{L}^n+\beta_{0,T}^n(q_{L}^n)^{2}+(q_{L}^n)^{4}.
\]
Therefore we check
\[
n^{5}\left(\beta_{0,T}^{3}+\beta_{0,T}^{2}+\beta_{0,T}^{2}q_{L}+\beta_{0,T}q_{L}^{2}+q_{L}^{4}\right)=o\left(n^{4}q_{L}^{2}+n^{6}\beta_{0,T}^{2}\right).
\]
The relevant rate on the right-hand side is $n^{4-2h_{L}}$, since
$h_{T}\geq h_{L}+1$. Dividing both sides by $n^{4}$, the inequality boils down to four conditions:
\begin{enumerate}
	\item $1-3h_{T}<-2h_{L}$ or $1+2h_{L}<3h_{T}$, which
	is implied by $h_T\geq h_L+1$;
	\item $1-2h_{T}<-2h_{L}$ or $h_{L}+\frac{1}{2}<h_{T}$ which is also implied by $h_T\geq h_L+1$;
	\item $1-2h_{T}-h_{L}<-2h_{L}$ or $1+h_{L}< 2h_{T}$, which is also implied by $h_T\geq h_L+1$;
	\item $1-4h_{L}<-2h_{L}$ or $h_{L}>1/2$, which is implied by  $h_L> 2/3$.
\end{enumerate}
This concludes the proof.
\end{proof}

Let $Y_T := \sum_{i<j<k}g_{ijk}$ denote the sum of triads.

\begin{lem}\label{lem:triangles-normality}
	Assume the rate requirements for identifiable uniqueness and consistency.
	Then
	\[
	a_{T}^{-1/2}\left(Y_{T}-{\rm E}\left[Y_{T}\right]\right)\rightsquigarrow\mathcal{N}\left(0,1\right)
	\]
	if	
	\(
h_{L}\in\left(\frac{2}{3},2\right)\text{ and }h_{T}\in\left[h_{L}+1,3h_{L}\right], \text{ with } h_T<3.
\)
\end{lem}

\begin{proof}
For this proof we appeal to Corollary \ref{sparseCLT}. Note that $\E[Z_\alpha Z_\eta \vert Z_\eta ] \geq 0$ is always satisfied (regardless of the subgraph types in question) in our subgraph model since one subgraph's presence either has no effect on another or can help generate the other incidentally.  Thus, we verify that the other conditions are satisfied.

 Here we set $\mathcal{A}
 \left(ijk,\binom{n}{3}\right)=\left\{ ijk\right\} $.

We begin with (ii) since it implies (i) for Bernoullis, which is equivalent to showing
\[
\sum_{\alpha\neq\eta}{\rm cov}\left(X_{\alpha},X_{\eta}\right)=o\left(N\cdot{\rm var}\left(X_{\alpha}\right)\right).
\]
Applying the calculation in the proof of Proposition \ref{prop:identifiablyunique},
\[
{\rm var}\left(X_{\alpha}\right)=q_{T}^n\left(1-q_{T}^n\right)=\Theta\left(\beta_{0,T}^n\right).
\]
We calculate the covariances for the various cases of
$\alpha,\eta$ and check when they are of lesser order. We have two
relevant cases: where the two indices intersect on one node and when
they intersect on two nodes. By independence if they do not intersect
at all, the covariance is zero.
\begin{enumerate}
	\item $\left|\alpha\cap\eta\right|=1$: ${\rm cov}\left(X_{\alpha},X_{\eta}\right)=\Theta\left(\beta_{0,T}^n(q_{L}^n)^{4}\right)$
	and there are $O\left(n^{5}\right)$ of these.
	
	The triangles are node-adjacent. Since at least one link needs to
	form together, this can only happen if the joint node is part of a
	triangle and neither of the triangles formed directly. This gives
	$(\beta_{0,T}^n)^{4}(\widetilde{q}_{L}^n)^{4}=O\left(\beta_{0,T}^n(q_{L}^n)^{4}\right)$.
	\item $\left|\alpha\cap\eta\right|=2$: ${\rm cov}\left(X_{\alpha},X_{\eta}\right)=\Theta\left(\beta_{0,T}^n(q_{L}^n)^{2}+(q_{L}^n)^{5}\right)$
	and there are $O\left(n^{4}\right)$ of these.
	
	The triangles are edge-adjacent. This is because we need the common
	link from each triangle to have formed together and not have already
	formed independently in both cases, which can happen only if exactly
	one of the triangles formed directly and the other did not, or else
	neither triangle formed and all of the links have to form. So, this
	is of order $\beta_{0,T}^n(\widetilde{q}_{L}^n)^{2}+\left(\widetilde{q}^n_{L}{}'\right)^{5}\leq\beta_{0,T}^n(q_{L}^n)^{2}+(q_{L}^n)^{5}$.
	
\end{enumerate}
Each of these must be of order $o\left(n^{3}\beta_{0,T}\right)$ for
the result to hold. The conditions therefore are
\begin{enumerate}
	\item $5-h_{T}-4h_{L}<3-h_{T}$ or $1/2<h_{L}$, which is satisfied since $h_L>2/3$.
	\item $4-h_{T}-2h_{L}<3-h_{T}$ or $1/2<h_{L}$ which is the same as above.
	\item $4-5h_{L}<3-h_{T}$ or $h_{T}<5h_{L}-1$. But notice that $3h_{L}<5h_{L}-1$
	so long as $h_{L}>1/2$, and so this is implied by $h_{T}<3h_{L}$
\end{enumerate}
This completes the proof.\end{proof}

\subsubsection{Joint Asymptotic Normality}

Let $Y = (Y_L, Y_T)'$, the vector of the sums of links and triangles.

\begin{lem}\label{lem:joint_normality}
	If 	\(
	h_{L}\in\left(\frac{2}{3},2\right)\text{ and }h_{T}\in\left[h_{L}+1,3h_{L}\right], \text{ with } h_T<3.
	\), then
	\[
	a^{-1/2}\left(Y-{\rm E}\left[Y\right]\right)\rightsquigarrow\mathcal{N}\left(0,I\right)
	\]
	where $a$ is the variance-covariance matrix defined below.
\end{lem}
\begin{proof}  We will apply the Cram\'{e}r-Wold device to show joint normality through showing all weighted (normalized) sums are normally distributed. Specifically, Lemma 2.1 of \cite{biscio2018note} contains a useful generalization which we use.
	
	Let $\Lambda=\left\{ ijk:\ i,j,k\in\left[1:n\right]\right\} \cup\left\{ ij:\ i,j\in\left[1:n\right]\right\} .$
	This consists of a set of $\binom{n}{2}+\binom{n}{3}$ terms. Notice
	that the set has two types of random variables coming from links and
	triangles. We now alter the affinity set for the sake of links (for
	triangles they may remain the same). Specifically
	\[
	\mathcal{A}
 \left(ij;N\right):=\left\{ ik:\ \forall k\right\} \cup\left\{ jk:\ \forall k\right\} \cup\left\{ irs\text{ and }jrs:\ \forall r,s\right\} .
	\]

Let
\[
a=\left(\begin{array}{cc}
	a_{L} & a_{LT}\\
	a_{LT} & a_{T}
\end{array}\right).
\]
where the two diagonal variance terms have been studied
\begin{enumerate}
	\item $a_L = \sum_{ij,rs\in\mathcal{A}
 \left(ij,N\right)}{\rm cov}\left(X_{ij},X_{rs}\right)$,
	\item $a_T = \sum_{ijk}{\rm var}\left(X_{ijk}\right)$, and
	\item $a_{LT} = \sum_{ij.rst\in\mathcal{A}
 \left(ij,N\right)}{\rm cov}\left(X_{ij},X_{rst}\right)$.
\end{enumerate}

 We need to check that for every $w\in \mathbb{R}^2$,
 \[\left(w'a_{N}w\right)^{-1/2}w'(Y - \E[Y]) \wkto \mathcal{N}(0,1)\] which is Lemma 2.1 of \cite{biscio2018note}. But this  reduces to checking the conditions of   Theorem \ref{clt} for these now $w$-weighted sums.

 We need to calculate growth rates for the new covariance term:
	\[
a_{LT}	= \sum{\rm cov}\left(X_{ij},X_{rst}\right)=\Theta\left(n^{4}\beta_{0,T}^nq_{L}^n+n^{4}(\beta_{0,T}^n)^{2}+n^{3}\beta_{0,T}^n\right).
	\]
	We have for any weighted sum of $a_L$, $a_T$, and $a_{LT}$ ($w$ is fixed in $n$  so does not matter) the order
	\[
	\left(n^{2}q_{L}^n+n^{3}\beta_{0,T}^n\right)+\left(n^{4}\beta_{0,T}^nq_{L}^n+n^{4}(\beta_{0,T}^n)^{2}+n^{3}\beta_{0,T}^n\right)+\left(n^{3}\beta_{0,T}^n\right)
	\]
	or collecting terms and dropping the obviously dominated ones
	\[
	n^{4}\beta_{0,T}^nq_{L}^n+n^{3}\beta_{0,T}^n+n^{2}q_{L}^n
	\]
	where, notice, the latter two terms were the rates of $a_L$ and $a_T$ and the possible new component is given by the first term.
	
Again Condition \eqref{three} follows directly, so we check the other two.
	
	Condition \eqref{one} is as follows. We examine the new terms not covered by the prior two
	lemmas and appeal to the sufficient condition as shown in footnote \ref{fn:new}. These are of the form
	\[
	{\rm E}\left|X_{\alpha}X_{\eta}X_{\gamma}\right|={\rm P}\left(X_{\alpha}X_{\eta}X_{\gamma}=1\right)
	\]
	which we now bound. The loose bounds across the two cases are
	\begin{enumerate}
		\item one link and two triangles: 	
		This must constitute \emph{edge adjacent} triangles. Otherwise we
		automatically have independence.
		This leaves 4 nodes, so order $n^{4}$ terms with a bound on probability
		$(\beta_{0,T}^n)^{2}$ which is notably loose.
		\item one link, one triangle, and a second link:
		This has a loose upper bound of $(\beta_{0,T}^n)^{2}+\beta_{0,T}^n(q_{L}^n)^{2}$.
		This leaves 4 nodes again so order $n^{4}$ of these.
		
	\end{enumerate}
	This exhausts the list.
	
	So we compare
	\begin{enumerate}
		\item $n^{4}(\beta_{0,T}^n)^{2}<n^{6}(\beta_{0,T}^n)^{3/2}(q_{L}^n)^{3/2}+n^{9/2}(\beta_{0,T}^n)^{3/2}+n^{3}(q_{L}^n)^{3/2}$
		\begin{enumerate}
			\item $4-2h_{T}<6-\frac{3}{2}h_{T}-\frac{3}{2}h_{L}$ so $h_{L}<\frac{4}{3}+h_{T}$
			which holds for every $h_{T}\geq h_{L}+1$.
			\item $4-2h_{T}<9/2-\frac{3}{2}h_{T}$ so this is always true.
			\item $4-2h_{T}<3-\frac{3}{2}h_{L}$ so $\frac{1}{2}+\frac{3}{4}h_{L}<h_{T}$
			which already holds.
		\end{enumerate}
		\item $n^{4}\left((\beta_{0,T}^n)^{2}+\beta_{0,T}^n(q_{L}^n)^{2}\right)<n^{6}(\beta_{0,T}^n)^{3/2}(q_{L}^n)^{3/2}+n^{9/2}(\beta_{0,T}^n)^{3/2}+n^{3}(q_{L}^n)^{3/2}$.
		\begin{enumerate}
			\item $4-h_{T}-2h_{L}<6-\frac{3}{2}h_{T}-\frac{3}{2}h_{L}$ so $h_{T}<4+h_{L}$
			but this holds precisely because $h_{L}<2$ and $h_{T}\leq3h_{L}$.
			\item $4-h_{T}-h_{L}<9/2-\frac{3}{2}h_{T}$ is mechanical.
			\item $4-h_{T}-h_{L}<3-\frac{3}{2}h_{L}$ follows from $1+\frac{1}{2}h_{L}<h_{T}$
			which is true by assumption.
		\end{enumerate}
	\end{enumerate}
	As a consequence, we have
	\[
	\sum_{\alpha;\eta,\gamma\in\Delta\left(\alpha,N\right)}{\rm E}\left|X_{\alpha}X_{\eta}X_{\gamma}\right|=o\left(a^{3/2}\right)
	\]
	since we have controlled the within-link and within-triangle terms
	in the prior two Lemmas and the cross-term above.

	Next we need to verify Condition \eqref{two}. The condition will be met with the exact same rates. To see this, first observe  that we only  need to consider terms of the form
	\[
	{\rm cov}\left(\left(g_{ij}-q_{L}^n\right)\cdot\left(g_{ik}-q_{L}^n\right),\left(g_{rs}-q_{L}^n\right)\cdot\left(X_{rvw}-q_{T}^n\right)\right)\leq{\rm E}\left[g_{ij}g_{ik}g_{rs}g_{rv}g_{rw}\right]
	\]
	and
	\[
	{\rm cov}\left(\left(g_{ij}-q_{L}^n\right)\cdot\left(X_{ikl}-q_{T}^n\right),\left(g_{rs}-q_{L}^n\right)\cdot\left(X_{rvw}-q_{T}^n\right)\right)\leq{\rm E}\left[g_{ij}g_{ik}g_{kl}g_{rs}g_{rv}g_{rw}\right].
	\]
	That is,  the reference nodes in this condition must be pairs because triples have dependency neighborhoods that are singletons. Also observe that the neighbors considered must have at least one triangle because the all-links case has already been covered in Lemma \ref{lem:links-normality}.
	
	We have two cases to consider: a 7-node case with $\{i,j,k\} \cap \{r,s,v,w\} \neq \emptyset$ and an 8-node case with $\{i,j,k,l\} \cap \{r,s,v,w\}\neq \emptyset$.
	
	Let us begin with the 7-node case. Here we can have one or two intersections (we have already calculated cases with 4 or fewer nodes, meaning three or more intersections).
 Begin with a single node in common. One can check that amongst all configurations (which intersects a two-star $ijk$ with a triangle with a leaf ($ruv$ and $rs$)),  an upper bound on the probability of the structure forming is of order $(\beta_{0,T}^n)^2 q_L^n$. There are  order $n^6$ such potential structures. We need to compare this to the square of the weighted sum of variance-covariance terms, which by the above is
 \[
\left( n^{4}\beta_{0,T}^nq_{L}^n+n^{3}\beta_{0,T}^n+n^{2}q_{L}^n\right)^2.
 \]
	Observe that the first term is the only one to consider---the other two have been studied,
	\[
	6 - 2h_T - h_L < 4 - 2h_L\iff 2  + h_L < 2 h_T
	\]
	follows directly from  $h_T \geq h_L + 1$ so the result follows in this case.
	
	Next we can look at two nodes in common. This involves a number of configurations of one or two triangles and a collection of leafs and/or stars. Here we have order $n^5$ free nodes and we can check a loose upper bound on the probability of formation is of order $\beta_{0,T}^n (q_L^n)^2$. As such we see
	\[
	5 - h_T - 2h_L < 4 - 2h_L
	\]
	which is implied by $h_T > 1$, which is satisfied in our setting since $h_T \geq h_L+1 > 1+2/3$.

	This covers all cases that have not previously been calculated. So then we turn to the 8-node case. We can have one or two nodes in common before we repeat calculations already covered. If we have one node in common, the loose upper bound is probability of order $(\beta_{0,T}^n)^2 (q_L^n)^2$ and there are of course $n^7$ such potential collections of nodes. The condition to check is then
	\[
	7-2h_T - 2h_L < 4-2h_L \iff 3 < 2h_T.
	\]
	But this is mechanical since $h_T \geq 1+2/3$ in any case.
	
	The final case to consider is with two nodes in common in this case of 8 nodes with two in common. In this case, a loose upper bound is $(\beta_{0,T}^n)^2 q_L^n$ and there are $n^6$. But we have studied this above and the restriction is already satisfied.
	
	Therefore, with $\Lambda$ and $\mathcal{A}
 (\alpha,N)$ for $\alpha \in \Lambda$ defined as above, we have shown that Conditions \eqref{one} and \eqref{two} are satisfied. Since the remaining conditions were already discussed, the sum is asymptotically normally distributed. But, since the weights were arbitrary (indeed fixed in $n$ and therefore did not contribute to any of the dependency calculations)	the result held irrespective of $(w_1,w_2) \in \mathbb{R}^2$, so joint normality follows.
\end{proof}

%-----------------------------------------------------------------------------------%
\subsubsection{Asymptotic Normality of SUGM Estimators}
It is useful to define the variance-covariance matrix of the moments and a rate matrix
\[
V_{n}=\left(\begin{array}{cc}
\var(n^{h_L} S_L) & \cov(n^{h_L} S_L,n^{h_T} S_T)\\
\cov(n^{h_L} S_L,n^{h_T} S_T) & \var(n^{h_T} S_T)
\end{array}\right)
\ \ \ \mbox{ and  } R_{n}=\left(\begin{array}{cc}
n^{h_{L}} & 0\\
0 & n^{h_{T}}
\end{array}\right).\]

\begin{prop} \label{prop:LT_SUGM1} Consider a links and triangles SUGM with associated parameters
$\beta_{0,L}^n,\beta_{0,T}^n= \left(\frac{b_{0,L}}{n^{h_L}}, \frac{b_{0,T}}{n^{h_T}} \right)$ with $0\leq \underline{D}<b_{0,L},b_{0,T}<\overline{D}$
such that
\[
h_L \in (2/3,2) \mbox{ and } h_T \in [h_L + 1, 3h_L] \text{ with }h_T < 3.
\]
Consider the minimum distance estimator $\betanMD$ using moments $S = (S_L(g),S_T(g))$. Then
\[\delta \left(\betanMD,\beta_{0}^{n}\right) \cvgto 0\]
and\footnote{The expression for $V_n$ is different when $h_T = h_L+1$, and is given in the proof of the proposition.}
\[
V^{-1/2}_n R_n \left(\betanMD-\beta_{0}^{n}\right)\wkto\mathcal{N}\left(0,I \right).
\]
\end{prop}

\

\begin{proof}[{\bf Proof of Proposition \ref{prop:LT_SUGM1}}]

The proof of the result follows the outline of  standard results  on asymptotic normality of parameter estimates (e.g., \cite{newey1994large}).

It is
convenient to normalize
things via a change of variables via the diagonal normalizing matrix $R_n={\rm diag}\{n^{h_L},n^{h_T}\}$ to a parameter vector
\(
b:=R_n\beta^n,
\)
so that the magnitude of the parameter vector does not change with $n$.
Observe that $\delta\left(\betanMD,\beta_{0}^{n}\right)\cvgto 0$ if and only if $\bnMD\cvgto b_{0}$, and consistency in the $\delta$-metric holds by Proposition \ref{lem:SUGM-consistency}.

It is also useful to then define the expected and empirical moment functions in terms of this rescaled parameter
\[
\widehat{M}_{L}\left(b\right)=\left[\frac{n^{h_{L}}}{\binom{n}{2}}\sum_{i<j}g_{ij}-\Q_{L}\left(b_{L},b_{T}\right)\right]
\]
and
\[
\widehat{M}_{T}\left(b\right)=\left[\frac{n^{h_{T}}}{\binom{n}{3}}\sum_{i<j<k}g_{ij}g_{jk}g_{ik}-\Q_{T}\left(b_{L},b_{T}\right)\right]
\]
where $\Q_L(b)=n^{h_{L}}q_L^n$ and $\Q_T(b)=n^{h_T}q_T^n$ are the normalized expectations given parameter $(b_{L},b_{T})=R_n\beta^n$.

Recall  $\Delta=h_{T}-h_{L}$.
We treat two separate cases, $\Delta>1$ and $\Delta=1$.   The second case allows links to generate triangles at a similar rate as triangles, and so is a more complex case to treat, and so each step of the argument involves different arguments for the two cases.

From the first order condition of GMM estimation, we take a mean value expansion around the true normalized parameter $b_{0}$  by applying the mean-value theorem, and then solve for $\bnMD-b_0$.\footnote{This is valid because $b_{0}$ is assumed to lie in the interior of
$B$, a compact set, which then implies the sequence of $\mathcal{B}^{n}$
under consideration.} Note that the mean value $\bar{b}$ is evaluated component by component
in the matrix $\nabla\widehat{M}\left(\bar{b}\right)$. This abuse
of notation is standard (e.g., \cite{newey1994large}).
\[
R_{n}\left(\betanMD-\beta_{0}^n\right)=\left(\bnMD-b_{0}\right)=-\left[\nabla\widehat{M}\left(\bnMD\right)'\nabla\widehat{M}\left(\bar{b}\right)\right]^{-1}\nabla\widehat{M}\left(\bnMD\right)'\widehat{M}\left(b_{0}\right).
\]
Below we will show that for $\Delta > 1$
\[
-\left[\nabla\widehat{M}\left(\bnMD\right)'\nabla\widehat{M}\left(\bar{b}\right)\right]^{-1}\nabla\widehat{M}\left(\bnMD\right)'\cvgto I
\]
and by Lemma \ref{lem:joint_normality}, for
\[
V_{n}=\left(\begin{array}{cc}
\var(n^{h_L} S_L) & \cov(n^{h_L} S_L,n^{h_T} S_T)\\
\cov(n^{h_L} S_L,n^{h_T} S_T) & \var(n^{h_T} S_T)
\end{array}\right)
\]
it follows that
\[
V^{-1/2}_n\widehat{M}\left(b_{0}\right)\rightsquigarrow\mathcal{N}\left(0,I\right).
\]
Therefore by Slutzky's theorem, it follows that
\[
V^{-1/2}_n R_n\left(\betanMD-\beta_{0}^n\right)\rightsquigarrow\mathcal{N}\left(0,I\right).
\]

Thus, to complete the proof for the case of  $\Delta > 1$, it  suffices to show that
\[
-\left[\nabla\widehat{M}\left(\bnMD\right)'\nabla\widehat{M}\left(\bar{b}\right)\right]^{-1}\nabla\widehat{M}\left(\bnMD\right)'\cvgto I
\]

For the case of $\Delta=1$ we will end up with a different expression for the limit of
\[
-\left[\nabla\widehat{M}\left(\bnMD\right)'\nabla\widehat{M}\left(\bar{b}\right)\right]^{-1}\nabla\widehat{M}\left(\bnMD\right)'
\]
and so will have a different covariance and normalization.

To find the limit of these gradient terms, we need to compute
$\nabla \Q$, where we define
\begin{align*}
\Q_{L}\left(b\right) & :=n^{h_{L}}\left[q_{L}^n\left(\beta^n\right)\right]\\
 & =n^{h_{L}}\left[\beta_{L}^n+\left(1-\beta_{L}^n\right)\left[1-\left(1-\beta_{T}^n\right)^{n-2}\right]\right]\\
 & =b_{L}+\left(n^{h_{L}}-b_{L}\right)\left(n-2\right)\cdot\frac{b_{T}}{n^{h_{T}}}\\
 & =b_{L}+\frac{b_{T}}{n^{h_{T}-h_{L}-1}}+o\left(1\right).
\end{align*}

Similarly
\begin{align*}
\Q_{T}\left(b\right) & :=n^{h_{T}}\left[q_{T}^n\left(\beta^n\right)\right]\\
 & =n^{h_{T}}\left[\beta_{T}^n+\left(1-\beta_{T}^n\right)\left\{ \beta_{L}^n+\left(1-\beta_{L}^n\right)\left[1-\left(1-\beta_{T}^n\right)^{n-2}\right]\right\} ^{3}\right]\\
 & =b_{T}+\left(n^{h_{T}}-b_{T}\right)\left\{ \frac{b_{L}}{n^{h_{L}}}+\frac{b_{T}}{n^{h_{T}-1}}-\frac{b_{T}}{n^{h_T + h_L-1}}b_{L}\right\} ^{3}\\
 & \approx b_{T}+\left\{ \underbrace{ \frac{b_{L}}{n^{h_{L}-h_{T}/3}}+\frac{b_{T}}{n^{2h_{T}/3-1}}-\frac{b_{T}}{n^{h_T + h_L-1-h_{T}/3}}b_{L} }_{=:x}\right\}^{3}
\end{align*}
where we dropped the final term that is of lesser order.

Note that the third term will always be of lesser order, so
\[
x=O\left(\left\{ \frac{b_{L}}{n^{h_{L}-h_{T}/3}}+\frac{b_{T}}{n^{2h_{T}/3-1}}\right\} \right).
\]
Also notice
\[
h_{L}-h_{T}/3<2h_{T}/3-1\iff \Delta>1.
\]
Thus, if $\Delta>1$ only the first term in $x$ matters, while if $\Delta=1$ then the two terms are of the same order.

Finally it will be useful to write
\[
\nabla \Q_{L}=\left(\begin{array}{c}
1+o\left(1\right)\\
n^{1-\Delta}+o\left(n^{1-\Delta}\right)
\end{array}\right)
\]
and
\[
\nabla \Q_{T}=\left(\begin{array}{c}
3x^{2}\left(\frac{1}{n^{h_{L}-h_{T}/3}}(1+o(1))\right)\\
1+3x^{2}\left(\frac{1}{n^{2h_{T}/3-1}}(1+o(1))\right)
\end{array}\right).
\]

Consider the case where $\Delta>1$. Then
\[
\nabla\Q_{L}=\left(\begin{array}{c}
1+o\left(1\right)\\
n^{1-\Delta}+o\left(n^{1-\Delta}\right)
\end{array}\right)=\left(\begin{array}{c}
1+o\left(1\right)\\
o\left(1\right)
\end{array}\right)
\]
and
\[
\nabla\Q_{T}=\left(\begin{array}{c}
o\left(1\right)\\
1+o\left(1\right)
\end{array}\right),
\]
since $3h_{L}>h_{T}$ and $h_{T}>\frac{3}{2}$.

Now consider the case where $\Delta=1$. Again
\[
\nabla\Q_{T}=\left(\begin{array}{c}
o\left(1\right)\\
1+o\left(1\right)
\end{array}\right)
\]
but in this case
\[
\nabla\Q_{L}=\left(\begin{array}{c}
1+o\left(1\right)\\
n^{1-\Delta}+o\left(n^{1-\Delta}\right)
\end{array}\right)=\left(\begin{array}{c}
1+o\left(1\right)\\
1+o\left(1\right)
\end{array}\right).
\]

Notice that $\Q(b)$ is a continuously differentiable function of $b \in B$, where $B$ is compact, and $\nabla_b \Q(b)$ has a bounded derivative. This allows us to write
\[
\nabla \widehat{M}(\bnMD) = \nabla \widehat{M}(b_0)+o_p(1) = -\nabla \Q(b_0) + o_p(1).
\]
We explicitly compute the inverse of $\nabla \Q(b)' \nabla \Q(b)$ below, which exists.

If $\Delta > 1$  we can write
\[
-\nabla\widehat{M}\left(b\right)=\left(\begin{array}{cc}
\frac{\partial \Q_{L}}{\partial b_{L}} & \frac{\partial \Q_{L}}{\partial b_{T}}\\
\frac{\partial \Q_{T}}{\partial b_{L}} & \frac{\partial \Q_{T}}{\partial b_{T}}
\end{array}\right)=\left(\begin{array}{cc}
1+o\left(1\right) & o\left(1\right)\\
o\left(1\right) & 1+o\left(1\right)
\end{array}\right)
\]
and if $\Delta = 1$ we can write
\[
-\nabla\widehat{M}\left(b\right)=\left(\begin{array}{cc}
\frac{\partial \Q_{L}}{\partial b_{L}} & \frac{\partial \Q_{L}}{\partial b_{T}}\\
\frac{\partial \Q_{T}}{\partial b_{L}} & \frac{\partial \Q_{T}}{\partial b_{T}}
\end{array}\right)=\left(\begin{array}{cc}
1+o\left(1\right) & 1+ o\left(1\right)\\
o\left(1\right) & 1+o\left(1\right)
\end{array}\right).
\]
We can also compute
\begin{align*}
\nabla\widehat{M}\left(b\right)'\nabla\widehat{M}\left(b\right) & =\left(\begin{array}{cc}
\frac{\partial \Q_{L}}{\partial b_{L}} & \frac{\partial \Q_{T}}{\partial b_{L}}\\
\frac{\partial \Q_{L}}{\partial b_{T}} & \frac{\partial \Q_{T}}{\partial b_{T}}
\end{array}\right)\left(\begin{array}{cc}
\frac{\partial \Q_{L}}{\partial b_{L}} & \frac{\partial \Q_{L}}{\partial b_{T}}\\
\frac{\partial \Q_{T}}{\partial b_{L}} & \frac{\partial \Q_{T}}{\partial b_{T}}
\end{array}\right)\\
 & =\left(\begin{array}{cc}
\left(\frac{\partial \Q_{L}}{\partial b_{L}}\right)^{2}+\left(\frac{\partial \Q_{T}}{\partial b_{L}}\right)^{2} & \frac{\partial \Q_{L}}{\partial b_{L}}\frac{\partial \Q_{L}}{\partial b_{T}}+\frac{\partial \Q_{L}}{\partial b_{T}}\frac{\partial \Q_{T}}{\partial b_{T}}\\
\frac{\partial \Q_{L}}{\partial b_{L}}\frac{\partial \Q_{L}}{\partial b_{T}}+\frac{\partial \Q_{T}}{\partial b_{L}}\frac{\partial \Q_{T}}{\partial b_{T}} & \left(\frac{\partial \Q_{L}}{\partial b_{T}}\right)^{2}+\left(\frac{\partial \Q_{T}}{\partial b_{T}}\right)^{2}
\end{array}\right)
\end{align*}
and so the inverse is
\[
\left[\nabla\widehat{M}\left(b\right)'\nabla\widehat{M}\left(b\right)\right]^{-1}=\frac{1}{{\rm det}\left[\nabla\widehat{M}\left(b\right)'\nabla\widehat{M}\left(b\right)\right]}\left(\begin{array}{cc}
\left(\frac{\partial \Q_{L}}{\partial b_{T}}\right)^{2}+\left(\frac{\partial \Q_{T}}{\partial b_{T}}\right)^{2} & -\left[\frac{\partial \Q_{L}}{\partial b_{L}}\frac{\partial \Q_{L}}{\partial b_{T}}+\frac{\partial \Q_{L}}{\partial b_{T}}\frac{\partial \Q_{T}}{\partial b_{T}}\right]\\
-\left[\frac{\partial \Q_{L}}{\partial b_{L}}\frac{\partial \Q_{L}}{\partial b_{T}}+\frac{\partial \Q_{T}}{\partial b_{L}}\frac{\partial \Q_{T}}{\partial b_{T}}\right] & \left(\frac{\partial \Q_{L}}{\partial b_{L}}\right)^{2}+\left(\frac{\partial \Q_{T}}{\partial b_{L}}\right)^{2}
\end{array}\right).
\]

The determinant is given by
\begin{align*}
{\rm det}\left[\nabla\widehat{M}\left(b\right)'\nabla\widehat{M}\left(b\right)\right] & =\left[\left(\frac{\partial \Q_{L}}{\partial b_{T}}\right)^{2}+\left(\frac{\partial \Q_{T}}{\partial b_{T}}\right)^{2}\right]\left[\left(\frac{\partial \Q_{L}}{\partial b_{L}}\right)^{2}+\left(\frac{\partial \Q_{T}}{\partial b_{L}}\right)^{2}\right]\\
 & -\left[\frac{\partial \Q_{L}}{\partial b_{L}}\frac{\partial \Q_{L}}{\partial b_{T}}+\frac{\partial \Q_{T}}{\partial b_{L}}\frac{\partial \Q_{T}}{\partial b_{T}}\right]^{2}\\
 & =\left(\frac{\partial \Q_{L}}{\partial b_{L}}\frac{\partial \Q_{T}}{\partial b_{T}}-\frac{\partial \Q_{T}}{\partial b_{L}}\frac{\partial \Q_{L}}{\partial b_{T}}\right)^{2}.
\end{align*}

If $\Delta > 1$ then the determinant is $1+o\left(1\right).$ If $\Delta =1$ it is the same.

So the inverse is, if $\Delta > 1,$
\[
\left[\nabla\widehat{M}\left(b\right)'\nabla\widehat{M}\left(b\right)\right]^{-1}=\frac{1}{1+o\left(1\right)}\left(\begin{array}{cc}
1+o\left(1\right) & o\left(1\right)\\
o\left(1\right) & 1+o\left(1\right)
\end{array}\right).
\]

We can compute the final object in the case $\Delta > 1$ as
\begin{align*}
-\left[\nabla\widehat{M}\left(b_0\right)'\nabla\widehat{M}\left(b_0\right)\right]^{-1}\nabla\widehat{M}\left(b_0\right)' & =\frac{1}{1+o\left(1\right)}\left(\begin{array}{cc}
1+o\left(1\right) & o\left(1\right)\\
o\left(1\right) & 1+o\left(1\right)
\end{array}\right)\left(\begin{array}{cc}
1+o\left(1\right) & o\left(1\right)\\
o\left(1\right) & 1+o\left(1\right)
\end{array}\right)\\
 & =\left(\begin{array}{cc}
1+o\left(1\right) & o\left(1\right)\\
o\left(1\right) & 1+o\left(1\right)
\end{array}\right)  \cvgto I,
\end{align*}
which completes the argument for the case of $\Delta>1$.

Meanwhile  if $\Delta = 1$ then the inverse is
\[
\left[\nabla\widehat{M}\left(b\right)'\nabla\widehat{M}\left(b\right)\right]^{-1}=\frac{1}{1+o\left(1\right)}\left(\begin{array}{cc}
2+o\left(1\right) & -1+o\left(1\right)\\
-1+o\left(1\right) & 1+o\left(1\right)
\end{array}\right).
\]
Therefore,
\begin{align*}
-\left[\nabla\widehat{M}\left(b_0\right)'\nabla\widehat{M}\left(b_0\right)\right]^{-1}\nabla\widehat{M}\left(b_0\right)' & =\frac{1}{1+o\left(1\right)}\left(\begin{array}{cc}
2+o\left(1\right) & -1+o\left(1\right)\\
-1+o\left(1\right) & 1+o\left(1\right)
\end{array}\right)\left(\begin{array}{cc}
1+o\left(1\right) & o\left(1\right)\\
1+o\left(1\right) & 1+o\left(1\right)
\end{array}\right)\\
 & =\left(\begin{array}{cc}
1+o\left(1\right) & -1+o\left(1\right)\\
o\left(1\right) & 1+o\left(1\right)
\end{array}\right)  \cvgto \left(\begin{array}{cc}
1 &  {-1}\\
0 & 1
\end{array}\right)
\end{align*}

Now consider the case with $\Delta = 1$. In this case since $n^{2-h_L} = n^{3-h_T}$, it follows from our calculations above that
\begin{align*}
\sqrt{n^{2-h_{L}}}R_{n}\left(\betanMD-\beta_{0}\right) & =\sqrt{n^{2-h_{L}}}\left(\begin{array}{cc}
1 &  {-1}\\
0 & 1
\end{array}\right)\left(\begin{array}{c}
\frac{n^{h_{L}}}{\binom{n}{2}}\sum g_{ij}-\Q_{L}\left(b_{0}\right)\\
\frac{n^{h_{T}}}{\binom{n}{3}}\sum g_{ij}g_{ik}g_{jk}-\Q_{T}\left(b_{0}\right)
\end{array}\right)\\
 & =\sqrt{n^{2-h_{L}}}\left(\begin{array}{c}
\{\frac{n^{h_{L}}}{\binom{n}{2}}\sum g_{ij}-\Q_{L}\left(b_{0}\right)\}- \{\frac{n^{h_{T}}}{\binom{n}{3}}\sum g_{ij}g_{ik}g_{jk}-\Q_{T}\left(b_{0}\right)\}\\
\frac{n^{h_{T}}}{\binom{n}{3}}\sum g_{ij}g_{ik}g_{jk}-\Q_{T}\left(b_{0}\right).
\end{array}\right),
\end{align*}
which still jointly converge to a mean zero random variable, but with a different variance-covariance matrix:
\begin{align*}
V_{n}=\left(\begin{array}{cc}
\var(n^{h_L} S_L -  n^{h_T}S_T) & \cov(n^{h_L} S_L - n^{h_T}S_T,n^{h_T} S_T)\\
\cov(n^{h_L} S_L -  n^{h_T}S_T,n^{h_T} S_T) & \var(n^{h_T} S_T)
\end{array}\right)
\end{align*}
for the $\Delta=1$ case.
\end{proof}

\clearpage
\newpage

\setcounter{table}{0}
\renewcommand{\thetable}{D.\arabic{table}}
\setcounter{figure}{0}
\renewcommand{\thefigure}{D.\arabic{figure}}

\section{Simulations}\label{sec:sim_consistency}

We demonstrate that $\betanMD$ is consistent and  asymptotically normal distributed and also investigate the quality of estimated standard errors. We show that the minimum distance estimator performs well in the links and triangles model---estimates are centered on the true value of the parameters, the distribution indeed looks asymptotically normal, and the constructed 95\% confidence interval has proper coverage.

We set $h_L = 1$, $h_T = h_L+1$, and $b_{0,L} = b_{0,T}$, which we vary in order to generate the requisite average degrees described below.

\subsection{Consistency}

First we present the point estimates of our minimum distance and direct estimators over a wide range of average degrees. Consistent with the theory, the estimators coincide at low densities and as we look at increasingly dense graphs, the direct estimator misattributes direct link formation to triangles and inherits a bias, whereas the minimum distance estimator does not. Figure \ref{fig:consistency} presents the results.

\begin{figure}[!h]
\centering
		\subfloat[$n = 200, \ \beta_L$]{
	\includegraphics[trim = 0in 0in 0in 0.2in, clip = true, scale=0.35]{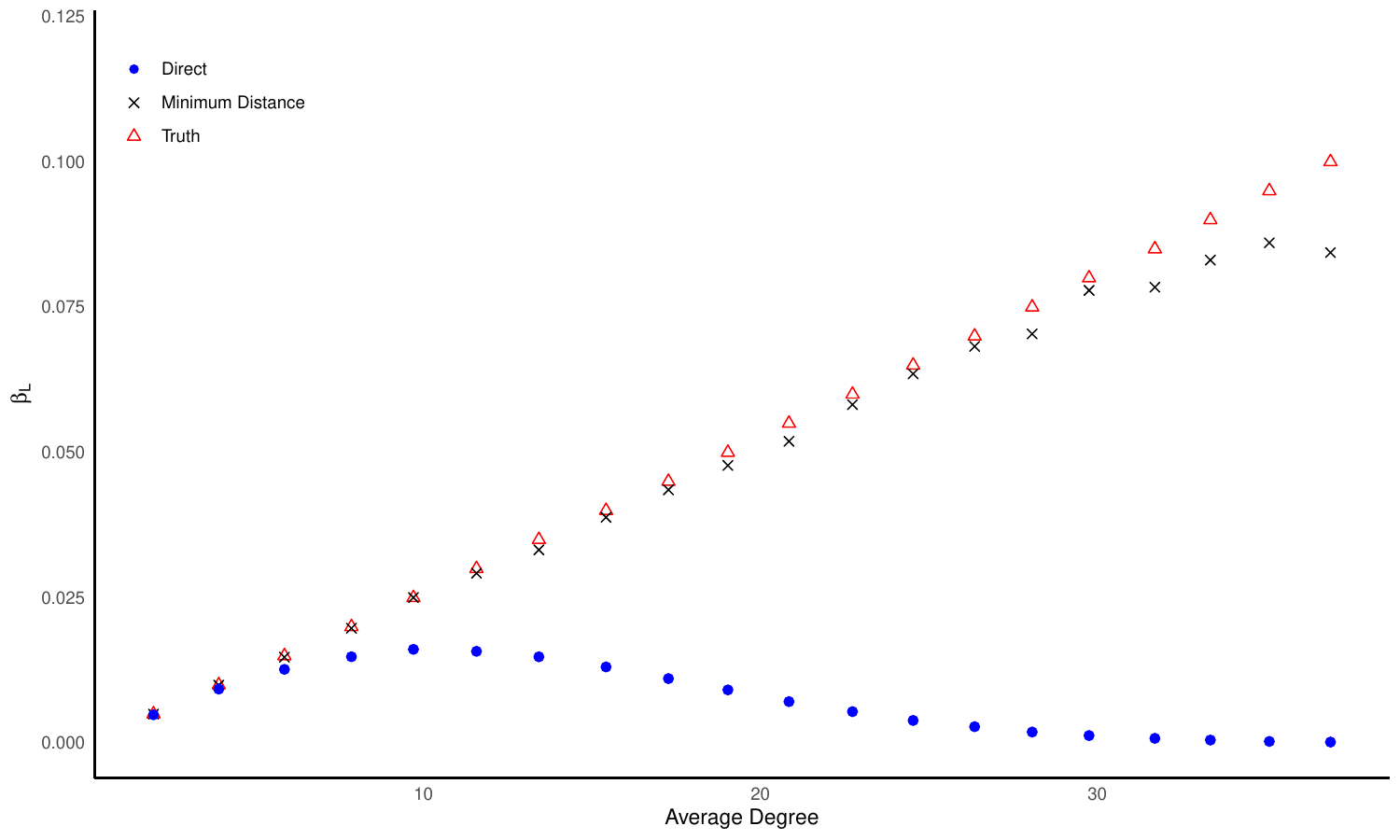}
}
	\subfloat[$n = 200, \ \beta_T$]{
	\includegraphics[trim = 0in 0in 0in 0.2in, clip = true, scale=0.35]{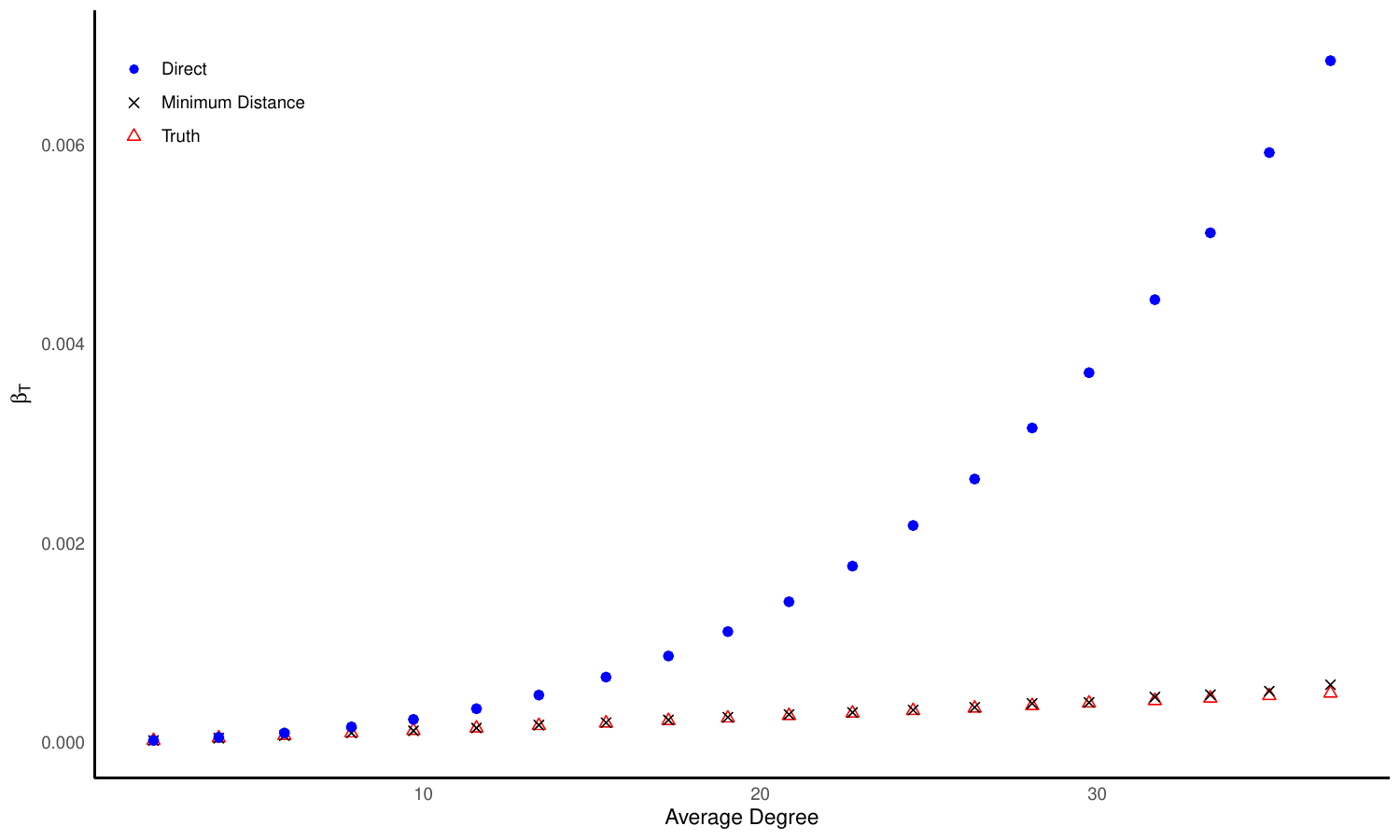}
}

		\subfloat[$n = 500, \ \beta_L$]{
	\includegraphics[trim = 0in 0in 0in 0.2in, clip = true, scale=0.35]{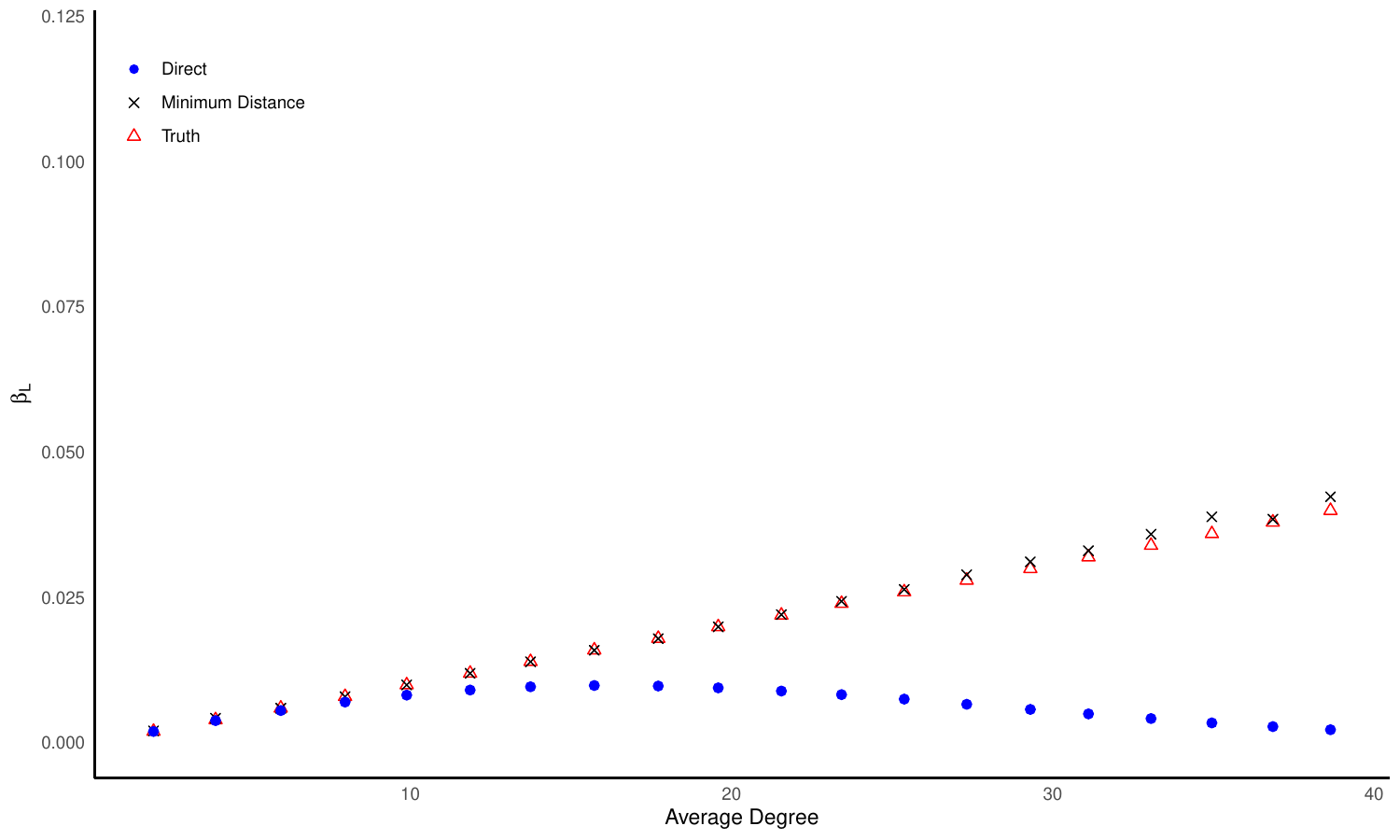}
}
	\subfloat[$n = 500, \ \beta_T$]{
	\includegraphics[trim = 0in 0in 0in 0.2in, clip = true, scale=0.35]{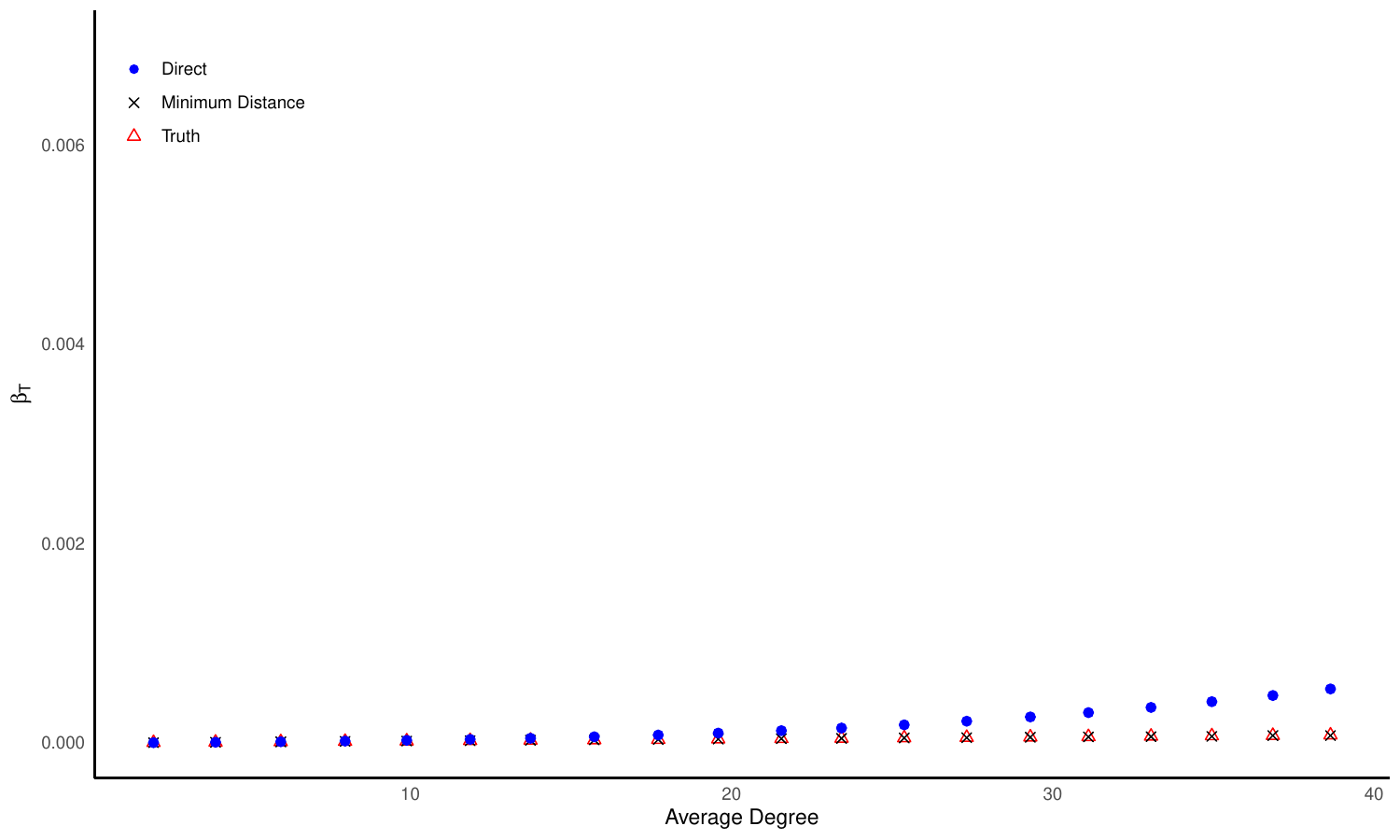}
}

		\subfloat[$n = 1000, \ \beta_L$]{
	\includegraphics[trim = 0in 0in 0in 0.2in, clip = true, scale=0.35]{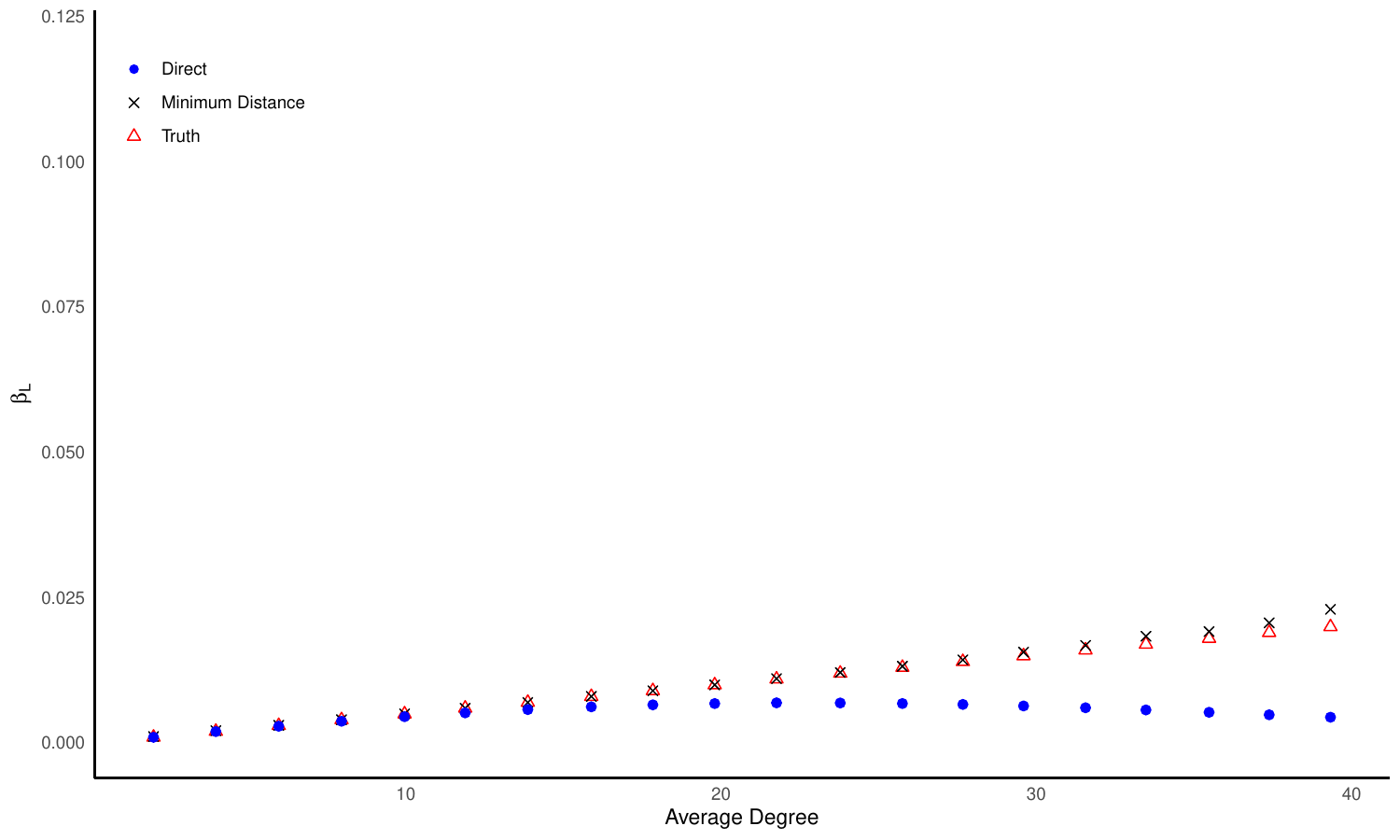}
}
	\subfloat[$n = 1000, \ \beta_T$]{
	\includegraphics[trim = 0in 0in 0in 0.2in, clip = true, scale=0.35]{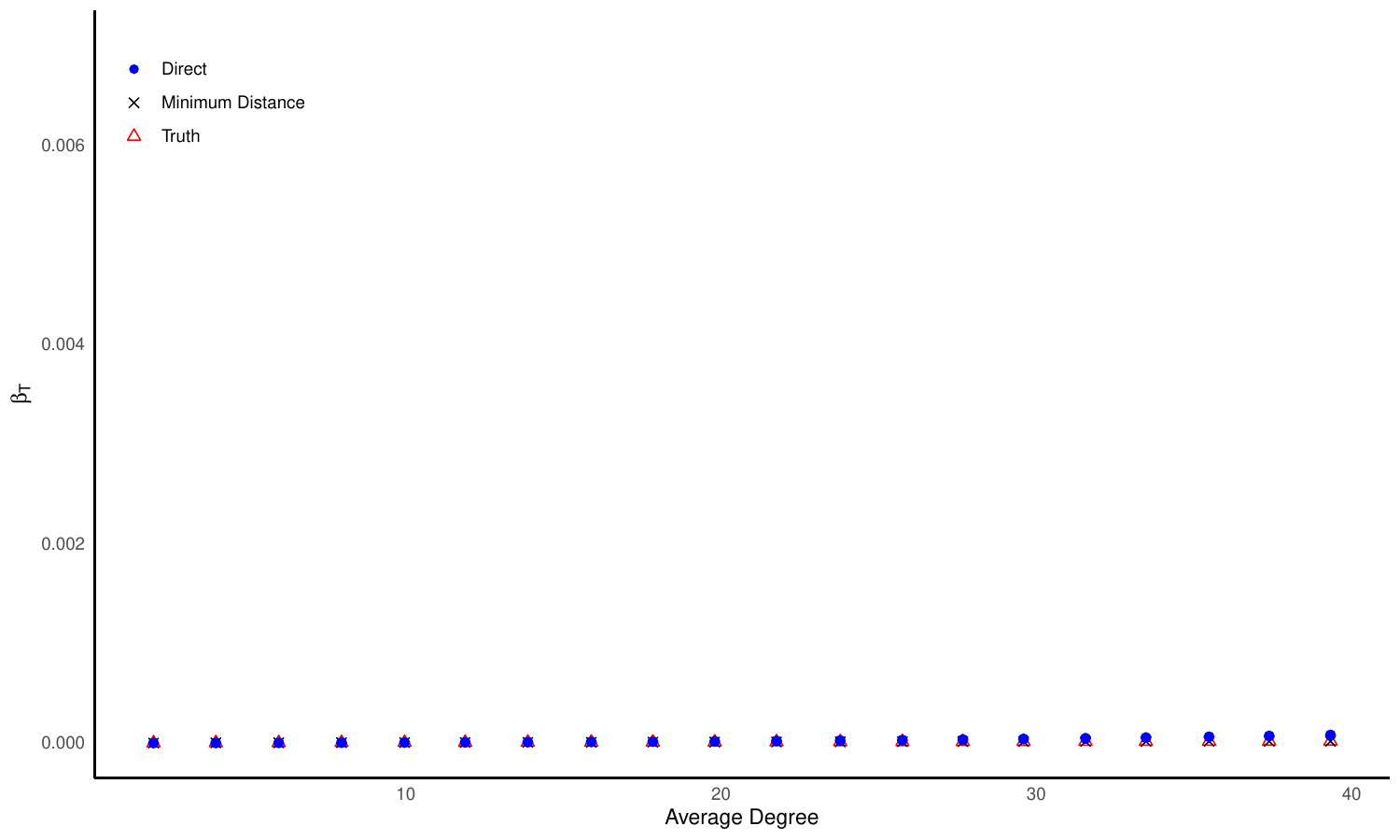}
}
	\caption{Estimates from the mimimum distance $\betanMD$ and direct-count $\betanDC$ estimators for $n \in \{200,500,1000\}$ with 100 simulations used.}
\label{fig:consistency}
\end{figure}

\subsection{Asymptotic Normality}

Next, we turn to the asymptotic distribution of the estimator and show that it is normal. We display the (standardized) parameter estimates $\betanMD$ and a standard normal distribution for lower and higher degrees in the range consistent with data and across a broad range of network sizes. The results, displayed in Figures \ref{fig:asymptotic_normality-6} and \ref{fig:asymptotic_normality-14} clearly show that the normal approximation is good.

\begin{figure}[!h]
		\subfloat[$n = 200, \ \beta_L$]{
	\includegraphics[trim = 0in 0in 0in 0.27in, clip = true, scale = 0.3]{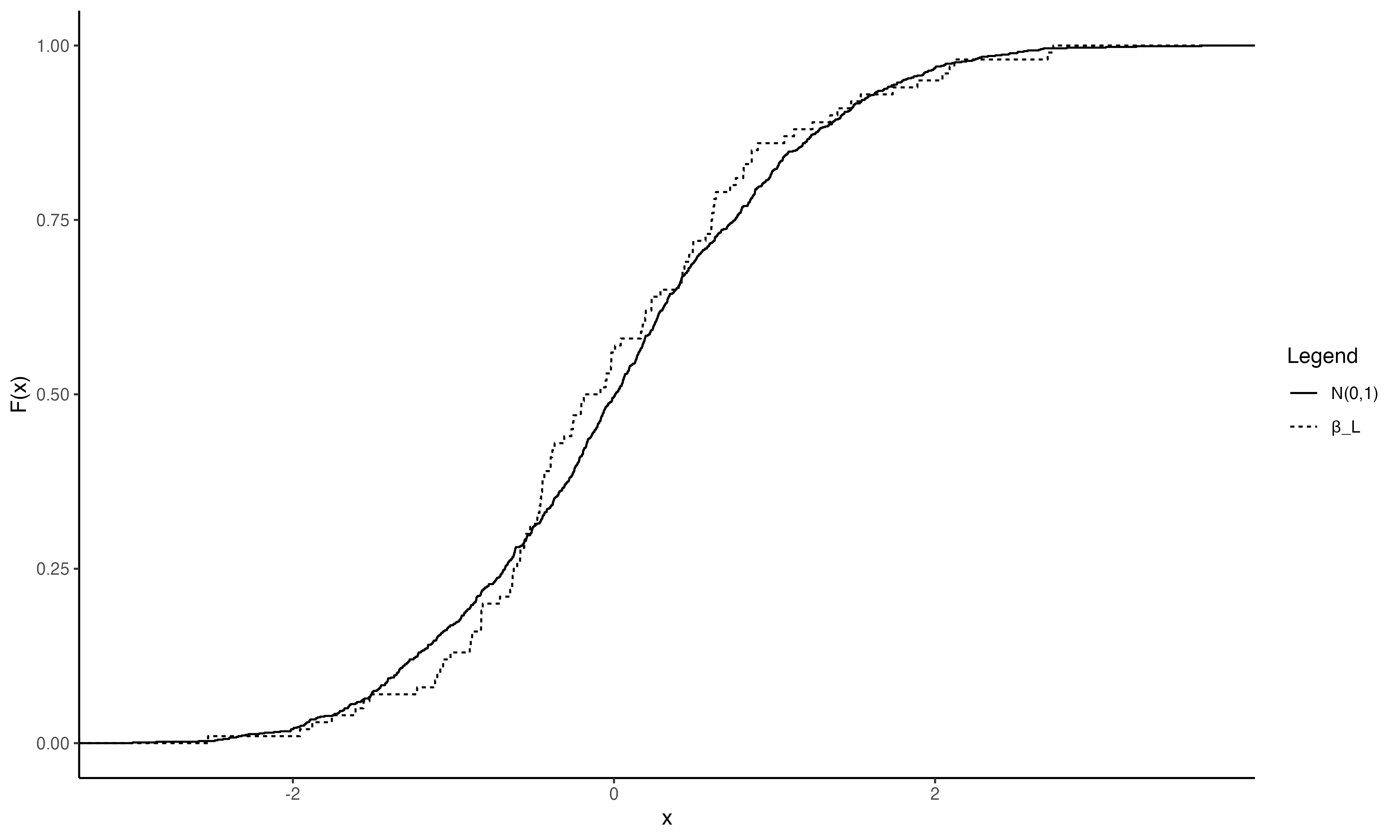}
	}
	\subfloat[$n = 200, \ \beta_T$]{
	\includegraphics[trim = 0in 0in 0in 0.27in, clip = true, scale = 0.3]{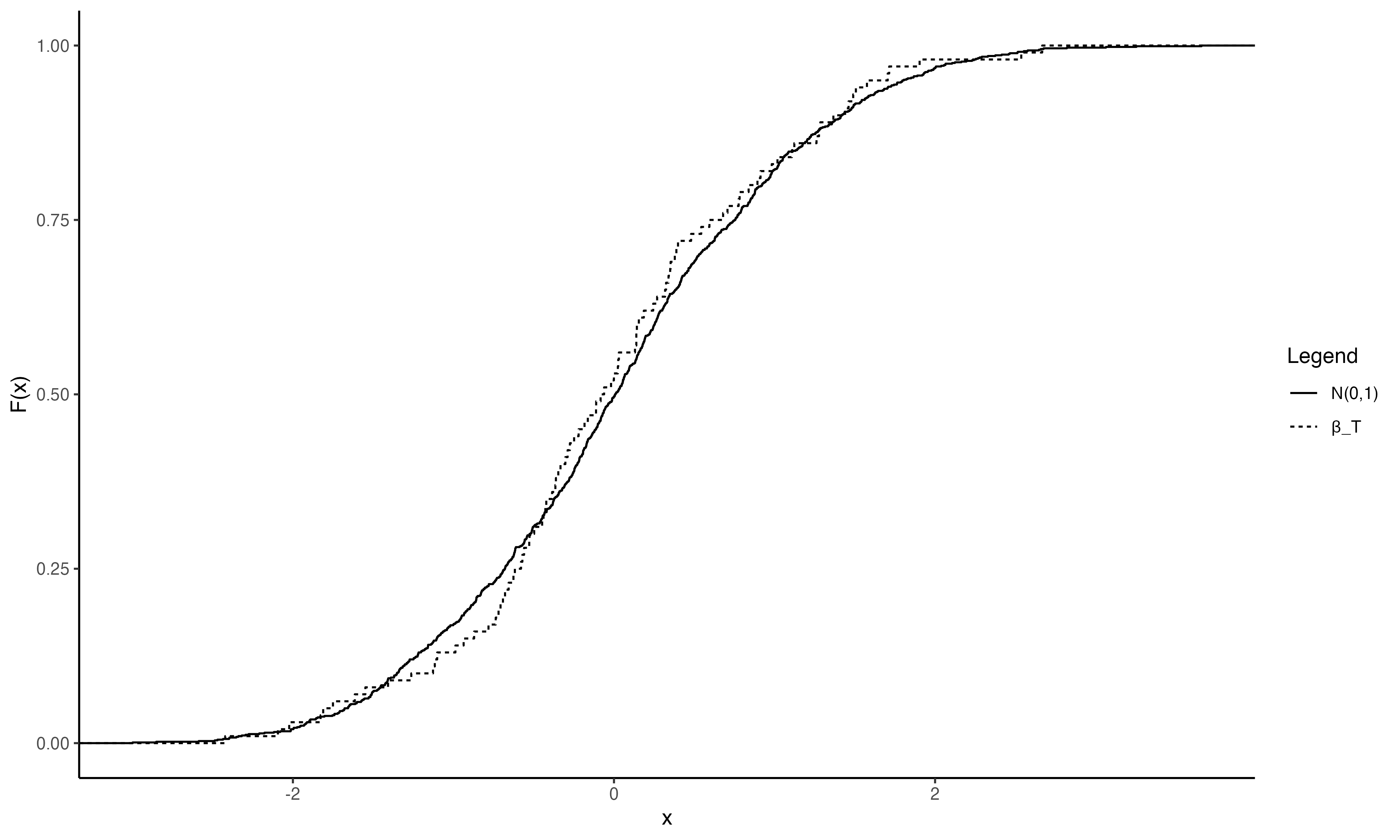}
	}

	\subfloat[$n = 500, \ \beta_L$]{
	\includegraphics[trim = 0in 0in 0in 0.27in, clip = true, scale = 0.3]{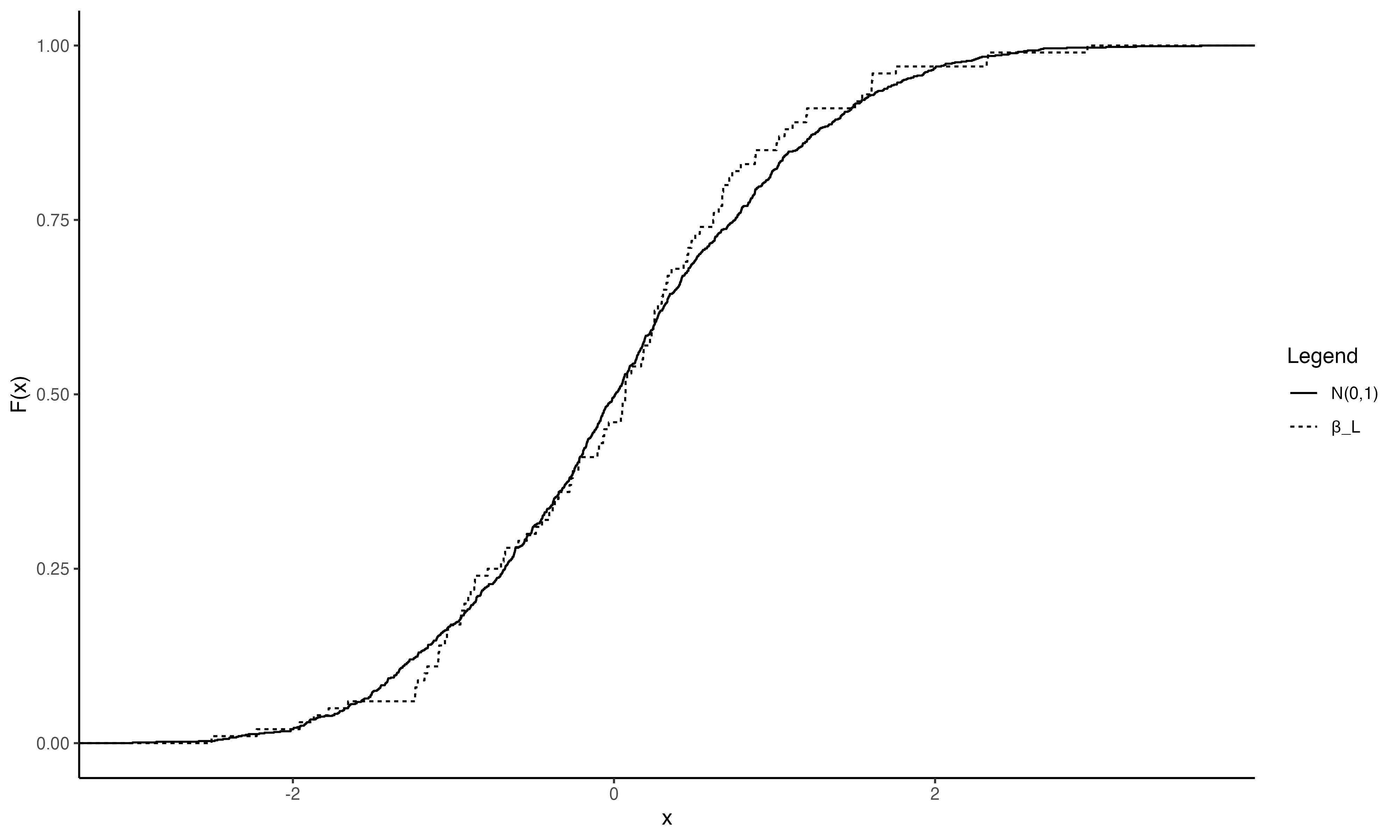}
	}
	\subfloat[$n = 500, \ \beta_T$]{
	\includegraphics[trim = 0in 0in 0in 0.27in, clip = true, scale = 0.3]{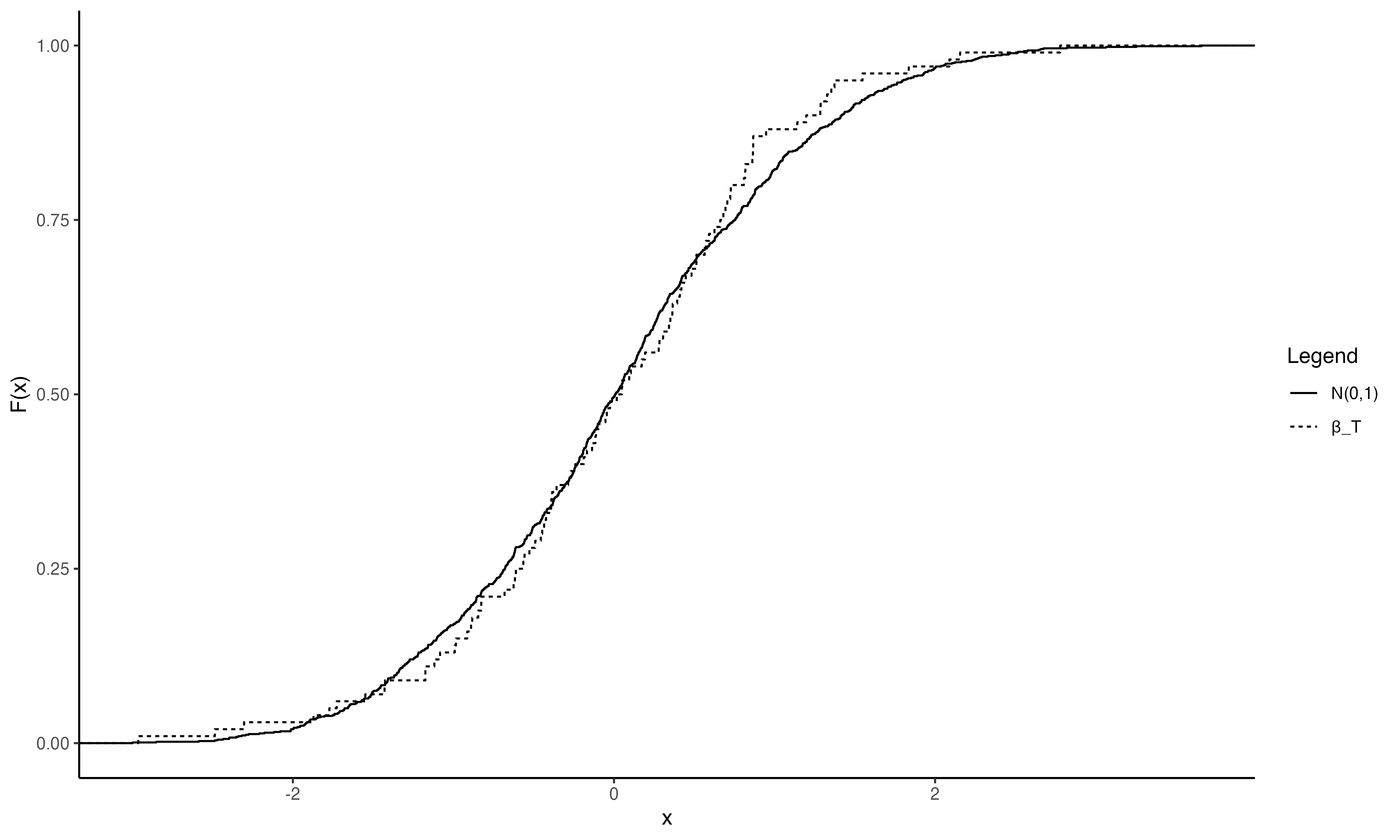}
}

	\subfloat[$n = 1000, \ \beta_L$]{
	\includegraphics[trim = 0in 0in 0in 0.27in, clip = true, scale = 0.3]{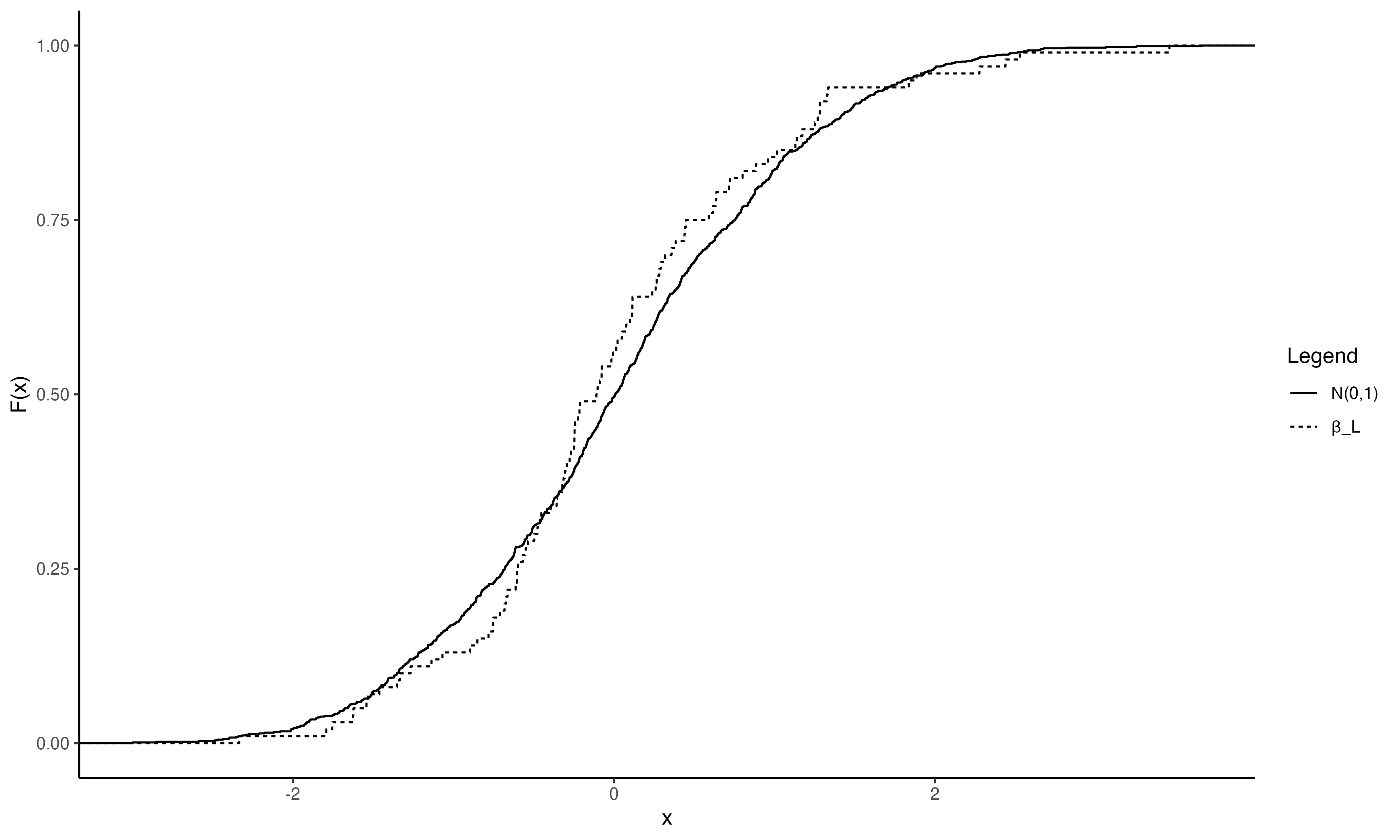}
}
\subfloat[$n = 1000,  \ \beta_T$]{
	\includegraphics[trim = 0in 0in 0in 0.27in, clip = true, scale = 0.3]{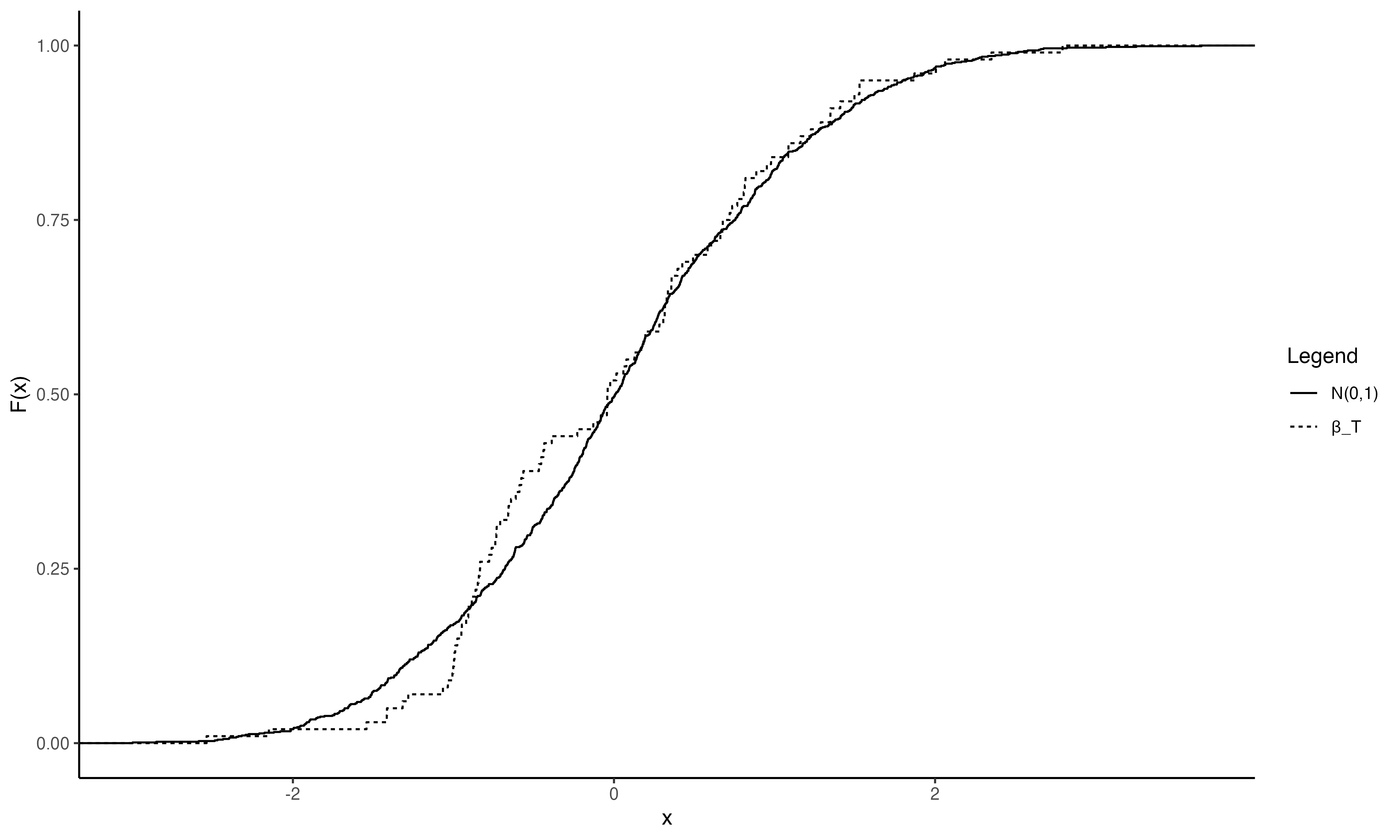}
}

	\caption{CDFs of standardized parameter estimates $\betanMD$ for average degrees of 5.9 with $n \in \{200, 500, 1000\}$ with 100 simulations.}
\label{fig:asymptotic_normality-6}
\end{figure}

\begin{figure}[!h]
		\subfloat[$n = 200, \ \beta_L$]{
	\includegraphics[trim = 0in 0in 0in 0.27in, clip = true, scale = 0.3]{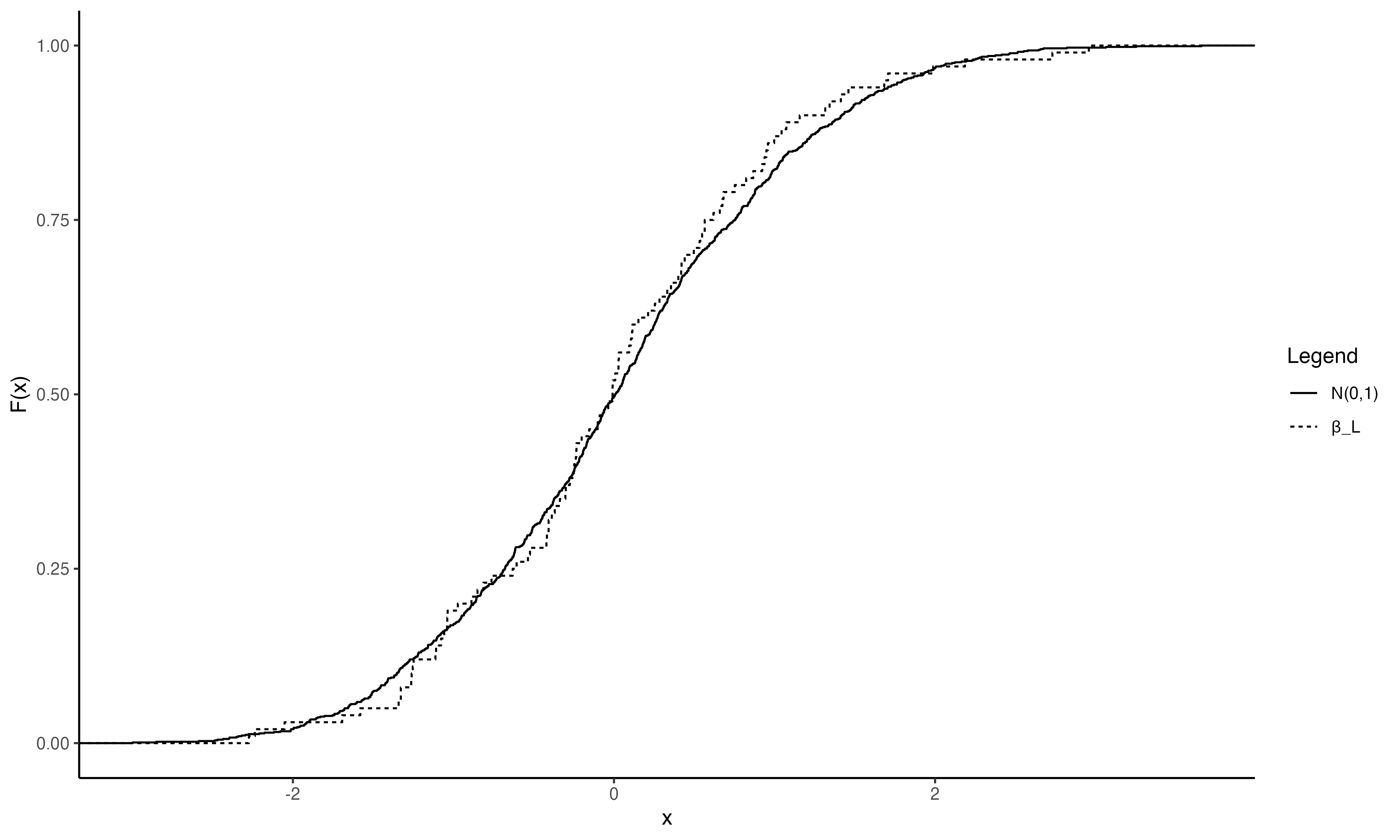}
	}
	\subfloat[$n = 200, \ \beta_T$]{
	\includegraphics[trim = 0in 0in 0in 0.27in, clip = true, scale = 0.3]{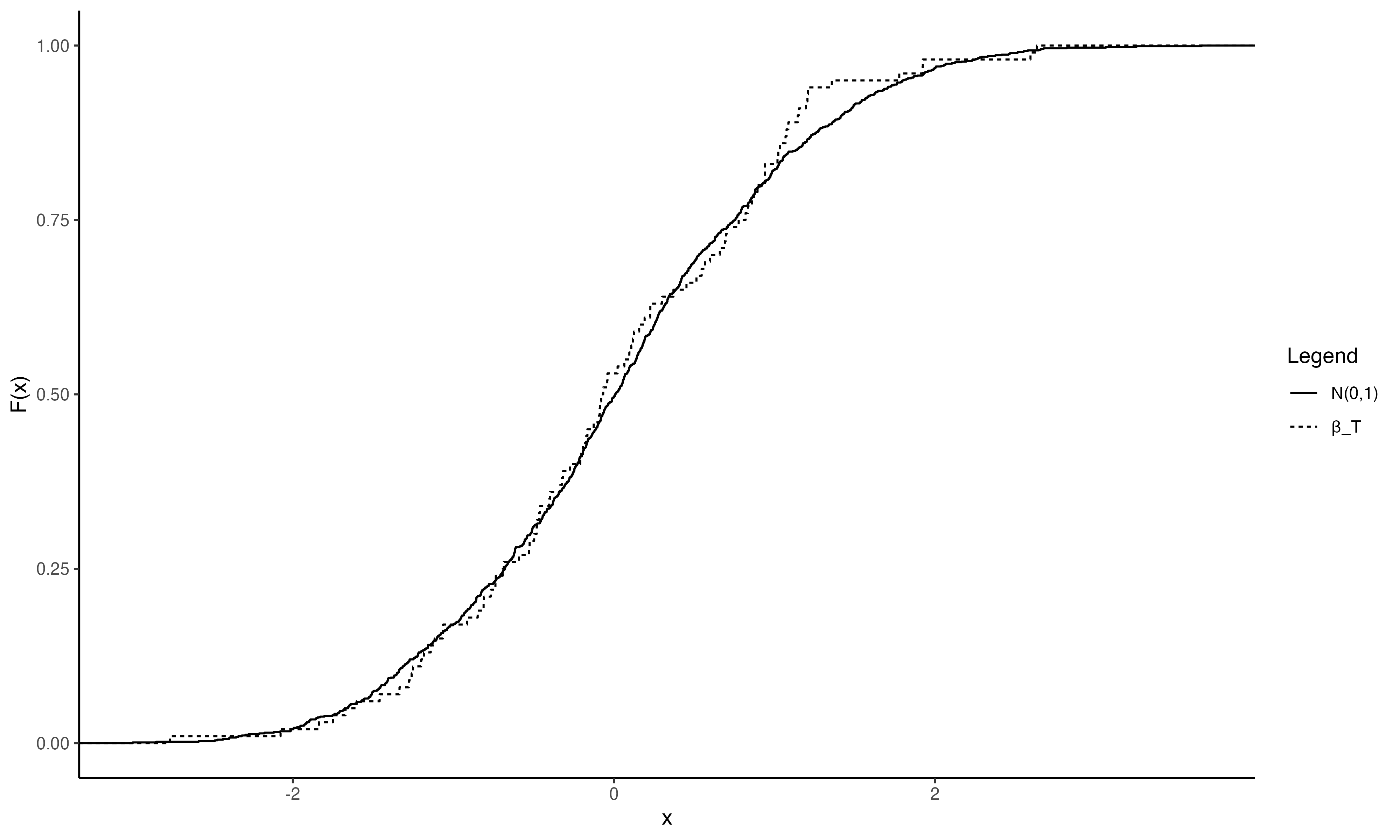}
	}

	\subfloat[$n = 500, \ \beta_L$]{
	\includegraphics[trim = 0in 0in 0in 0.27in, clip = true, scale = 0.3]{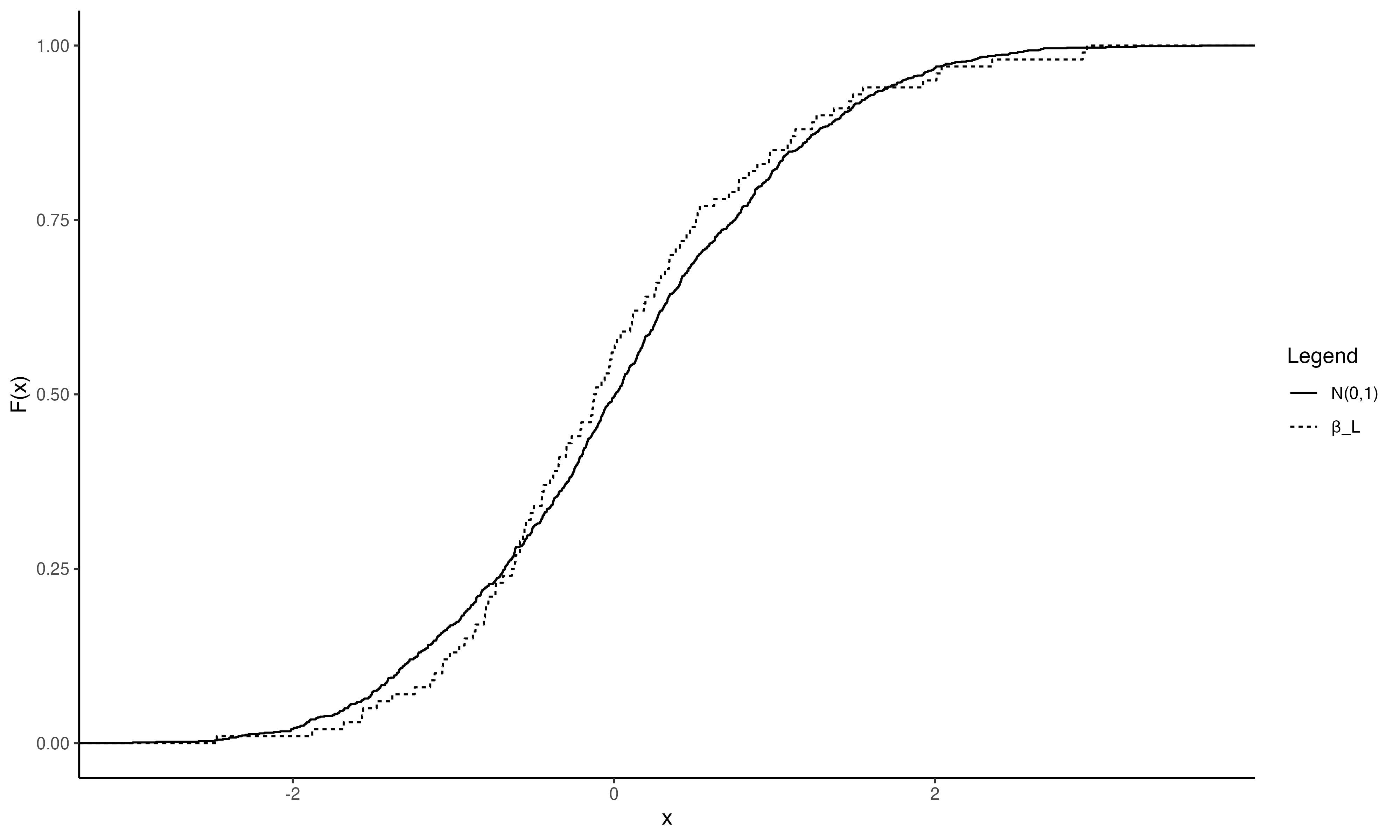}
	}
	\subfloat[$n = 500, \ \beta_T$]{
	\includegraphics[trim = 0in 0in 0in 0.27in, clip = true, scale = 0.3]{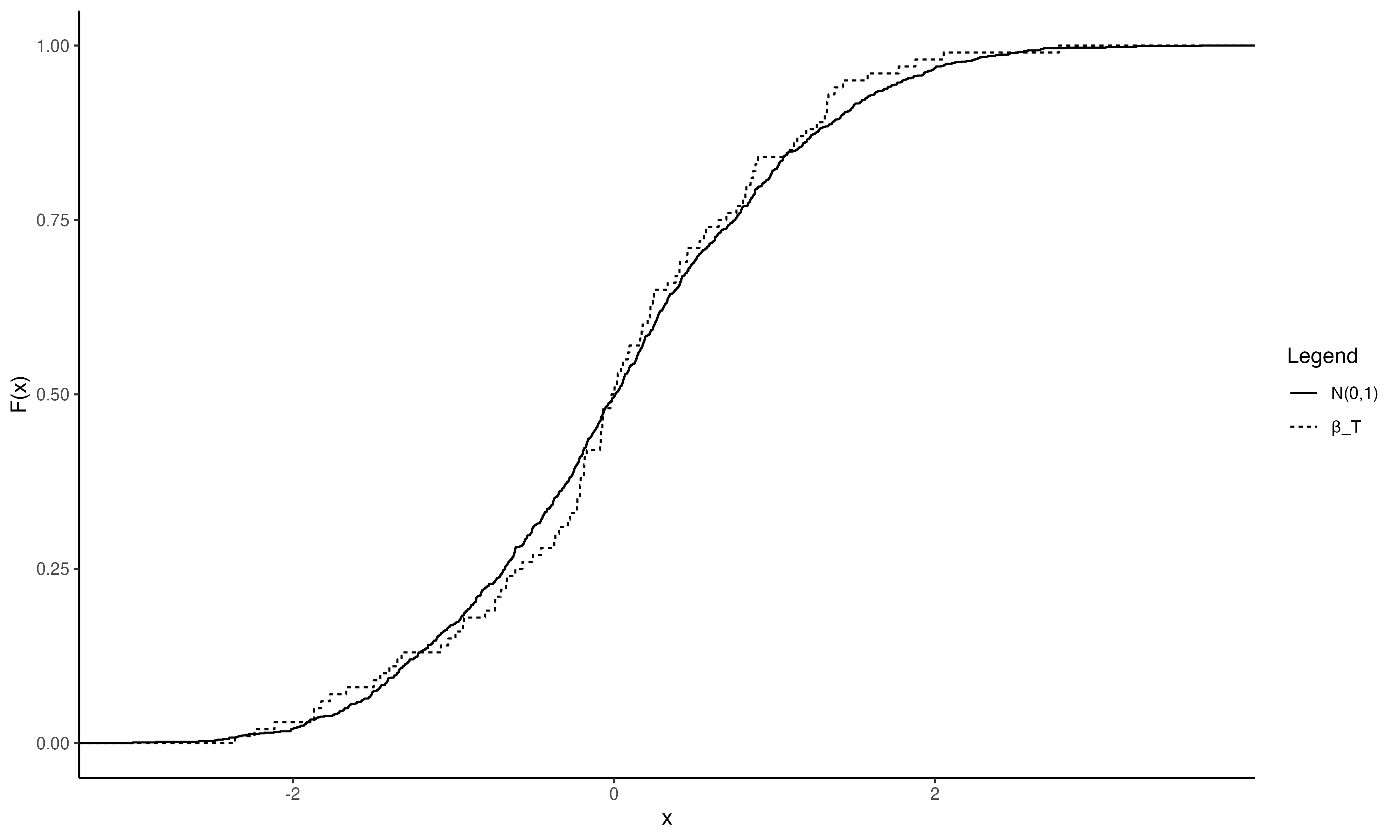}
}

	\subfloat[$n = 1000, \ \beta_L$]{
	\includegraphics[trim = 0in 0in 0in 0.27in, clip = true, scale = 0.3]{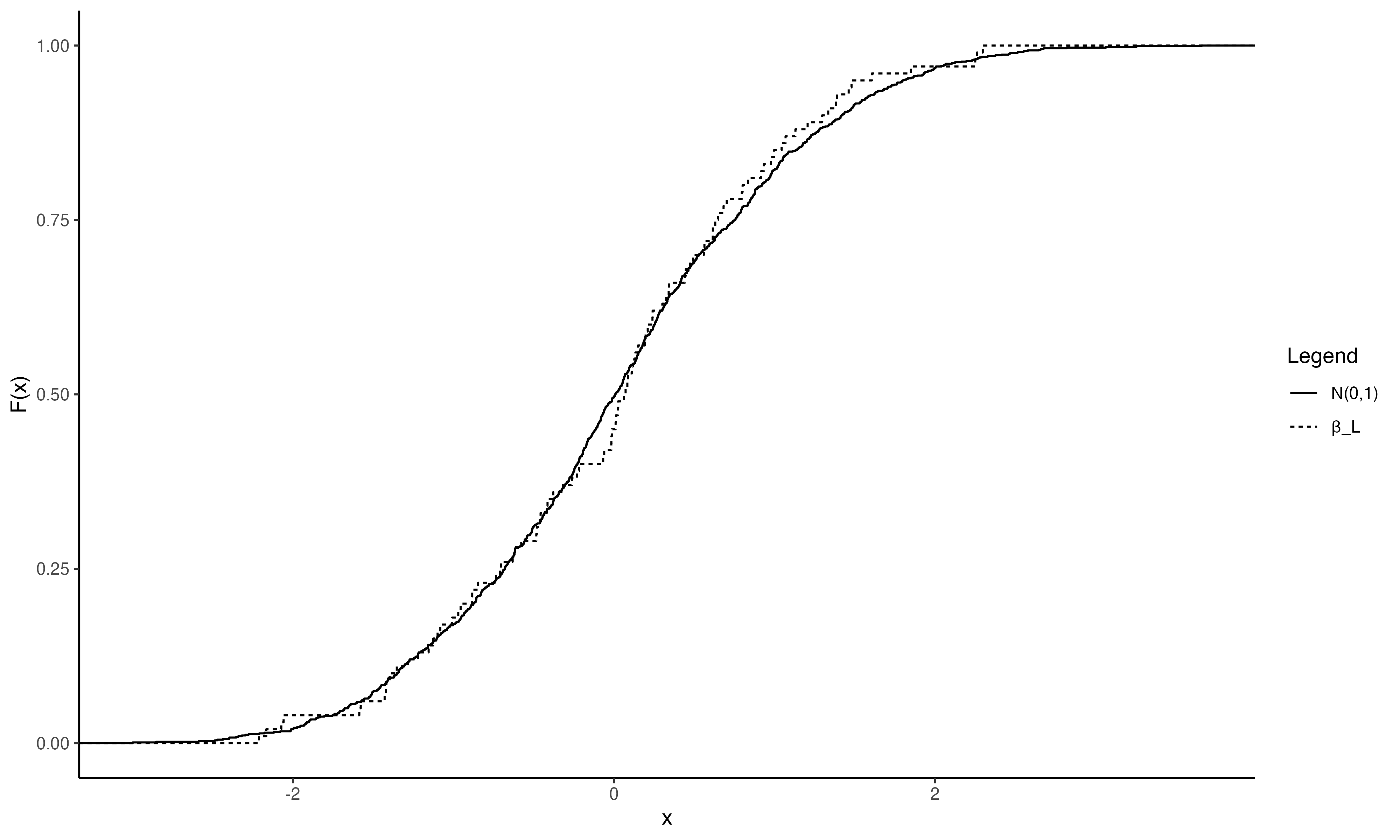}
}
\subfloat[$n = 1000,  \ \beta_T$]{
	\includegraphics[trim = 0in 0in 0in 0.27in, clip = true, scale = 0.3]{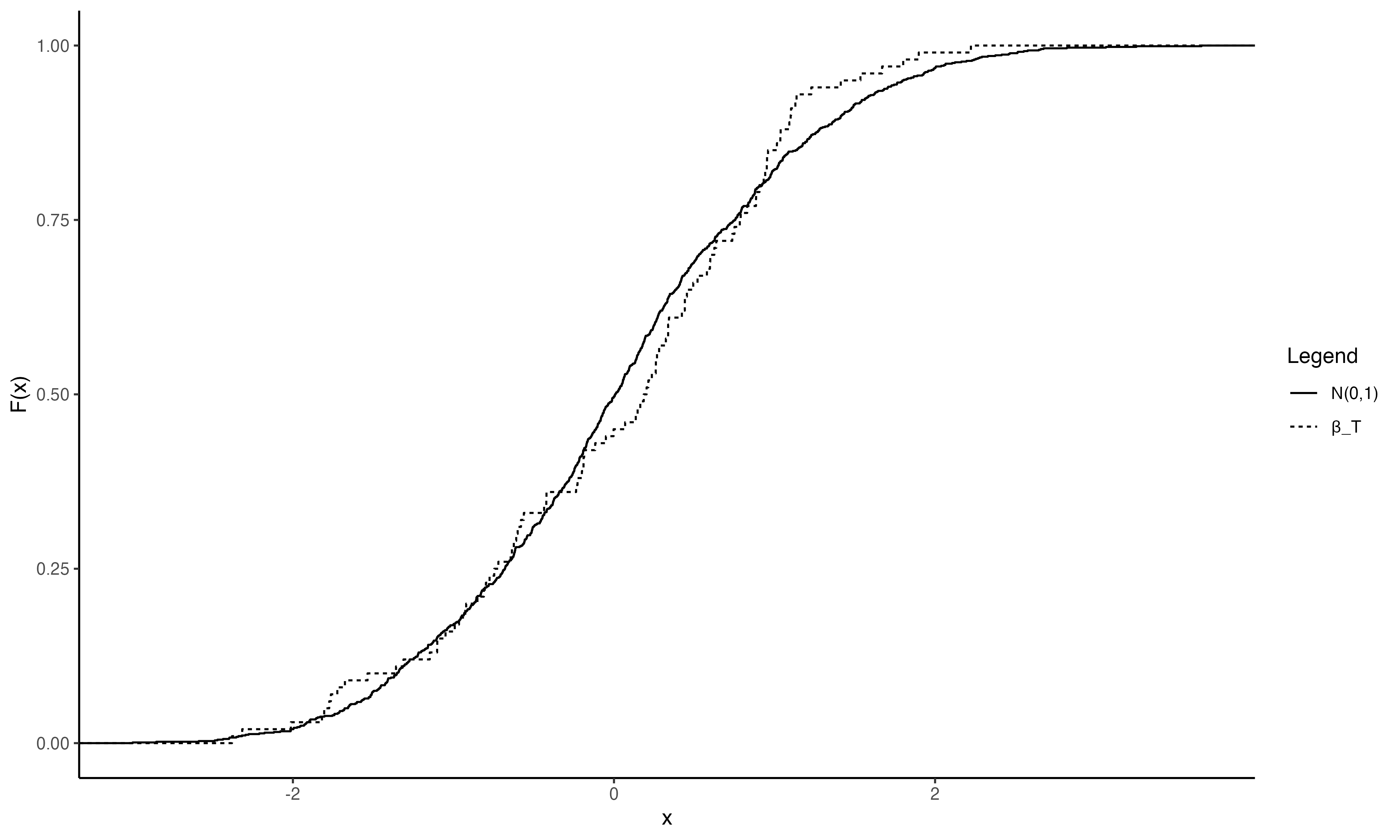}
}

	\caption{CDFs of standardized parameter estimates $\betanMD$ for average degrees of 13.5 with $n \in \{200, 500, 1000\}$ with 100 simulations.}
\label{fig:asymptotic_normality-14}
\end{figure}

\clearpage
\subsection{Coverage}

Having seen that both the parameter estimates are accurate and appear to be asymptotically normally distributed, consistent with theory, finally we investigate the analytic, estimated standard errors. We construct these using the minimum distance estimator of Proposition \ref{prop:LT_SUGM1-2}, $\betanMD$. We omit for brevity but note that by simply simulating the model from the estimated parameters and then calculating standard errors through this process has excellent coverage properties, but developing the theory of a bootstrap is beyond the scope of the present paper and remains a topic for future work.

\begin{table}[!h]
\centering
\caption{Coverage for $\beta_L$}
\centering
\fontsize{10}{12}\selectfont
\begin{threeparttable}
\begin{tabular}[t]{rrrrrrrr}
\toprule
Avg. Degree & 100 & 150 & 200 & 250 & 300 & 350 & 400\\
\midrule
7.914 & 0.945 & 0.945 & 0.980 & 0.975 & 0.985 & 1.000 & 0.995\\
9.890 & 0.935 & 0.955 & 0.955 & 0.970 & 0.960 & 0.965 & 1.000\\
11.804 & 0.870 & 0.935 & 0.900 & 0.960 & 0.970 & 0.955 & 0.960\\
13.779 & 0.885 & 0.920 & 0.930 & 0.945 & 0.935 & 0.925 & 0.980\\
15.647 & 0.830 & 0.875 & 0.880 & 0.930 & 0.915 & 0.945 & 0.945\\
\bottomrule
\end{tabular}
\begin{tablenotes}[para]
\item \textit{Note: }
\item Note: Coverage of the 95\% confidence interval for the link probability parameter
           implied by Proposition 3, for various average degrees for various network sizes.
           For each network size we and each degree we conduct 200 simulations.
\end{tablenotes}
\end{threeparttable}
\end{table}

\begin{table}[!h]
\centering
\caption{Coverage for $\beta_T$}
\centering
\fontsize{10}{12}\selectfont
\begin{threeparttable}
\begin{tabular}[t]{rrrrrrrr}
\toprule
Avg. Degree & 100 & 150 & 200 & 250 & 300 & 350 & 400\\
\midrule
7.914 & 0.930 & 0.930 & 0.930 & 0.930 & 0.940 & 0.960 & 0.945\\
9.890 & 0.925 & 0.890 & 0.950 & 0.945 & 0.930 & 0.920 & 0.930\\
11.804 & 0.930 & 0.940 & 0.930 & 0.970 & 0.965 & 0.925 & 0.950\\
13.779 & 0.905 & 0.915 & 0.945 & 0.960 & 0.930 & 0.950 & 0.940\\
15.647 & 0.910 & 0.940 & 0.955 & 0.945 & 0.920 & 0.950 & 0.955\\
\bottomrule
\end{tabular}
\begin{tablenotes}[para]
\item \textit{Note: }
\item Note: Coverage of the 95\% confidence interval for the triangle probability parameter
           implied by Proposition 3, for various average degrees for various network sizes.
           For each network size we and each degree we conduct 200 simulations.
\end{tablenotes}
\end{threeparttable}
\end{table}

To take stock at an aggregate level, in this range the overall average coverage is 0.941 and 0.936 for links and triangles respectively.

\clearpage

\end{document}